\documentclass[11pt,a4paper]{article}
\pdfoutput=1 %for arxiv
\usepackage[utf8]{inputenc}

\usepackage[T1]{fontenc}
\usepackage{lmodern}

\usepackage[margin=1in]{geometry}
\usepackage{amsfonts}
\usepackage{amsmath}
\usepackage{amssymb}
\usepackage{amsthm}
\usepackage{thmtools}
\usepackage{thm-restate}
\usepackage{float}
\usepackage[font=small,labelfont=bf]{caption} %customise text style on fig captions
\allowdisplaybreaks
\usepackage{bbold}
\usepackage{subcaption}

\usepackage{physics}
\usepackage{xspace}

\usepackage{multirow}

\usepackage{bold-extra} %adds bold+smallcaps (for standard fonts)

\usepackage{hyperref}
\usepackage[svgnames]{xcolor}
\hypersetup{colorlinks={true},urlcolor={blue},linkcolor={DarkBlue},citecolor=[named]{DarkGreen}}

\usepackage{microtype}
\usepackage[capitalise,nameinlink,noabbrev]{cleveref}

\usepackage{tikz}
\usetikzlibrary{arrows.meta}

\usepackage{doi}

\usepackage{graphicx}
\usepackage{mdframed}

\usepackage{natbib}

\usepackage{enumerate}

\theoremstyle{definition}
\newtheorem{definition}{Definition}

\theoremstyle{plain}
\newtheorem{theorem}{Theorem}[section]
\newtheorem{lemma}[theorem]{Lemma}

\newtheorem{proposition}[theorem]{Proposition}
\newtheorem{claim}{Claim}

\theoremstyle{remark}

\Crefname{claim}{Claim}{Claims}

\DeclareMathOperator{\opt}{OPT}

\def\fp/{\textup{\textsf{FP}}}
\def\p/{\textup{\textsf{P}}}
\def\np/{\textup{\textsf{NP}}}
\def\conp/{\textup{\textsf{co-NP}}}
\def\fnp/{\textup{\textsf{FNP}}}
\def\tfnp/{\textup{\textsf{TFNP}}}
\def\ptfnp/{\textup{\textsf{PTFNP}}}
\def\ppa/{\textup{\textsf{PPA}}}
\def\ppad/{\textup{\textsf{PPAD}}}
\def\ppads/{\textup{\textsf{PPADS}}}
\def\ppp/{\textup{\textsf{PPP}}}
\def\pwpp/{\textup{\textsf{PWPP}}}
\def\pls/{\textup{\textsf{PLS}}}
\def\cls/{\textup{\textsf{CLS}}}
\def\ppadpls/{\textup{$\textsf{PPAD} \cap \textsf{PLS}$}}
\def\ppapls/{\textup{$\textsf{PPA} \cap \textsf{PLS}$}}
\def\eopl/{\textup{\textsf{EOPL}}}
\def\sopl/{\textup{\textsf{SOPL}}}
\def\ueopl/{\textup{\textsf{UEOPL}}}
\def\fixp/{\textup{\textsf{FIXP}}}
\def\bu/{\textup{\textsf{BU}}}
\def\bbu/{\textup{\textsf{BBU}}}
\def\linearfixp/{\textup{\textsf{Linear-FIXP}}}
\def\pspace/{\textup{\textsf{PSPACE}}}

\providecommand{\ceil}[1]{\ensuremath{\left \lceil #1 \right \rceil }}
\providecommand{\floor}[1]{\ensuremath{\left \lfloor #1 \right \rfloor }}

\def\pcircuit{\textup{\textsc{Pure-Circuit}}\xspace}
\def\dpcircuit{\textup{\textsc{$\delta$-Pure-Circuit}}\xspace}
\def\gcircuit{\textup{\textsc{GCircuit}}\xspace}
\def\gcircuitp{\textup{\textsc{GCircuit+}}\xspace}
\def\edgcircuit{\textup{\textsc{$(\eps,\delta)$-GCircuit}}\xspace}
\def\edgcircuitp{\textup{\textsc{$(\eps,\delta)$-GCircuit+}}\xspace}

\newcommand{\pmin}{P_{\text{min}}\xspace}
\newcommand{\pmax}{P_{\text{max}}\xspace}
\newcommand{\bmax}{e_{\text{max}}\xspace}
\newcommand{\bmin}{e_{\text{min}}\xspace}
\newcommand{\epsm}{\eps_\mathfrak{m}\xspace}
\newcommand{\epsc}{\eps_\mathtt{c}\xspace}
\newcommand{\deltam}{\delta_\mathfrak{m}\xspace}
\newcommand{\deltac}{\delta_\mathtt{c}\xspace}

\newcommand{\pcppad}{\ensuremath{\textup{\textsf{PCP-for-PPAD}}}\xspace}

\newcommand{\garbo}{\ensuremath{\bot}\xspace}

\newcommand{\val}[1]{\boldsymbol{\mathrm{x}}[#1]}
\newcommand{\valonly}{\boldsymbol{\mathrm{x}}}
\newcommand{\valtwo}[1]{\boldsymbol{\mathrm{x'}}[#1]}
\newcommand{\valtwoonly}{\boldsymbol{\mathrm{x'}}}

\newcommand{\PURE}{\textup{\textsf{PURIFY}}\xspace}
\newcommand{\PURIFY}{\PURE}

\newcommand{\NAND}{\textup{\textsf{NAND}}\xspace}
\newcommand{\NOT}{\textup{\textsf{NOT}}\xspace}
\newcommand{\OR}{\textup{\textsf{OR}}\xspace}
\newcommand{\AND}{\textup{\textsf{AND}}\xspace}

\newcommand{\marbpb}{\textsf{\textup{marginal-bpb}}}

\newcommand{\eps}{\ensuremath{\varepsilon}\xspace}
\newcommand{\del}{\ensuremath{\delta}\xspace}

\newcommand{\aux}{\ensuremath{\mathsf{aux}}\xspace}

\newcommand{\inverter}{\ensuremath{\mathsf{inverter}}\xspace}
\newcommand{\ing}{\ensuremath{\mathsf{in}}\xspace}

\newcommand{\outg}{\ensuremath{\mathsf{out}}\xspace}

\newcommand{\inv}{\ensuremath{\mathsf{inv}}\xspace}
\newcommand{\inslope}{a}
\newcommand{\inthres}{t}
\newcommand{\Uopt}{U_{\textup{opt}}}
\newcommand{\Uactual}{U}

\title{Fisher Markets with Approximately Optimal Bundles and the Need for a PCP Theorem for PPAD}

\author{
\begin{tabular}{cc}
& \\
\textbf{Argyrios Deligkas} & \textbf{John Fearnley}\\
\small{Royal Holloway, United Kingdom} & \small{University of Liverpool, United Kingdom}\\
\href{mailto:argyrios.deligkas@rhul.ac.uk}{\small{\texttt{argyrios.deligkas@rhul.ac.uk}}} & \href{mailto:john.fearnley@liverpool.ac.uk}{\small{\texttt{john.fearnley@liverpool.ac.uk}}}\\
& \\
\textbf{Alexandros Hollender} & \textbf{Themistoklis Melissourgos}\\
\small{Technical University of Munich, Germany} & \small{University of Essex, United Kingdom}\\
\href{mailto:alexandros.hollender@tum.de}{\small{\texttt{alexandros.hollender@tum.de}}} & \href{mailto:themistoklis.melissourgos@essex.ac.uk}{\small{\texttt{themistoklis.melissourgos@essex.ac.uk}}}\\
& \\
\end{tabular}
}

\date{}

\begin{document}

\maketitle
\thispagestyle{empty}

\begin{abstract}
We study the problem of computing a competitive equilibrium with approximately optimal bundles in Fisher markets with separable piecewise-linear concave (SPLC) utility functions, meaning that every buyer receives a $(1-\delta)$-optimal bundle, instead of a perfectly optimal one. We establish the first intractability result for the problem by showing that it is \ppad/-hard for some constant $\delta > 0$, assuming the \pcppad conjecture. This hardness result holds even if all buyers have identical budgets (competitive equilibrium with equal incomes), linear capped utilities, and even if we also allow $\eps$-approximate clearing instead of perfect clearing, for any constant $\eps < 1/9$. Importantly, we show that the \pcppad conjecture is in fact required to show hardness for constant $\delta$: showing \ppad/-hardness for finding such approximate market equilibria in a broad class of markets encompassing those generated by our hardness result would prove the conjecture. This is the first natural problem where the conjecture is provably required to establish hardness for it.
\end{abstract}

\newpage

\setcounter{page}{1}

\newpage
\pagenumbering{arabic}

\section{Introduction}
\label{sec:intro}

We study Fisher markets~\citep{brainard2005compute-fisher}, which are a
fundamental model of a market in which a set of buyers are interested in
purchasing a set of goods. Unlike more general market models, buyers are not endowed with any goods, but instead come to the market with money and they spend this budget to buy goods in a way that maximizes their utility. Importantly, the buyers do not have any utility for money, but only for goods.

For Fisher markets, it is known that when the utility functions satisfy some standard sufficiency
conditions, then a \emph{competitive equilibrium}, or {\em market equilibrium}, is guaranteed to exist. A market
equilibrium consists of prices for each of the goods and an allocation of the
goods to the buyers such that (a) the market {\em clears}, i.e., for every good
its supply is exactly equal to its demand, and (b) every buyer buys an {\em
optimal bundle}, i.e., a set of goods that maximizes their utility under their budget constraint at the current prices.

\paragraph{\bf The computational complexity of finding market equilibria.}

The problem of computing a market equilibrium in a Fisher market has received a
lot of attention in prior work. When buyers have linear utility functions, a
market equilibrium can always be computed in polynomial
time~\citep{Devanur2008-market_equilibrium,orlin2010improved,vegh2012strongly}.
However, the problem becomes intractable as soon as one considers {\em additive
separable piecewise-linear concave} (SPLC) utilities, which are a 
generalization of linear utilities.  A buyer with an SPLC utility function has
a piecewise-linear concave utility for each good, and their utility for a
bundle of goods is simply the sum of their utilities for each of the individual
goods. Finding an equilibrium in a Fisher market with SPLC utilities is
known to be \ppad/-complete, i.e., the problem does not admit a polynomial-time algorithm, unless \ppad/ =
\p/~\citep{vazirani2011market-plc-ppad,chen2009spending,DeligkasFHM24-fisher-constant}.

In fact, this hardness result continues to hold even for {\em
\eps-approximate} market equilibria. In an \eps-approximate market equilibrium, the clearing constraint is relaxed and only requires every good to \eps-clear, meaning that the discrepancy between
the supply and the demand for any good is at most \eps.
\citet{chen2009spending} and \citet{vazirani2011market-plc-ppad} showed that it
is \ppad/-complete to find an \eps-approximate market equilibrium when \eps is
inversely polynomial in the size of the market, i.e., the number of buyers and
the number of goods, and  more recently, \citet{DeligkasFHM24-fisher-constant}
proved \ppad/-completeness even for {\em constant} $\eps < 1/11$, ruling out
the existence of a polynomial-time approximation scheme (PTAS) for \eps-approximate market equilibria.

\paragraph{\bf Approximately optimal bundles.}

The computational complexity lower bounds have so far only considered
approximations that relax the clearing constraint of the market, while still
insisting that each buyer should receive an optimal allocation. It is natural
to ask what happens if the buyers' optimality is relaxed instead. In that case, we  allow every buyer to receive a $(1-\delta)$-optimal bundle, that
is, a bundle that gives them utility that is at least a $(1 - \delta)$-fraction of the utility of an optimal bundle. This is arguably a more natural relaxation of the problem,
since it only requires the buyers be willing to accept slightly sub-optimal
bundles, whereas $\eps$-approximate clearing requires the market to find a way to deal with the mismatch between supply and demand (e.g., by destroying an $\eps$-fraction of some goods).

For approximately optimal bundles, \cite{garg2025approximating} recently obtained a positive result. They proved that for a large family of utility
functions, which include SPLC utilities, any allocation that maximizes Nash
welfare can be coupled with appropriate prices so that the market exactly clears and every
buyer receives
a $1/2$-optimal bundle. The Nash welfare of an allocation is the product of the
utilities of the buyers under this allocation, and it is known that such
an allocation can be computed in polynomial time using convex optimization. So, in particular, this shows that it is possible to efficiently compute a market equilibrium with $(1-\delta)$-optimal bundles when $\delta = 1/2$. However, so far no lower bounds under this relaxation have been established. A very natural question is thus: does this problem admit a PTAS?

\paragraph{\bf The \pcppad conjecture.}

Unrelated to Fisher markets, there has been a line of work that has studied the 
\pcppad conjecture. The $\eps$-\gcircuit problem is a well-known \ppad/-complete
problem that asks us to find an $\eps$-approximate satisfying assignment to a
fixed-point problem defined by arithmetic gates~\citep{ChenDT09-Nash}, and this
problem is known to be \ppad/-hard even when $\eps$ is
constant~\citep{Rubinstein18-Nash-inapproximability,DFHM24}. 

We can further weaken this to the $(\eps, \delta)$-\gcircuit problem, which
requires only that a $(1 - \delta)$-fraction of the gates are
$\eps$-approximately satisfied, and so a $\delta$-fraction of the constraints
are allowed to be broken. The so-far unproven \pcppad conjecture states
that 
there exist small constants $\eps$ and $\delta$ such that 
$(\eps, \delta)$-\gcircuit is \ppad/-hard. This conjecture was first proposed by \citet{BPR16}, where they hoped
that it would be useful in showing that finding an approximate Nash equilibrium in a
bimatrix game requires quasi-polynomial time.
In the end, however, the
conjecture was not used for this task, since the lower bound for
bimatrix games was shown via alternate means~\citep{Rubinstein16-two-player}.

The \pcppad conjecture has found other uses, however, and it has recently been used to
prove conditional hardness for computing approximate non-perfect Markov
stationary coarse correlated equilibria in multiplayer
games~\citep{park2023multi}, for computing stationary CCE in two-player stochastic
games~\citep{daskalakis2023complexity}, and very recently for computing approximate equilibria in Hylland-Zeckhauser markets~\citep{BravermanLXZ26-hylland-PCP}. 

All of these hardness results show \ppad/-hardness assuming that the \pcppad conjecture
holds. As such, they are open to the straightforward criticism that the \pcppad
conjecture may simply be false, meaning that no lower bound is actually shown.
Another criticism is that the \pcppad conjecture {\em may not be necessary at all} to
show hardness for these problems, and there may be a direct \ppad/-hardness
reduction that we have so far missed, as was the case for bimatrix games.

\subsection{Our Contribution}

This paper provides the first intractability results for market
equilibria with approximately optimal bundles, i.e., when the buyers are happy to receive $(1-\delta)$-optimal bundles. Our main hardness result is the following:
\begin{itemize}
\item Assuming that the \pcppad conjecture holds, there exists a suitably small constant $\delta > 0$ such that it is \ppad/-hard to find a market equilibrium with $(1-\delta)$-optimal bundles.
\end{itemize}
Furthermore, as a byproduct of the reduction establishing this result, we also obtain the following two unconditional \ppad/-hardness results:
\begin{itemize}
\item If $\delta > 0$ is inverse-polynomial (i.e., given in the input in unary), then it is \ppad/-hard to find a market equilibrium with $(1-\delta)$-optimal bundles.
\item If buyers are not allowed to spend any money on goods for which they have zero utility, then there exists a constant $\delta > 0$ such that it is \ppad/-hard to find a market equilibrium with $(1-\delta)$-optimal bundles.
\end{itemize}
Importantly, these three results continue to hold if we also \emph{simultaneously} relax the equilibrium notion by allowing $\eps$-approximate clearing, instead of perfect clearing. In fact, they hold for any constant $\eps < 1/9$. In particular, our results improve on the state of the art even for the setting with perfectly optimal bundles, since \ppad/-hardness was previously only known for $\eps < 1/11$~\citep{DeligkasFHM24-fisher-constant}.

Furthermore, our hardness results apply to a fairly simple class of markets. In particular, all buyers have \emph{linear capped} utilities and identical budgets. This latter restriction corresponds to the setting of competitive equilibrium with equal incomes (CEEI), which is known to have desirable fairness properties. To the best of our knowledge, these are the first intractability results for CEEI, even for exact equilibria.\footnote{Some works~\citep{BPR16,Rubinstein18-Nash-inapproximability,othman2016complexity}
prove \ppad/-hardness for approximate CEEI in a setting with indivisible goods, which should not be confused with our divisible setting.}

By itself, our main hardness result for constant $\delta$ is open to the same
criticisms as the other results that depend on the \pcppad conjecture: perhaps
the \pcppad conjecture is false, or perhaps the hardness results can be shown
without it. However, our second and perhaps more interesting result shows that
this is not the case. Namely, we prove the following:
\begin{itemize}
    \item If it is \ppad/-hard to find a market equilibrium with $(1-\delta)$-optimal bundles for some constant $\delta > 0$ in a broad class of Fisher markets encompassing those generated by our hardness result, then the \pcppad conjecture holds.
\end{itemize}

This shows that resolving the \pcppad conjecture is \emph{required} to resolve
the computational complexity of finding equilibria with approximately optimal
bundles in this broad class of markets: the problem is hard for constant $\delta$ if and only if the \pcppad
conjecture holds. We note that this is the first \emph{natural}
problem for which this property is known to hold. Namely, the problem satisfies the following two conditions: (i) it does not have an inherent ``PCP-like'' structure in which a $\delta$-fraction of
the constraints are allowed to fail, and (ii) it has been studied for its own sake in papers unrelated to the \pcppad conjecture. The only other problem whose hardness is known to be equivalent to the \pcppad conjecture is the $(\eps,
\delta)$-weak approximate Nash equilibrium problem for polymatrix games~\citep{BPR16}. However, as far as we are aware, this problem has only been considered in the context of the \pcppad conjecture, and, perhaps more importantly, it is inherently ``PCP-like'': it allows a $\delta$-fraction of
the players to not be in $\eps$-equilibrium, namely to play an arbitrarily bad response. On the other hand, a market
equilibrium with approximately optimal bundles requires that \emph{all} goods exactly clear, and that \emph{all}
buyers receive a $(1 - \delta)$-approximate bundle.

We argue that this revitalizes the \pcppad conjecture, and provides a compelling
need for its proof or disproof. Before our work, the conjecture had only been
used to show conditional hardness results, and for those works the criticism
that perhaps the \pcppad conjecture could be bypassed always applied. Since we
now have a natural problem that is hard if and only if the conjecture holds,
resolving the conjecture now appears necessary.

\subsection{Further related work}
The computation of market equilibria, for both exact and approximate
equilibria, has
received significant attention over the years. For Fisher markets,
\citet{vazirani2011market-plc-ppad} and \citet{chen2009spending} have
established \ppad/-hardness for SPLC utilities albeit for a sub-constant \eps.
On the other hand, polynomial-time algorithms were derived for the cases where
the utility functions of the buyers are
linear~\citep{Devanur2008-market_equilibrium,orlin2010improved,vegh2012strongly},
homogeneous~\citep{eisenberg1961aggregation}, or weak gross
substitutes~\citep{codenotti2005polynomial}.

Furthermore, when the number of goods is constant \citet{kakade2004graphical}
gave a PTAS while~\citet{devanur2008market} gave a polynomial-time algorithm
for exact equilibria. In fact, the algorithm of~\cite{devanur2008market} also works
when the number of buyers is constant in the SPLC utility setting.
For non-separable PLC utilities~\citet{garg2022approximating} derived a fixed
parameter approximation scheme that has the number of buyers as a parameter.

Market equilibria where buyers get approximately optimal bundles (both with and without approximate clearing) have been studied in the context of obtaining approximation algorithms for various classes of utilities or variants of the problem~\citep{DengPS03-price-equilibria,JainMS03-approx-market,codenotti2005market,codenotti2005polynomial,codenotti2007computation,KarandeD07-WGS,ChakrabortyDK10-market-transaction,GargHM24-satiation}.

Matching markets are an interesting variant of general Fisher markets.
\citet{alaei2017computing} designed a polynomial-time algorithm for markets with a 
constant number of goods or buyers, while~\citet{vazirani2020computational}
derived a polynomial algorithm when the buyers have dichotomous utilities. For
one-sided matching markets, the most famous problem is the Hylland-Zeckhauser
market, for which the existence of an equilibrium was initially established
in~\citep{hylland1979efficient} and was recently simplified by
\citet{braverman2021optimization}. \citet{chen2022computational-Hylland-Zeckhauser-PPAD-h} have established
\ppad/-completeness for the problem, and, as mentioned above, very recently hardness of constant approximation was established subject to the \pcppad conjecture~\citep{BravermanLXZ26-hylland-PCP}.

There has also been interest in Fisher markets with
additional constraints. \citet{birnbaum2010new}, \citet{devanur2004spending},
and~\citet{vazirani2010spending} considered the case where the utilities of the
buyers depend on prices of goods through spending constraints.
\citet{jalota2023fisher} considered additional linear constraints that include
matching markets, and they gave a tâtonnement process which was
found to converge to a market equilibrium in experiments.

The \pcircuit problem was recently introduced in~\citep{DFHM24} , where it was used
to prove strong, improved, \ppad/-hardness results for a variety of problems
related mainly to approximate Nash equilibria. Since then it was further used
in~\citep{deligkas2023tight} to prove tight \ppad/-hardness for approximate Nash
equilibria in graphical games, in~\citep{ioannidis2023clearing} and in~\citep{dohn2025improved} to prove stronger \ppad/-hardness in the problem of clearing financial
networks, and in~\citep{hansen2025complexity} to derive improved inapproximability \ppad/-hardness for Markov equilibria in stochastic games.

\section{Preliminaries}

\subsection{Fisher Markets}\label{sec:fisher_markets}

\paragraph{\bf Fisher markets.} A Fisher market is given by a tuple $(G,B,(e_i)_{i\in B},(u_i)_{i \in B})$, where:
\begin{itemize}
\item $G$ is a set of (divisible) goods. Without loss of generality, we assume that there is one unit of each good available.\footnote{This can be achieved by a simple normalization, and it simplifies the expression for the clearing constraint below.}
\item $B$ is a set of buyers.
\item For every $i \in B$, $e_i > 0$ is the budget of buyer $i$.
\item For every $i \in B$, $u_i: \mathbb{R}_{\geq 0}^{|G|} \to \mathbb{R}_{\geq 0}$ is the utility function of buyer $i$. For any allocation $x_i \in \mathbb{R}_{\geq 0}^{|G|}$ of goods to buyer $i$ (where $x_{i,j} \geq 0$ denotes the amount of good $j$ allocated to buyer $i$), $u_i(x_i)$ denotes the utility derived by the buyer. We assume that the utility functions are separable piecewise-linear concave (SPLC), meaning that $u_i(x_i)$ can be written as $\sum_{j\in G} u_{i,j}(x_{i,j})$, where each $u_{i,j}: \mathbb{R}_{\geq 0} \to \mathbb{R}_{\geq 0}$ satisfies
\begin{enumerate}
    \item $u_{i,j}(0) = 0$,
    \item $u_{i,j}$ is continuous and piecewise-linear,
    \item $u_{i,j}$ is concave and non-decreasing.
\end{enumerate}
In particular, we represent each piecewise-linear concave utility function
$u_{i,j}$ as a sequence $\langle (s_{i,j,1}, \ell_{i,j,1}), (s_{i,j,2},
\ell_{i,j,2}), \dots, (s_{i,j,m_{i,j}}, \ell_{i,j,m_{i,j}}) \rangle$ where for each $k$ we have
that $s_{i,j,k}$ gives a slope of a linear piece and $\ell_{i,j,k}$ gives the
length of that piece, and where $m_{i,j}$ is the number of pieces used by
$u_{i,j}$. The length of the last piece can be infinite.
So to compute $u_{i,j}(x)$ we find the largest value $a$ such
that $\sum_{k=1}^a \ell_{i,j,k} \le x$ and then we set
$$u_{i,j}(x) = \sum_{k=1}^a (s_{i,j,k} \cdot \ell_{i,j,k}) + (x - \sum_{k=1}^a
\ell_{i,j,k}) \cdot s_{i,j,a+1}.$$
The concavity of the function implies that $s_{i,j,k} \ge s_{i,j,k+1}$ for all
$k$, and we can without loss of generality assume that $s_{i,j,k} >
s_{i,j,k+1}$ for all $a$, since we can simply merge adjacent pieces with
identical slopes.

A notable special case of SPLC utilities are \emph{linear capped} utilities, where $u_{i,j}$ is of the form $u_{i,j}(x_{i,j}) = \min\{a \cdot x_{i,j}, b\}$ for some $a \in \mathbb{R}_{\geq 0}$ and $b \in \mathbb{R}_{\geq 0} \cup \{+ \infty\}$. Our hardness result will apply to these utilities.
\end{itemize}

\paragraph{\bf Approximately optimal bundles.}
Given $\delta \in [0,1]$ and a price vector $p \in \mathbb{R}_{\geq 0}^{|G|}$, where $p_j$ denotes the price of good $j$, the set of $(1-\delta)$-optimal bundles for buyer $i$, denoted $\opt^\delta_i(p) \subseteq \mathbb{R}_{\geq 0}^{|G|}$, is the set of all $(1-\delta)$-optimal solutions of the following optimization problem:
\begin{equation}\label{eq:opt-bundle}
\begin{split}
\max \quad &u_i(x_i) \\
\text{ s.t.} \quad  
& \sum_{j \in G} p_j x_{i,j} \leq e_i \\
& x_{i,j} \geq 0 \quad \forall j \in G.
\end{split}
\end{equation}
In other words, letting $F_i(p) \subseteq \mathbb{R}_{\geq 0}^{|G|}$ denote the feasible set of optimization problem~\eqref{eq:opt-bundle}, and $U_i^* := \max_{x_i \in F_i(p)} u_i(x_i) \in \mathbb{R}_{\geq 0} \cup \{+ \infty\}$ its optimal value, we can write
$$\opt^\delta_i(p) := \{x_i \in F_i(p): u_i(x_i) \geq (1-\delta) \cdot U_i^*\}.$$
For $\delta = 0$, this simply corresponds to the set of all optimal solutions to~\eqref{eq:opt-bundle}.
Note that it is possible for $\opt^\delta_i(p)$ to be empty. Indeed, it is possible to have $U_i^* = \infty$, when some good has price zero and the buyer is never satiated with that good.

For SPLC utilities, we will often be interested in the \emph{bang-per-buck} of
a particular utility-function segment. The bang-per-buck of the $k$th segment
of $x_{i,j}$ under the price vector $p$ is defined as $s_{i,j,k}/p_j$.

\paragraph{\bf Approximate competitive equilibrium.}
For any $\eps, \delta \in [0,1]$, an $(\eps, \delta)$-approximate market equilibrium is a price vector $p$ and an allocation vector $x = (x_i)_{i \in B}$ satisfying the following conditions:
\begin{enumerate}
    \item For each buyer $i$, $x_i$ is a $(1-\delta)$-optimal bundle at prices $p$, i.e., $x_i \in \opt^\delta_i(p)$.
    \item For each good $j$, the market clears approximately up to $\eps$ units of good, i.e.,
    $$\left| \sum_{i \in B} x_{i,j} - 1 \right| \le \eps.$$
\end{enumerate}
When $\eps = \delta = 0$, this corresponds to an exact market equilibrium.

\paragraph{\bf Existence of equilibria.}
The following condition is sufficient to guarantee the existence\footnote{In fact, existence of a market equilibrium is guaranteed even without this condition. However, the condition guarantees the very desirable property that any market equilibrium is Pareto optimal, see e.g.~\citep{garg2025approximating}.} of a market equilibrium \citep{Maxfield1997,vazirani2011market-plc-ppad}:

\begin{quote}
    \textbf{Sufficient Condition:} For every buyer $i \in B$, there exists a good $j \in G$ such that $u_{i,j}$ is a strictly increasing function (i.e., buyer $i$ is never satiated with good $j$).
\end{quote}

\paragraph{\bf Computational problem.}
Let $\eps, \delta \in [0,1]$. The computational problem of computing an $(\eps, \delta)$-approximate market equilibrium is defined as follows: 
\begin{quote}
    \textbf{Input:} A Fisher market $(G,B,(e_i)_{i\in B},(u_i)_{i \in B})$ with SPLC utilities satisfying the sufficient condition for the existence of equilibria. For each $i \in B$ and $j \in G$, $u_{i,j}$ is explicitly described in the input, i.e., for each linear affine piece we are given the positions and values at its endpoints.
\end{quote}
\begin{quote}
    \textbf{Output:} An $(\eps, \delta)$-approximate market equilibrium $(p,x)$.
\end{quote}
Given $(p,x)$, the equilibrium conditions can be verified in polynomial time, because, for SPLC utilities, the optimization problem \eqref{eq:opt-bundle} used to define $\opt^\delta_i(p)$ can be solved in polynomial time using a simple greedy approach; see, e.g., \citep{Garg2015-pivot}. Together with the existence of solutions guaranteed by the sufficient condition, this puts the problem in the complexity class \tfnp/ of total \np/ search problems. Prior work~\citep{vazirani2011market-plc-ppad} has shown that the problem lies in the subclass \ppad/ of \tfnp/, even for $\eps = \delta = 0$. In particular, exact rational solutions are guaranteed to exist. The problem is known to be \ppad/-complete when $\delta = 0$ and $\eps < 1/11$~\citep{chen2009spending,vazirani2011market-plc-ppad,DeligkasFHM24-fisher-constant}. 
No hardness result is known for any constant $\delta > 0$, even when $\eps = 0$.

\subsection{\pcircuit, \gcircuit, and the PCP conjecture for \ppad/}

In this section, we formally present the PCP conjecture for \ppad/ and two new equivalent formulations. The original formulation is in terms of the generalized circuit problem. The first new equivalent formulation is in terms of the \pcircuit problem and it makes it easier to prove hardness results assuming the conjecture. The other new formulation makes it easier to prove that the conjecture is needed to prove some hardness result. We prove that the three formulations are equivalent and so can be used interchangeably.

\subsubsection{The \pcircuit Problem}\label{sec:pcircuit_def}

An instance of the \pcircuit problem is given by a set of nodes (or \emph{variables}) $V=[n]$ and a set
$C$ of gate-constraints (or just \emph{gates}). Each gate $g \in C$ is of the
form $g = (T,u,v,w)$ where $u,v,w \in V$ are distinct nodes, and $T \in \{\NAND, \PURE\}$ is the type of the gate, with the following interpretation.
\begin{itemize}
	\item If $T=\NAND$, then $u$ and $v$ are the inputs of the gate, and $w$ is its output.
	\item If $T=\PURE$, then $u$ is the input of the gate, and $v$ and $w$ are its outputs.
\end{itemize}
We require that each node is the output of exactly one gate. A node can be used as an input by multiple gates.

A solution to instance $(V,C)$ is an assignment $\valonly: V \to \{0,1,\garbo\}$ that satisfies all the gates (see~\cref{fig:gates}), i.e., for each gate $g=(T,u,v,w) \in C$ we have the following. 
\begin{itemize}
\item If $T=\NAND$ in $g = (T,u,v,w)$, then $\valonly$ satisfies
    \begin{align*}
        \val{u} = \val{v} = 1 \implies \val{w} = 0\\
        (\val{u} = 0) \lor (\val{v} = 0) \implies \val{w} = 1
    \end{align*}

	\item If $T=\PURE$, then $\valonly$ satisfies
	\begin{align*}
		& \{\val{v}, \val{w}\} \cap \{0,1\} \neq \emptyset\\
		& \val{u} \in \{0,1\} \implies \val{v} = \val{w} = \val{u}.
	\end{align*}
\end{itemize}

\begin{figure*}[t!]
     \begin{center}
    \begin{minipage}{0.35\textwidth}
        \begin{center}
            \begin{tabular}{c|c||c}
                $u$ & $v$ & $w$ \\ \hline
                1 & 1 & 0 \\
                0 & $\{0,1,\garbo\}$ & 1 \\
                $\{0,1,\garbo\}$ & 0 & 1 \\
                \multicolumn{2}{c||}{Else} & $\{0,1,\garbo\}$
            \end{tabular}
			\caption*{\NAND gate}
        \end{center}
    \end{minipage}
	\begin{minipage}{0.35\textwidth}
		\begin{center}
			\begin{tabular}{c||c|c}
				$u$ & \phantom{xx}$v$\phantom{xx}  & $w$ \\ \hline
				$0$ & $0$ & $0$ \\
				$1$ & $1$ & $1$ \\
				\multirow{2}{*}{$\garbo$} & \multicolumn{2}{c}{At least one} \\
				& \multicolumn{2}{c}{output in  $\{0, 1\}$}
			\end{tabular}
			\caption*{\PURE gate}
		\end{center}
	\end{minipage}
    \end{center}
	\caption{The truth tables of the two gates of \pcircuit.}
	\label{fig:gates}
\end{figure*}

\begin{theorem}[\citep{DFHM24}]
	\label{thm:pancircuit}
	\pcircuit is \ppad/-complete, even when every node is the input to exactly one gate.
\end{theorem}

For any constant $\delta \in [0,1]$, the \dpcircuit problem is defined as: given a \pcircuit instance $(V,C)$, find an assignment $\valonly \in \{0,1,\garbo\}^{|V|}$ such that at least a $(1-\del)$ fraction of the gates in $C$ are satisfied. The \pcppad conjecture can be formulated as follows.

\bigskip

\noindent
{\bf \pcppad Conjecture \textmd{(\pcircuit version)}.}
\textit{There exists a constant $\delta > 0$ such that $\dpcircuit$ is \ppad/-hard, even when every node is the input to exactly one gate.}

\bigskip

Furthermore, as we show in \cref{thm:pcp-equivalent}, removing the restriction on every node being the input to exactly one gate does not change the conjecture, i.e., the versions with and without that restriction are equivalent.

The conjecture was originally formulated in terms of the generalized circuit problem, \gcircuit, which we define next.
In \cref{thm:pcp-equivalent} we show that these formulations are equivalent. We expect the new simplified formulation in terms of \pcircuit to be useful for future work.

\subsubsection{The \gcircuit Problem}

A generalized circuit is defined in the following way.

\begin{definition}[Generalized Circuit~\citep{ChenDT09-Nash}]
A generalized circuit is a tuple $(V, T)$, where $V$ is a set of nodes, and $T$
is a set of gates. Each gate $t \in T$ is a five-tuple $(G, u, v, w, c)$,
where $G$ is a gate type from the set $\{G_c, G_{\times c}, G_=, G_+, G_-, G_<,
G_\lor, G_\land, G_\lnot\}$, $u, v \in V \cup \{\textsf{nil}\}$ are input
variables, $w \in V$ is an output variable, and $c \in [0, 1]
\cup \{\textsf{nil}\}$ is a rational constant.

The following requirements must be satisfied for each gate 
$(G, u, v, w, c) \in T$.
\begin{itemize}
\item 
$G_c$ gates take no input variables and uses a constant in $[0, 1]$. So $u = v = \textsf{nil}$ and $c \in [0, 1]$ whenever $G = G_c$.

\item

$G_{\times c}$ gates take one input variable and a constant. So $u \in V$, $v =
\textsf{nil}$, and $c \in [0, 1]$ whenever $G = G_{\times c}$. 

\item 
$G_{=}$ and $G_{\lnot}$ gates take one input variable and do not use a constant. So $u \in V$,
$v = c = \textsf{nil}$, whenever $G \in \{G_{=}, G_{\lnot}\}$. 

\item All other gates take two input variables and do not use a constant. So 
$u \in
V$, $v \in V$, and $c = \textsf{nil}$ whenever $G \notin \{G_c, G_{\times c},
G_{=}, G_{\lnot}\}$. 

\item Every variable in $V$ is the output variable for exactly one gate. More
formally, for each variable $w \in V$, there is exactly one gate $t \in T$ such
that $t = (G, u, v, w, c)$.  
\end{itemize}
\end{definition}

For constants $\eps, \delta \in [0,1]$, the \edgcircuit problem is defined as follows. Given a generalized circuit
$(V, T)$, find a vector $\valonly \in [0, 1]^{|V|}$ such that a $(1- \delta)$
fraction of the gates in $T$ satisfy the following constraints.

\begin{center}
\begin{tabular}{l|l}
Gate & Constraint \\ \hline
$(G_c, \textsf{nil}, \textsf{nil}, w, c)$ & $\val{w} = c \pm \eps$ \\
$(G_{\times c}, u, \textsf{nil}, w, c)$ & $\val{w} = \min(\val{u} \cdot c, 1) \pm \eps$ \\
$(G_{=}, u, \textsf{nil}, w, \textsf{nil})$ & $\val{w} = \val{u} \pm \eps$ \\
$(G_{+}, u, v, w, \textsf{nil})$ & $\val{w} = \min(\val{u} + \val{v}, 1) \pm \eps$ \\
$(G_{-}, u, v, w, \textsf{nil})$ & $\val{w} = \max(\val{u} - \val{v}, 0) \pm \eps$ \\
$(G_{<}, u, v, w, \textsf{nil})$ & $\val{w} = 
    \begin{cases} 
    1 \pm \eps & \text{if } \val{u} < \val{v} - \eps \\ 
    0 \pm \eps & \text{if } \val{u} > \val{v} + \eps 
    \end{cases}$ \\
$(G_{\lor}, u, v, w, \textsf{nil})$ & $\val{w} = 
    \begin{cases} 
    1 \pm \eps & \text{if } \val{u} \geq 1 - \eps \text{ or } \val{v} \geq 1 - \eps \\ 
    0 \pm \eps & \text{if } \val{u} \leq \eps \text{ and } \val{v} \leq \eps
    \end{cases}$ \\
$(G_{\land}, u, v, w, \textsf{nil})$ & $\val{w} = 
    \begin{cases} 
    1 \pm \eps & \text{if } \val{u} \geq 1 - \eps \text{ and } \val{v} \geq 1 - \eps \\ 
    0 \pm \eps & \text{if } \val{u} \leq \eps \text{ or } \val{v} \leq \eps
    \end{cases}$ \\
$(G_{\lnot}, u, \textsf{nil}, w, \textsf{nil})$ & $\val{w} = 
    \begin{cases} 
    1 \pm \eps & \text{if } \val{u} \leq \eps \\
    0 \pm \eps & \text{if } \val{u} \geq 1 - \eps
    \end{cases}$ \\
\end{tabular}
\end{center}
Here the notation $a = b \pm \eps$ is used as a shorthand for $a \in [b - \eps, b + \eps]$.

The \pcppad conjecture was originally introduced by \cite{BPR16} and formulated as follows.\footnote{\citet{BPR16} actually give a more specific conjecture regarding the existence
of a quasi-linear reduction from End-of-the-Line, which is the canonical
\ppad/-complete problem, to $\edgcircuit$. However, all subsequent works use the weaker conjecture that we also use in this work~\citep{Rubinstein16-two-player,park2023multi,daskalakis2023complexity}.}

\bigskip

\noindent
{\bf \pcppad Conjecture \textmd{(\gcircuit version)}.}
\textit{There exist constants $\eps > 0$ and $\delta > 0$ such that $\edgcircuit$ is
\ppad/-hard.}

\bigskip

As mentioned above, in \cref{thm:pcp-equivalent} we show that this is equivalent to our formulation in terms of \pcircuit presented earlier. For the purpose of proving that some problem is \ppad/-hard assuming the \pcppad conjecture, it is easier to use the \pcircuit version of the conjecture, since the gates are easier to implement. However, for the purpose of showing that the \pcppad conjecture is \emph{necessary} to show \ppad/-hardness for some problem of interest, it is more convenient to have a formulation of the conjecture that uses a more expressive problem. Next, we present a third equivalent version of the conjecture in terms of a generalization of the \gcircuit problem.

\subsubsection{The \gcircuitp Problem}

When we reduce to \gcircuit, it will be more convenient to express our circuits
using more general gates. Specifically, we define the \edgcircuitp
problem to be the same as the \edgcircuit problem but with the following
modifications. Fix some absolute constants $L$ and $U$ with $L < U$, and
another absolute constant $b$.
\begin{itemize}
\item In this version of the problem, each variable $v
\in V$ comes equipped with a rational lower bound $v_l$ and a rational upper
bound $v_u$, where $v_l < v_u$ and $v_l, v_u \in [L,U]$. Then, instead of seeking a solution $\valonly \in [0, 1]^{|V|}$,
we seek a solution $\valonly$ where each $x_v$ is constrained to lie in the
range $[v_l, v_u]$. 
\item
We extend the $G_{\times c}$ gate to permit $c$ to be an arbitrary
rational in $[U, L]$, rather than restricting $c$ to lie in $[0, 1]$. 
\item We replace the $G_+$, $G_-$, and $G_{\times c}$ gates with versions that do not take a
$\min$ or $\max$ with $1$ or $0$, but instead use the bounds for the output
variable. That is, we use the following definitions.
\begin{itemize}
\item The $(G_{+}, u, v, w, \textsf{nil})$ gate now requires that $\val{w} =
\max(\min(\val{u} + \val{v}, w_u), w_l) \pm \eps$.
\item The $(G_{-}, u, v, w, \textsf{nil})$ gate now requires that $\val{w} =
\max(\min(\val{u} - \val{v}, w_u), w_l) \pm \eps$.
\item The $(G_{\times c}, u, \textsf{nil}, w, c)$ gate now requires that
$\val{w} = \max(\min(\val{u} \cdot c, w_u), w_l) \pm \eps $.
\end{itemize}
Note that we now need to bound these operations from above and below, since we
allow negative variables in \gcircuitp, and we also allow negative
multiplicands in our $G_{\times c}$ gates.
\item We introduce explicit gates that allow us to take a $\min$ or $\max$ with a
rational value $c \in [L,U]$. 
\begin{itemize}
\item A $(G_{\max}, u, \textsf{nil}, w, c)$ gate requires that 
$\val{w} = \max(\val{u}, c) \pm \eps$. 
\item A $(G_{\min}, u, \textsf{nil}, w, c)$ gate requires that $\val{w} = \min(\val{u}, c) \pm \eps$.
\end{itemize}
\item For the comparison gate $G_<$, we require that the output node has lower bound $0$ and upper bound $1$. For the logical gates $G_{\lor}$, $G_{\land}$, and $G_{\lnot}$, we require that the input and output nodes have lower bound $0$ and upper bound $1$. For the constant gate $G_c$, we require that the constant $c$ lies between the lower and upper bound of the output node. Finally, for the equality gate $G_=$, we require that the bounds on the input and output are the same.
\end{itemize}

It is not hard to reduce this new
version of the problem to the original (see \cref{thm:pcp-equivalent}). As a result, we obtain the following third equivalent formulation of the conjecture.

\bigskip

\noindent
{\bf \pcppad Conjecture \textmd{(\gcircuitp version)}.}
\textit{There exist constants $\eps > 0$ and $\delta > 0$ such that $\edgcircuitp$ is \ppad/-hard.}

\bigskip

We will use this version to prove that the \pcppad conjecture is \emph{necessary} to obtain \ppad/-hardness for our market equilibrium problem.

\subsubsection{Equivalence of the three formulations}

In \cref{sec:app:pcppad-equivalence} we prove the following theorem which shows that all these formulations are indeed equivalent.

\begin{restatable}{theorem}{pcpequivalence}
\label{thm:pcp-equivalent}
The following are all equivalent formulations of the \pcppad conjecture:
\begin{enumerate}
    \item There exists a constant $\delta > 0$ such that $\dpcircuit$ is \ppad/-hard, even when every node is the input to exactly one gate.
    \item There exists a constant $\delta > 0$ such that $\dpcircuit$ is \ppad/-hard.
    \item There exist constants $\eps > 0$ and $\delta > 0$ such that $\edgcircuit$ is \ppad/-hard.
    \item There exist constants $\eps > 0$ and $\delta > 0$ such that $\edgcircuitp$ is \ppad/-hard.
\end{enumerate}
\end{restatable}

\section{Technical Overview}

\subsection{Hardness for Fisher Markets from the \pcppad conjecture}

We show hardness results for Fisher markets by reducing from the \dpcircuit
problem, whose \ppad/-hardness for some constant $\delta$ is equivalent to the \pcppad conjecture (by \cref{thm:pcp-equivalent}).

\paragraph*{\textbf{Reduction from prior work.}}
The starting point for obtaining such a reduction is to examine the existing reduction from \pcircuit to Fisher markets with approximate clearing (but with exact optimal bundles) in \citep{DFHM24-Fisher_constant}. Starting from a \pcircuit instance $(V,T)$, the reduction first constructs one good $g_v$ for each node $v \in V$. The goal of the rest of the construction of the market is then to ensure that, at equilibrium, the prices of these goods encode a solution to the \pcircuit instance. Namely, if the price $p_v$ of good $g_v$ is above some threshold $H$, then we assign the value $1$ to node $v$. If $p_v$ is below some threshold $L$, where $L < H$, then the value of $v$ is $0$. Finally, if the price $p_v$ lies between $L$ and $H$, then the node $v$ receives value $\garbo$. If, at any equilibrium, the extracted values for the nodes in $V$ yield a solution to the \pcircuit instance, then the reduction is correct.

The construction of the market proceeds as follows. In addition to the goods $g_v$ representing nodes in $V$, an additional good, called the \emph{reference good}, is introduced, together with a buyer who only desires that good. This buyer has an enormous budget (compared to the budgets of the other buyers that are introduced later), and thus spends all its money on that good at equilibrium. As a result, the price of the reference good is very stable, as it barely moves (in relative terms) depending on whether any other buyers spend any money on it or not.

The last step of the construction is to add some buyers whose job it is to implement the gates of the \pcircuit instance. The high level approach for this is to use the idea of \emph{anti-endowment}. Consider a buyer who has very high utility for good $g_v$, but only up to, say, $0.5$ units, and then has zero utility for any additional quantity of that good. As a result, at equilibrium the buyer will first buy $0.5$ units of good $g_v$, before spending any money on other goods. Thus, the higher the price of $g_v$, the lower the remaining budget of the buyer after buying those $0.5$ units, and similarly, the lower the price of $g_v$, the higher the remaining budget. If the buyer then has positive utility only for some other good $g_w$, it will spend all of its remaining budget on that good. Since that good has to (approximately) clear at equilibrium, it will need to have a high price if the remaining budget of the buyer was high. On the other hand, if the remaining budget was low, then the good will need to have a low price (assuming here for simplicity that no other buyers are buying any of it). We thus obtain that a high price for $g_v$ implies a low price for $g_w$, and similarly a low price for $g_v$ implies a high price for $g_w$.

This general idea can be used to implement each gate of \pcircuit by introducing new buyers that are only interested in the goods involved in the gate and that ensure that the correct constraint is enforced on their prices. These buyers also have positive utility for the reference good. This is used to ensure that the prices of the goods do not become too high: if the price of some good $g_v$ was too high, then the buyer would buy the reference good instead, and $g_v$ would not clear.

\paragraph*{\textbf{Approximately optimal bundles.}}
The aforementioned reduction from prior work assumes that each buyer is buying an optimal bundle. If we instead consider $(1-\deltam)$-optimal bundles, we would need to set $\deltam$ to an inverse-polynomial value for the reduction to work. Indeed, setting $\deltam$ to be a constant would fail due to various parameters in the reduction being set to inverse-polynomial values. This is due to the reference good and the imbalance between the enormous budget of one buyer and the much smaller budgets of all other buyers. In particular, the buyer with the enormous budget can spend a small constant fraction of its budget on arbitrary goods while keeping a $(1-\deltam)$-optimal bundle. But since its budget is very large, this constant fraction is huge compared to the budgets of the other buyers and can thus completely break the implementation of all the gates.

We bypass this obstacle by avoiding the use of a reference good altogether. This yields a simpler, more direct reduction where all parameters are constant, including the slopes of the utility functions and the degrees in the graph of interactions between buyers and goods. As a result, we can set $\deltam$ to be a sufficiently small constant and thus obtain \ppad/-hardness assuming the \pcppad conjecture. Furthermore, by setting $\deltam$ to a sufficiently small inverse-polynomial value, we also obtain unconditional \ppad/-hardness for that regime.

The intuition for why we need the \pcppad conjecture in order to show hardness for constant $\deltam$ is the following. If $\deltam$ is some constant, say $0.1$, then that means that every buyer can spend $10\%$ of their budget on arbitrary goods, as long as they spend their remaining $90 \%$ optimally. Indeed, doing so will still guarantee $90 \%$ of the optimal total utility they could have achieved by spending their whole budget optimally. But now this means that $10 \%$ of the total amount of money available in the market can be spent in a completely arbitrary manner. In particular, this ``rogue budget'' can completely disrupt the functionality of some gadgets implementing gates. However, because it is only a constant fraction of the total amount of money, we can show that this rogue budget can only disrupt a constant fraction of the gadgets. This is where the \pcppad conjecture comes in, as it precisely states that \pcircuit remains hard even if a constant fraction of the gates malfunction. 

The instances we construct have two further properties that yield improvements over prior work, even for the case where $\deltam = 0$.
\begin{itemize}
    \item The reduction works for any clearing parameter $\epsm < 1/9$, improving upon the previously best $\epsm < 1/11$ from \citep{DFHM24-Fisher_constant}, even when $\deltam = 0$.
    \item The constructed market satisfies the CEEI property, i.e., all the buyers have identical budget. No hardness result for this setting was known, even when $\epsm = \deltam = 0$.
\end{itemize}
To summarize, this reduction is a significant simplification of a construction appearing in prior work, and it also yields stronger results. For the details, see \cref{sec:hardness}. In the next section, we explain how we show that the \pcppad conjecture is in fact needed to prove hardness for constant $\deltam$, at least for a family of markets with some nice structure.

\subsection{Hardness for Fisher Markets Implies the \pcppad conjecture}

We now give a high-level overview of our result that if $(\eps,
\delta)$-approximate market equilibrium is \ppad/-hard for some constants $\eps, \delta$, then the \pcppad conjecture holds. In fact, this remains true even if $\eps=0$. Full details and proofs
can be found in~\cref{sec:markettopcp}.

\paragraph{\bf Reducible markets.}

We say that a market is \emph{reducible} if it satisfies the following
properties.
\begin{itemize}
\item It has constant degree $d$, meaning that each buyer is interested in at
most $d$ goods, and each good is desired by at most $d$ buyers. 

\item The budgets of all buyers lie in a range $[\bmin, \bmax]$, where $\bmin$
and $\bmax$ are both constants.

\item Each buyer has an SPLC utility function where each utility-function uses
a constant number of segments, and the slope of each
utility-function segment is either zero or lies in a range $[1, \kappa]$ where $\kappa$ is a
constant.
\end{itemize}

We show that if it is \ppad/-hard to find an $(\eps, \delta)$-approximate market
equilibrium in a reducible Fisher market for some constants $\eps$ and
$\delta$, then the \pcppad conjecture holds. We note in particular that
the family of instances generated by our hardness results is reducible, so we
get that the \pcppad conjecture holds if and only if it is \ppad/-hard to
find an $(\eps, \delta)$-approximate market equilibrium in a reducible Fisher
market. 

\subsubsection{High-level idea}

We begin by explaining the overall idea behind our construction. We want to
show that \ppad/-hardness of finding an $(\eps, \delta)$-approximate market
equilibrium for some constants $\eps$ and $\delta$ implies \ppad/-hardness of $(\eps',
\delta')$-\gcircuitp for some constants $\eps'$ and $\delta'$. So we need to
reduce from the market equilibrium problem to the \gcircuitp problem. We will do this
in a sequence of steps, which we will describe briefly here, before giving a
more detailed overview of each of them in the rest of this section.

\paragraph{\bf Step 1: Reduce to \gcircuitp.}

The first step is to formulate the market equilibrium problem as a $\gcircuitp$
instance, and then to consider what an $(\epsc, \deltac)$-solution of this
instance gives us, for carefully chosen constants $\epsc$ and $\deltac$. We use
a straightforward encoding in which prices and demands are represented by
variables in the circuit with the market equilibrium condition being enforced
by the gates. 

An exact solution of this $\gcircuitp$ instance would give us an exact market
equilibrium, but we only get an $(\epsc, \deltac)$ solution in which gates are
$\epsc$-satisfied, and in which a $\deltac$ fraction of the gates are allowed
to fail completely. Interpreting this solution directly in the market would
give us a solution in which an $O(\deltac)$ fraction of the goods are
not
$\eps$-clearing, and an $O(\deltac)$ fraction of the buyers are
not receiving $(1-\eps)$-optimal allocations.

An $(\eps, \delta)$-approximate market equilibrium requires that all of the
goods $\eps$-clear, and all of the buyers receive $(1-\delta)$-optimal
bundles, and so the rest of the steps are dedicated to repairing this initial
solution by fixing the goods that are far from clearing, and the buyers who are
far from optimal.

\paragraph{\bf Step 2: Fix the broken buyers.}

The first repair that we make is to fix the buyers that do not receive
approximately optimal allocations. We do this by simply instructing them to
ignore their current allocations and to instead buy an optimal allocation
according to the current prices. 

This comes at the cost of changing the demands of the goods that are adjacent
to these buyers, and these goods may now fail to approximately clear, even if
they were approximately clearing at the end of Step 1. This is not an issue
however, because the degree $d$ of the market is constant, and only a constant
$O(\deltac)$ fraction of the buyers need to be fixed, so only a constant
$O(\deltac)$ fraction of the goods are affected. Moreover, already in Step 1 we
had that an $O(\deltac)$ fraction of the goods did not approximately clear, so
the extra goods that are broken by this step do not meaningfully change our
position.

\paragraph{\bf Step 3: Fix the over-clearing.}

We now have that all buyers receive approximately optimal bundles, but that an
$O(\deltac)$ fraction of the goods do not approximately clear. A good can fail
to approximately clear for two reasons: it can \emph{over-clear} if the amount
of the good that is bought is significantly greater than $1$, and it can
\emph{under-clear} if the amount bought is significantly less than $1$. 

In Step 3, we aim to address the goods that over-clear. We do this by gradually
increasing the prices of the goods while maintaining the invariant that all
buyers continue to receive approximately optimal bundles. This is by far the
most technically involved step of the reduction, and we only give a brief
overview of the ideas here. A more detailed overview is available in
\cref{sec:step3overview}.

Let $g$ be some good that over-clears. We want to increase the price of $g$ to
reduce its over-clearing. Consider a buyer $b$ that is
approximately \emph{marginally indifferent} between $g$ and some other good $g'$, meaning
that if we were to give them a small amount of extra money, then they would be
indifferent between allocating it to $g$ or $g'$ in an approximate equilibrium. If $b$ allocates money to
$g$, then raising the price of $g$ by a significant amount would break the tie and cause $b$
to want to allocate all of their money to $g'$ instead. So just raising the
price of $g$ alone is not a good idea: we may simply transfer all of the
over-clearing at $g$ to the goods that are adjacent to it.

Instead, we construct a graph $G$ (called the \emph{restricted flow graph}) over the goods and buyers in which goods $g$ and $g'$ are connected by a path if
there is a buyer that is approximately marginally indifferent between $g$ and
$g'$. Then we take the set of all over-clearing goods and every good that is
connected to them in the graph $G$ and simultaneously \emph{pump}, that is multiply, the prices of all of
them by some common factor~$c > 1$. This means that a buyer that was approximately
indifferent between two goods $g$ and $g'$ before remains approximately
indifferent afterward. 

After the pump, the buyers may no longer have approximately optimal bundles,
and so we fix this by shifting the allocations of affected buyers. We call this
combination the \emph{pump-and-shift} procedure. The main cost of
pump-and-shift is that goods that are adjacent to a pumped good in the market
but not in the graph $G$ may under-clear after this operation.

So pump-and-shift can introduce under-clearing. This is fine \emph{so long as
the set of goods that we pump is not too large}: the market has constant
degree $d$, so if we pump an $O(\deltac)$ fraction of the goods, then
pump-and-shift will introduce under-clearing in an $O(\deltac)$ fraction of the goods. 

Unfortunately, if we just built the graph $G$ and ran pump-and-shift,
then there is nothing stopping us from pumping every good in the market, which
would then destroy our current solution. So instead we must carefully restrict
the size of the graph $G$ before we pump it, and we do this by introducing the
concept of \emph{burning}. We permit each buyer to temporarily burn a small
proportion of their budget instead of spending it on goods. This will of course
cause them to lose utility, but so long as the amount burned is suitably small,
we can do this while still ensuring that all buyers have approximately optimal
bundles. 

Before we run pump-and-shift, we first run the \emph{shift-and-burn} procedure.
This procedure shifts money away from over-clearing goods to buyers in $G$ who
then burn the money instead, and it is implemented by solving a network flow
problem over the graph $G$. The key property is that $G$ cannot be large after
shift-and-burn has executed: we carefully construct $G$ so that every surviving
good in the graph has a buyer that fully burns, so if
the graph were large, then it would contain enough fully-burning buyers to
entirely burn the over-clearing that we are trying to reduce.

So the overall algorithm in Step 3 is to alternate shift-and-burn and
pump-and-shift. The shift-and-burn procedure ensures that the graph $G$ is
small, and then the pump-and-shift procedure increases the prices of all of the
goods in $G$. We show that a constant number of alternations between these two
procedures is enough to remove all over-clearing in the market, while not
significantly increasing the number of under-clearing goods.

\paragraph{\bf Step 4: Fix the under-clearing.}

At this point we have a situation in which all buyers have an approximately
optimal allocation, no good over-clears, an $O(\deltac)$ fraction of the
goods under-clear, and some of the buyers have burned money in Step 3 that should be
reallocated to goods in the market. 

It turns out to be relatively straightforward to fix the under-clearing goods.
We first instruct all buyers that have burned money to spend it on
under-clearing goods. This may give them zero utility, but they already have an
approximately optimal bundle, so this does not harm us. 
After this, we may still have under-clearing goods for which there was not
enough burned money to fix, or we may have fixed all under-clearing goods but
have burned money left over. It turns out that both cases can be dealt with by
a simple global reallocation rule.

\paragraph{\bf Step 5: Obtain exact clearing.}

At this point, we have an $(\eps, \delta)$-approximate market equilibrium for
some constants $\eps$ and $\delta$,
which is what we set out to find. As a bonus, in Step 5 we show that we can in
fact find a $(0, \delta')$-approximate market equilibrium for some constant
$\delta'$, meaning that we find
an approximate market equilibrium in which all goods exactly clear and all
buyers have $\delta'$-approximate optimal bundles. 

As in Step 4, there turns out to be a relatively simple global transformation
that does this. The idea is to first modify the solution so that all goods have
an identical demand $D \in [1 - \eps, 1 + \eps]$, and then to scale the prices
so that this demand is exactly equal to $1$. It is relatively straightforward
to do this while not losing too much from the optimality of the buyers'
bundles. 

We now proceed to give more detailed technical overviews of each of these five steps.

\subsubsection{\texorpdfstring{Step 1: Reduce to $(\epsc, \deltac)$-\gcircuitp}{Step 1: Reduce to \gcircuitp}}

We start by formulating the $(\eps, \delta)$-approximate market equilibrium problem
as an $(\epsc, \deltac)$-\gcircuitp instance for some constants $\epsc$
and~$\deltac$.

The variables of this instance will encode a price vector $p$ that assigns a
price $p_j$ to each good~$j$, and a vector $q$ such that for each buyer $i$,
good $j$, and utility-function segment $k$, the variable $q_{i,j,k}$ gives the
total amount of money that buyer $i$ spends on that utility-function segment.
Setting $x_{i,j} = \sum_k q_{i,j,k}/p_j$ then allows us to recover an allocation
from the $\gcircuitp$ solution.

We show that a system of constraints can be imposed on these variables to
ensure that, in an $(\epsc, \deltac)$-\gcircuitp solution, we have that there
is a constant $c$ such that the following properties hold.
\begin{itemize}
\item A $(1 - c \cdot \deltac)$-fraction of the goods $(c \cdot \epsc)$-clear.
\item A $(1 - c \cdot \deltac)$-fraction of the buyers have $(1 - c \cdot
\epsc)$-optimal allocations and satisfy their budget constraints with equality.
\end{itemize}
The fact that only a 
$(1 - c \cdot \deltac)$-fraction of the goods and buyers satisfy their
constraints arises from the fact that a $\deltac$-fraction of the constraints
in our \gcircuitp instance fail. The number of constraints used for each buyer is
proportional to the number of non-zero utility-function segments that the buyer
has, which is a constant in a reducible market. Likewise, the number of
constraints used for each good is proportional to the non-zero utility-function
segments for that good, which is also constant in a reducible market. 
So the $\deltac$-fraction of constraints that fail then translate to a 
$(c \cdot \deltac)$-fraction of buyers and goods for which at least one
constraint fails for some constant $c$.

\subsubsection{Step 2: Fix the broken buyers}

In Step 2 we address the $(c \cdot \deltac)$-fraction of buyers that do not
receive a $(1 - c \cdot \epsc)$-optimal allocation at the end of Step 1. We do
this simply by making those buyers buy an optimal allocation at the
current price vector $p$, while leaving all other buyers unchanged. This then
ensures that all buyers receive a $(1 - c \cdot \epsc)$-optimal allocation, and
that all buyers satisfy their budget constraints with equality.

The cost of doing this is that any buyer who changes their allocation in Step 2
may affect the clearing constraint of the goods that they move money from or
to. This is fine, however, because the degree $d$ of the market is constant,
and at most a $(c \cdot \deltac)$-fraction of the buyers change their
allocation during Step 2. So at most a 
$(c \cdot \deltac)$-fraction of the goods violated their clearing constraints
at the start of Step 2, and at most $(d \cdot c \cdot \deltac)$-fraction of the
goods have their clearing constraints broken by the shifting of allocations in
Step 2. 

To summarize, at the end of Step 2, we have a price vector $p$, a revised allocation
$x$, and a constant $c$ (larger than the constant used in Step 1) such that the
following hold.
\begin{itemize}
\item A $(1 - c \cdot \deltac)$-fraction of the goods $(c \cdot \epsc)$-clear.
\item All buyers have a 
$(1 - c \cdot \epsc)$-optimal allocation and satisfy their budget constraint with equality.
\end{itemize}

\subsubsection{Step 3: Fix the over-clearing}
\label{sec:step3overview}

In Step 3 we receive a price vector $p$, an allocation $x$ and a constant $c$
such that a $(1 - c \cdot \deltac)$-fraction of the goods $(c \cdot
\epsc)$-clear. We say that a good $j$ \emph{over-clears} if $\sum_{i} x_{i,j} >
1 + c \cdot \epsc$, and we say that a good \emph{under-clears} if $\sum_{i}
x_{i,j} < 1 - c \cdot \epsc$. In Step 3, our goal is to remove all of the
over-clearing in the current allocation, and this is by far the most
technically involved step.

\paragraph{\bf Burning.}

In Step 3 we will temporarily allow each buyer to \emph{burn} up to $\epsc$
money. When a buyer burns money, they do not spend it on any goods, and they
instead just destroy it. Technically, this is implemented by loosening the
budget constraint so that buyer $i$ is required to spend at least $e_i - \epsc$
and at most $e_i$ money. 

This is only a temporary measure. We will use burning as a technical tool to
help us reduce over-clearing goods in Step 3, and we will later make use of the
burned money in Step 4 to fix the under-clearing goods, and thereby restore the
budget constraints with equality for all buyers. 

\paragraph{\bf The Step 3 invariant.}

Throughout Step 3, we maintain the following invariant. For each buyer $i$, we
define $\marbpb_i(p,x)$ to be the bang-per-buck of the most desirable
utility-function segment (under prices $p$) that buyer $i$ does not fully buy.
We say that buyer $i$ satisfies the invariant if they do not spend any money on
segments that have bang-per-buck that is strictly worse than 
$\marbpb_i(p,x)/(1+c_3 \cdot \epsc)$, where $c_3$ is some constant. It is relatively easy to show that any
buyer that satisfies the invariant will receive an approximately optimal
allocation, and so we will maintain this invariant for all buyers throughout Step
3.

We also define the \emph{window} for each buyer to be the interval
$[\marbpb_i(p,x)/(1+c_3 \cdot \epsc), \marbpb_i(p,x)]$. We say that a
utility-function segment is \emph{in the window } if its bang-per-buck lies in
the window interval. An important property of the window is that a buyer can
shift money between utility-function segments that lie within the window
without affecting the invariant.

Although we do not necessarily satisfy the invariant at the start of Step 3,
the properties that we prove for Step 1 ensure that each buyer buys their
segments in approximate bang-per-buck order, and we show that a straightforward
preprocessing step can use this to transform the allocation at the end of Step
2 into an allocation that meets the invariant, without significantly altering
the clearing constraints of the goods. 

\paragraph{\bf Shift-and-burn.}

Our procedure for Step 3 involves alternating two algorithms, the first of which
we call shift-and-burn. The purpose of this algorithm is to shift money from
the over-clearing goods on to buyers, who then burn that money.

\begin{figure}
\begin{center}
\begin{subfigure}{0.4\textwidth}
    \begin{center}
    \begin{tikzpicture}
    \node at (0, 1.5) (g1upper) {};
    \node[circle,draw] at (0, 0) (g1) {$g_1$};
    \node[circle,draw] at (2, 0) (g2) {$g_2$};

    \node[draw, inner sep=0.25cm] at (1, -2) (b1) {$b_1$};
    \node[draw, inner sep=0.25cm] at (3, -2) (b2) {$b_2$};
    \node at (1, -3) (b1lower) {};

    \path[->,draw]
        (b1) edge[left] node {\footnotesize $\frac{0.6}{1}$} (g1)
        (b1) edge[left] node {\footnotesize $\frac{0.2}{0.4}$} (g2)
        (b2) edge[right] node {\footnotesize $\frac{0.1}{1}$} (g2)
        ;
    \end{tikzpicture}
    \end{center}
    \caption{Two buyers with in-window segments to two goods, where good
$g_1$ over-clears.}
    \label{fig:flowmarket}
\end{subfigure}
\hskip 0.1\textwidth
\begin{subfigure}{0.4\textwidth}
    \begin{center}
    \begin{tikzpicture}
    \node at (-1, 1) (g1upper) {s};
    \node[circle,draw] at (-1, 0) (g1) {$g_1$};
    \node[circle,draw] at (3, 0) (g2) {$g_2$};

    \node[draw, inner sep=0.25cm] at (1, -2) (b1) {$b_1$};
    \node[draw, inner sep=0.25cm] at (5, -2) (b2) {$b_2$};
    \node at (1, -3) (b1lower) {t};
    \node at (5, -3) (b2lower) {t};

    \path[->,draw]
        (g1upper) edge[right] node {} (g1)
        (b1) edge[bend left, left] node {\scriptsize $0.4 p_{g_1}$} (g1)
        (g1) edge[bend left, right] node {\scriptsize $0.6 p_{g_1}$} (b1)
        (b1) edge[bend left, right] node {\scriptsize $0.2 p_{g_2}$} (g2)
        (g2) edge[bend left, left] node {\scriptsize $0.2 p_{g_2}$} (b1)
        (b2) edge[bend left, right] node {\scriptsize $0.9 p_{g_2}$} (g2)
        (g2) edge[bend left, right] node {\scriptsize $0.1 p_{g_2}$} (b2)
    
        (b1) edge[right] node {\footnotesize $\epsc$} (b1lower)
        (b2) edge[right] node {\footnotesize $\epsc$} (b2lower)
        ;
    \end{tikzpicture}
    \end{center}
    \caption{The resulting flow graph.}
    \label{fig:flowgraph}
\end{subfigure}
\caption{The construction of the flow graph.}
\end{center}
\end{figure}

In~\cref{fig:flowmarket} we show a simple scenario in which there is a
good $g_1$ that over-clears, and a good $g_2$ that does not over- or
under-clear. The two buyers in the example have utility-function
segments to the goods that all lie in their buyer's windows, which are shown as
edges in the figure. The labels of the edges show how much of those utility
function segments is bought in the current allocation. So for example, buyer
$b_1$ has an in-window utility function segment to good $g_1$ with length
$1$, and is currently buying $0.6$ units of that segment, and another
in-window segment to good $g_2$, of which $0.2$ of the maximum $0.4$
units are currently bought. 

Note that buyer $b_1$ can \emph{shift} up to $\epsc$ money from $g_1$ and burn it
instead, and doing so reduces the over-clearing at $g_1$, since fewer units of
$g_1$ would now be bought. Buyer $b_1$ can also shift $\epsc$ money from $g_1$ to
$g_2$, which reduces the over-clearing at $g_1$ while increasing the amount of
money spent on $g_2$. If, at the same time, $b_2$ then shifts $\epsc$ money from
$g_2$ and burns it instead, then the over-clearing at $g_1$ is reduced, while
the clearing constraint at $g_2$ remains unaffected.

We view this shifting process as a flow problem. 
\begin{itemize}
\item The source $s$ of the flow problem has an edge to each over-clearing good
$j$, with the capacity of that edge being the amount of money that would need
to be removed from good $j$ to stop it over-clearing.

\item Each buyer has an edge to the target $t$ of the flow problem with
capacity $\epsc$, which represents the amount of money that they can burn.
In later iterations, when the buyer may have already burned some money, this
capacity will be reduced to ensure that no buyer burns more than $\epsc$
money.

\item For each utility-function segment for buyer $i$ and good
$j$, if the segment lies in buyer $i$'s window, then we add an edge from $j$ to $i$ whose capacity is the total amount of money
that is currently being spent on that utility-function segment, and an edge
from $i$ to $j$ whose capacity is equal to the total amount of extra money that
could be spent on that utility-function segment before it is fully used.
\end{itemize}
\cref{fig:flowgraph} shows the flow problem associated with the scenario
from~\cref{fig:flowmarket}. Note that the edge from $g_1$ to $b_1$ has
capacity $0.6 \cdot p_{g_1}$, since that is the total amount of money that
$b_1$ is currently spending on $g_1$ for this utility-function segment, and
thus the largest amount of money that $b_1$ could shift off this utility-function
segment. Meanwhile, the edge from $b_1$ to $g_2$ has capacity $0.2 \cdot
p_{g_2}$ since that is the total amount of money that $b_1$ could shift on to
that utility-function segment without exhausting it. 

Computing a maximal $s$-$t$ flow then tells us the maximum amount of money that
we can shift from the over-clearing goods to money that is burned by the
buyers. In particular, since we have a flow, any intermediate goods through
which money is routed do not have their clearing constraints changed, since
exactly as much money is spent on them after applying the flow as was
beforehand. 

We will show that applying the shift-and-burn procedure maintains the
invariant. Indeed, by construction, the flow is only able to change the amount
of money spent on segments that lie within their buyer's window. However, it is
possible that $\marbpb_i(p,x)$ changes after we apply the flow, which occurs
whenever we fully buy all of the segments that defined $\marbpb_i(p,x)$ before
the flow is applied. This changes the windows that we must consider in the new
allocation, but we will show that this is not an issue, and that the
invariant is maintained even if $\marbpb_i(p,x)$ changes for some buyer $i$. It
will however be technically convenient to re-run shift-and-burn whenever this
occurs, and we show that only a polynomial number of iterations are
required to reach a shift-and-burn step in which no buyer's window changes.

\paragraph{\bf The restricted flow graph.}

After the last shift-and-burn step, we then construct the \emph{restricted flow
graph}, which is the result of applying the following iterative algorithm to
the flow graph that was considered in the last shift-and-burn iteration. If
$f$ is the maximal $s$-$t$ flow that we computed in the flow graph, then we do
the following.
\begin{enumerate}
\item Add all over-clearing goods to the restricted flow graph.

\item If $j$ is a good in the restricted flow graph, and there is an edge from
$j$ to buyer $i$ that is not saturated by $f$, then we add buyer $i$ to the
restricted flow graph.

\item If $i$ is a buyer in the restricted flow graph, and there is an edge from
$i$ to good $j$ that is not saturated by $f$, then we add good $j$ to the
restricted flow graph.
\end{enumerate}
The restricted flow graph is obtained by repeating points 2 and 3 above until
convergence.

An important property of the restricted flow graph is that, if there are still
over-clearing goods in the market, and if a $(c
\cdot \deltac)$-fraction of the goods are over-clearing for some constant $c$,
then the restricted flow graph contains at most 
a $(c' \cdot \deltac/\epsc)$-fraction of the goods for some constant $c'$. 
This follows from the following two observations.
\begin{itemize}
\item Every buyer in the restricted flow graph fully burns, i.e., they burn
exactly $\epsc$ money. This can easily be shown by contradiction: if a buyer
$i$ in the restricted flow graph did not
fully burn then by construction there is a path in the restricted flow graph
from $s$ to $i$ that only uses unsaturated edges, so we could increase the flow
on those edges in order to burn more at $i$, which would contradict the
maximality of $f$.

\item The total amount of money spent on an over-clearing good is at most some
constant $M$ because each buyer's budget is less than $\bmax$, which is
constant, and since the market has constant degree.
\end{itemize}
So at most $M \cdot (c \cdot \deltac) \cdot |G|$ money is spent on over-clearing goods, and
therefore if the restricted flow graph contains more than 
$$\frac{M \cdot c \cdot \deltac}{\epsc} \cdot |G|$$ 
goods, then all of the over-clearing must have been burned,
meaning that there are no longer any over-clearing goods.

\paragraph{\bf Pump-and-shift.}

The second procedure that we use in Step 3 is called the \emph{pump-and-shift}
procedure. Here we 
\emph{pump} the prices of all goods in the restricted flow graph by multiplying
them by $(1 + \epsc)$. 

After the pump, we then need to repair the allocation by shifting money. In
particular, if a utility-function segment $(j,k)$ was above its buyer's window,
then buyer $i$ will need to spend more money in order to continue to
fully buy that segment and satisfy the invariant. 

To address this, we simply instruct each buyer who is interested in at least
one good that was pumped to shift money from their worst bang-per-buck
utility-function segments on to their best non-fully-bought utility-function
segments, and we order this process so we first shift on to the segments that
have the highest bang-per-bucks. We will show that, after this shifting
operation, all buyers will continue to satisfy the invariant.

The pump-and-shift step
does come with another cost, however, which is that as the buyers shift,
they may introduce under-clearing in goods that are adjacent to the restricted
flow graph. That is, if a buyer shifts money from a segment for their worst
bang-per-buck good $j$ onto another segment for a pumped good
$j'$, then the clearing constraint for $j$ may be affected, and in particular
we may have that $j$ under-clears even if it was not under-clearing beforehand. 

Here our bounds on the size of the restricted flow graph play a crucial role.
The scenario described above only occurs when there is a buyer $i$ that is
interested in goods $j$ and $j'$, and $j'$ is in the restricted flow graph.
Since the market has constant degree $d$, and since at most a $(c' \cdot
\deltac / \epsc)$-fraction of the goods lie in the restricted flow graph for
some constant $c'$, we get that at most a $(d^2 \cdot c' \cdot \deltac /
\epsc)$-fraction of the goods can under-clear as a result of the pump step,
which keeps the under-clearing small enough for our purposes.

\paragraph{\bf Combining into Step 3.}

Step 3 consists of alternating shift-and-burn and pump-and-shift until no
over-clearing remains. We prove that this will terminate after a constant number
of rounds. The argument for this uses the fact that all goods have prices
lying in the range $[\pmin, \pmax]$ for some constants $\pmin$ and $\pmax$
after
after Step 1. We can multiply $\pmin$ by $(1 + \epsc)$ at most
constantly many times before it exceeds $\pmax$, because $\pmin$ and $\pmax$
are both constant, and $\epsc$ is also constant. We then
argue that if a good has price
$\pmax$ then it cannot possibly over-clear, because $\pmax$ was selected to be
so high that even if all buyers that are interested in a good spend all of
their money on it, the good still cannot over-clear. 
Thus, after constantly many rounds all over-clearing goods will have been
pumped to the point where they no longer over-clear. 

At the end of Step 3, we arrive at a price vector $p$, an allocation $x$, and a
constant $c$ (again, larger than the constant from Step 2) such that
\begin{itemize}
\item Every buyer receives a $(1 - c \cdot \epsc)$-optimal allocation, and
burns at most $\epsc$ money.
\item At least $(1 - c \cdot (\deltac/\epsc))$-fraction of the goods $(c \cdot
\epsc)$-clear, and no good over-clears.
\end{itemize}

\subsubsection{Step 4: Fix the under-clearing}

There are only two tasks left at this point: fix the under-clearing and deal
with the money that was burned in Step 3. We can in fact use one of these
problems to fix the other. We instruct each buyer that has burned money to
spend that money on goods that under-clear. In general this will
deliver zero extra utility to those buyers, since they may not be interested in
those goods, but we already have that these buyers are given 
$(1 - c \cdot \epsc)$-optimal allocations, so this does not harm us. 

There may be too much or too little burned money to do this. We use the
following procedures to deal with this.
\begin{itemize}
\item If there is too much burned money, then after fixing all of the
under-clearing, we then instruct the buyers to spend their money equally on
each of the goods in the market. This will push up the demand of all goods
slightly, but we will show bounds on the amount of money burned in Step 3,
which will be sufficient to show that all goods $(c' \cdot \epsc)$-clear after
this operation, for some constant $c' > c$. 

\item If there is too little burned money, then after we exhaust the burn pool,
we then instruct all buyers in the market to shift a small amount of money off
their worst bang-per-buck utility segments, and then to use that money to clear
the under-clearing goods. We will show bounds on the amount of money needed to
clear the under-clearing goods that are sufficient to argue that only a small
amount of money needs to be shifted in this way. In particular, each buyer will
still have a $(1 - g(\epsc, \deltac))$-optimal allocation after this operation for
some function $g$, where $\epsc$ and $\deltac$ now
appear in this bound because at the end of Step 3 the fraction of under-clearing
goods depended on both $\epsc$ and $\deltac$.

Likewise, since the market has constant degree, and
each buyer shifts only a small amount of money, in the worst case the total
amount of money shifted off each good during this operation is still small, and
we maintain that each good $f(\epsc,\deltac)$-clears for some
function~$f$.
\end{itemize}

At the end of Step 4 we have a price vector $p$, an allocation $x$ such that the following hold.
\begin{itemize}
\item Every buyer receives a $(1 - g(\epsc, \deltac))$-optimal allocation, and
satisfies their budget constraint with equality.
\item All goods $f(\epsc, \deltac)$-clear.
\end{itemize}

\subsubsection{Step 5: Obtain exact clearing}

With the results from Step 4, by picking $\epsc$ and $\deltac$ to be suitably
small constants,  it is possible to show that if it is \ppad/-hard to find a
$(\eps, \delta)$-approximate market equilibrium of a reducible market for some
constants $\eps$ and $\delta$, then the \pcppad conjecture holds. We can
however show the stronger theorem that if it is 
\ppad/-hard to find a
$(0, \delta)$-approximate market equilibrium of a reducible market for some
constant $\delta$, then the \pcppad conjecture holds. We do this in Step 5,
where we transform the allocation $x$ and price vector $p$ from Step 4 so that
all goods exactly clear, while ensuring that not too much is lost from the
buyer's optimality during the process.

We do this in a three step process.
\begin{enumerate}
\item The first step is to remove money from the allocation so that every good
has demand exactly equal to $1 - f(\epsc, \deltac)$. Note that Step 4
ensures that all goods $f(\epsc, \deltac)$-clear, so this can be
achieved by removing at most $2 \cdot f(\epsc, \deltac)$ demand from each good.
Since $f(\epsc, \deltac)$ will be chosen to be small, we show that each buyer
does not lose too much of their utility during this operation, although they
are now underspending their budget. 

\item Next, we take the money that was removed in the previous step, and we
instruct all buyers to distribute it evenly across all of the goods in such a
way as to ensure that the demand of all goods is equal to some value $D \in [1
- f(\epsc, \deltac), 1 + f(\epsc, \deltac)]$. This can only increase the
  buyer's utilities, so the optimality of their bundles is not affected, and
this also reestablishes the budget constraint (with equality) of each of the buyers.

\item Finally, we multiply the prices by a factor of $D$, while keeping the
amount of money spent on each good the same. This normalizes the demand of each
good to 1, giving us the exact clearing that we desire. This step can affect
the optimality of the bundles, but we show that since 
$D \in [1 - f(\epsc, \deltac), 1 + f(\epsc, \deltac)]$, and since $f(\epsc,
\deltac)$ is very small, the bundles remain approximately optimal.
\end{enumerate}
At the end of this process we arrive at an allocation $x'$ and a price vector
$p'$ such that the following hold.
\begin{itemize}
\item Every good exactly clears.
\item Every buyer receives a $1 - g'(\epsc, \deltac)$-optimal bundle for some
function $g'$. 
\end{itemize}
As a final step, we then choose constant values for $\epsc$ and $\deltac$ to
ensure that 
$g'(\epsc, \deltac) \le \delta$, 
which then allows us to recover
the $(0, \delta)$-approximate market equilibrium that we seek.

\section{Hardness for Fisher Markets from the \pcppad Conjecture}
\label{sec:hardness}

In this section we prove that, if we assume the \pcppad conjecture, then, for any $\eps < 1/9$, there exist some sufficiently small constant $\delta > 0$ such that computing an $(\eps, \delta)$-approximate market equilibrium in a Fisher market with SPLC utilities is \ppad/-hard.

Our reduction also yields two further results that hold unconditionally, without needing to assume the \pcppad conjecture. Namely, we obtain \ppad/-hardness for the problem, if we either
\begin{itemize}
    \item set $\delta > 0$ to be inverse-polynomial instead of constant, or
    \item do not allow buyers to spend any money on goods for which their utility function is the zero function.
\end{itemize}

Furthermore, all of our hardness results hold for a class of Fisher markets that are particularly ``simple'', in the following sense.

\begin{definition}\label{def:simple-condition}
A family of Fisher markets is \emph{simple} if the utility functions are SPLC \emph{linear capped}, all buyers have budget $e_i = 1$ (CEEI), and there exist constants such that for any market in the family:
\begin{itemize}
    \item Every buyer $i$ has non-zero utility function $u_{ij}$ for at most a constant number of goods $j$.
    \item For every good $j$, the utility function $u_{ij}$ is non-zero for at most a constant number of buyers $i$.
    \item For every buyer $i$, the non-zero slopes in the utility functions $u_{ij}$ have ratio bounded by a constant.
\end{itemize}
Furthermore, the sufficient condition for existence of equilibrium holds, i.e., for every buyer $i$, there exists at least one good $j$ such that $i$ is not satiated with good $j$ (in other words, the function $u_{ij}$ is linear (uncapped) with strictly positive slope).
\end{definition}

\begin{theorem}\label{thm:fisher-pcp-hard}
Fix any constant $\eps < 1/9$. Assuming the \pcppad conjecture, there exists a constant $\delta > 0$ such that finding an $(\eps,\delta)$-approximate market equilibrium in \emph{simple} Fisher markets is \ppad/-hard.
\end{theorem}

The reduction establishing this theorem also yields the following two unconditional results.

\begin{theorem}\label{thm:fisher-inv-poly-hard}
Fix any constant $\eps < 1/9$. For inverse-polynomial $\delta > 0$ given as part of the input, finding an $(\eps,\delta)$-approximate market equilibrium in \emph{simple} Fisher markets is \ppad/-hard.
\end{theorem}

\begin{theorem}\label{thm:fisher-non-zero-hard}
Fix any constant $\eps < 1/9$. Then, there exists a constant $\delta > 0$ such that finding an $(\eps,\delta)$-approximate market equilibrium in \emph{simple} Fisher markets is \ppad/-hard, \emph{if buyers are not allowed to spend any money on goods for which their utility function is the zero function}.
\end{theorem}

These three theorems also hold for the \emph{Arrow-Debreu exchange market} setting, where buyers do not have a budget and are instead endowed with goods. This follows by a known direct reduction; see \cref{sec:app:exchange}.

\subsection{\texorpdfstring{Proof of the Hardness Theorems (\cref{thm:fisher-pcp-hard,thm:fisher-inv-poly-hard,thm:fisher-non-zero-hard})}{Proof of the Hardness Theorems}}

We prove \cref{thm:fisher-pcp-hard,thm:fisher-inv-poly-hard,thm:fisher-non-zero-hard} using essentially the same reduction. For the sake of clarity, we focus on proving \cref{thm:fisher-pcp-hard}, the main result of this section. In view of \cref{thm:pcp-equivalent}, \cref{thm:fisher-pcp-hard} follows from the following.

\begin{proposition}\label{prop:fisher-pcp-hard}
Fix any constant $\epsm < 1/9$. For any constant $\deltac > 0$, there exists a constant $\deltam > 0$ such that $\deltac$-\pcircuit reduces in polynomial time to the problem of finding an $(\epsm,\deltam)$-approximate market equilibrium in a family of \emph{simple} Fisher markets.
\end{proposition}

The rest of this section is devoted to the proof of this proposition. At the end of the proof, we explain the minor modifications in the argument that yield \cref{thm:fisher-inv-poly-hard,thm:fisher-non-zero-hard}.

\subsubsection{Construction of the market}

Fix any constant $\epsm < 1/9$.

Let $\deltac \in (0,1)$ and consider an instance $(V,C)$ of \pcircuit consisting of \NAND and \PURIFY gates, and such that every node is the input to exactly one gate (we can assume this due to \cref{thm:pcp-equivalent}). We wish to construct a Fisher market that is \emph{simple} (in the sense of \cref{def:simple-condition}) and such that for some sufficiently small constant $\deltam \in (0,1)$ (that is allowed to depend on $\deltac$), we have that any $(\epsm, \deltam)$-approximate market equilibrium yields an assignment that satisfies at least a fraction $(1 - \deltac)$ of the gates of the \pcircuit instance.

We begin by describing a primitive that will be used in our construction. All of the buyers will be of the following form.

\paragraph{\bf Inverter buyers.} An \inverter buyer is described by specifying an input good, denoted \ing, an output good, denoted \outg, and a parameter $\inthres \in [0,1]$. The buyer has budget $1$ and the following separable linear capped utility function:
\begin{itemize}
\item For the output good \outg, the buyer has utility $u_{\inverter, \outg}(x) = x$.
\item For the input good \ing, the buyer has utility
\begin{equation*}
u_{\inverter,\ing}(x) = \begin{cases}
\inslope \cdot x & \text{if $x \le \inthres$,} \\
\inslope \cdot \inthres & \text{otherwise}
\end{cases}
\end{equation*}
where $a > 1$ will be set to be a sufficiently large constant later. 
\item For all other goods, the buyer has zero utility.
\end{itemize}
See \cref{fig:inverter} for an illustration.

\begin{figure}
\begin{center}
\begin{tikzpicture}

\tikzset{
    inputedge/.style={thick,densely dashed,->},
    outputedge/.style={thick,->}
}

% Goods
\node[circle,draw,thick,minimum size=1.3cm] (ing)  at (0, 0) {$\ing$};
\node[circle,draw,thick,minimum size=1.3cm] (outg) at (7, 0) {$\outg$};

% Inverter buyer
\node[draw, inner sep=0.35cm, thick] (inv) at (3.5, 0) {\inverter};

% Input-good edge
\draw[inputedge]
    (inv.west) -- node[below,pos=0.55] {$t$} (ing.east);

% Output-good edge
\draw[outputedge]
    (inv.east) -- (outg.west);

\end{tikzpicture}
\end{center}
\caption{An \inverter buyer with parameter $t$.}
\label{fig:inverter}
\end{figure}

\noindent Next, we construct our market using this primitive. For each node $u \in V = [n]$ of the \pcircuit instance, we construct a corresponding good, which we also denote $u$ with a slight abuse of notation. For every gate $g = (T,u,v,w) \in C$, we introduce some buyers, and possibly an additional good, as described next.

\paragraph{\bf \NAND gates.} For every gate $g = (\NAND, u, v, w)$, i.e., every \NAND gate with inputs $u$ and $v$, and output $w$, we introduce the following \inverter buyers:
\begin{itemize}
\item Four \inverter buyers $\inv_{uw}^1, \inv_{uw}^2, \inv_{uw}^3, \inv_{uw}^4$, each with input good $\ing = u$, output good $\outg = w$, and parameter $t = 2/9$.
\item Similarly, four \inverter buyers $\inv_{vw}^1, \inv_{vw}^2, \inv_{vw}^3, \inv_{vw}^4$, each with input good $\ing = v$, output good $\outg = w$, and parameter $t = 2/9$.
\end{itemize}
This gate does not introduce any additional goods. See \cref{fig:NAND-inverters} for an illustration.

\begin{figure}
\begin{center}
\begin{tikzpicture}

\tikzset{
    inputedge/.style={thick,densely dashed,->},
    outputedge/.style={thick,->}
}

% Goods
\node[circle,draw,thick,minimum size=1.3cm] (u) at (0,  2.1) {$u$};
\node[circle,draw,thick,minimum size=1.3cm] (v) at (0, -2.1) {$v$};
\node[circle,draw,thick,minimum size=1.3cm] (w) at (7, 0) {$w$};

% Inverter buyers from u to w
\node[draw, inner sep=0.22cm, thick] (invuw1) at (2.9,  3.6) {$\inv_{uw}^1$};
\node[draw, inner sep=0.22cm, thick] (invuw2) at (2.9,  2.6) {$\inv_{uw}^2$};
\node[draw, inner sep=0.22cm, thick] (invuw3) at (2.9,  1.6) {$\inv_{uw}^3$};
\node[draw, inner sep=0.22cm, thick] (invuw4) at (2.9,  0.6) {$\inv_{uw}^4$};

% Inverter buyers from v to w
\node[draw, inner sep=0.22cm, thick] (invvw1) at (2.9, -0.6) {$\inv_{vw}^1$};
\node[draw, inner sep=0.22cm, thick] (invvw2) at (2.9, -1.6) {$\inv_{vw}^2$};
\node[draw, inner sep=0.22cm, thick] (invvw3) at (2.9, -2.6) {$\inv_{vw}^3$};
\node[draw, inner sep=0.22cm, thick] (invvw4) at (2.9, -3.6) {$\inv_{vw}^4$};

% Input-good edges
\draw[inputedge] (invuw1.west) -- (u);
\draw[inputedge] (invuw2.west) -- (u);
\draw[inputedge] (invuw3.west) -- (u);
\draw[inputedge] (invuw4.west) -- (u);

\draw[inputedge] (invvw1.west) -- (v);
\draw[inputedge] (invvw2.west) -- (v);
\draw[inputedge] (invvw3.west) -- (v);
\draw[inputedge] (invvw4.west) -- (v);

% Output-good edges
\draw[outputedge] (invuw1.east) -- (w);
\draw[outputedge] (invuw2.east) -- (w);
\draw[outputedge] (invuw3.east) -- (w);
\draw[outputedge] (invuw4.east) -- (w);

\draw[outputedge] (invvw1.east) -- (w);
\draw[outputedge] (invvw2.east) -- (w);
\draw[outputedge] (invvw3.east) -- (w);
\draw[outputedge] (invvw4.east) -- (w);

\end{tikzpicture}
\end{center}
\caption{\NAND gate gadget. All \inverter buyers have $t = 2/9$.}
\label{fig:NAND-inverters}
\end{figure}

\paragraph{\PURIFY gates.} For every gate $g = (\PURIFY, u, v, w)$, i.e., every \PURIFY gate with input $u$ and outputs $v$ and $w$, we introduce a new good $\aux_g$, as well as the following \inverter buyers:
\begin{itemize}
\item Four \inverter buyers $\inv_{u\aux}^1, \inv_{u\aux}^2, \inv_{u\aux}^3, \inv_{u\aux}^4$, each with input good $\ing = u$, output good $\outg = \aux_g$, and parameter $t = 2/9$.
\item Two \inverter buyers $\inv_{\aux v}^1, \inv_{\aux v}^2$, each with input good $\ing = \aux_g$, output good $\outg = v$, and parameter $t = 2/9$.
\item One \inverter buyer $\inv_{\aux w}$, with input good $\ing = \aux_g$, output good $\outg = w$, and parameter $t = 4/9$.
\end{itemize}
See \cref{fig:PURIFY-inverters} for an illustration.
We will usually drop the subscript $g$ from $\aux_g$, as it will be clear from the context.

\begin{figure}
\begin{center}
\begin{tikzpicture}

\tikzset{
    inputedge/.style={thick,densely dashed,->},
    outputedge/.style={thick,->}
}

% Goods
\node[circle,draw,thick,minimum size=1.3cm] (u)   at (0,  2.4) {$u$};
\node[circle,draw,thick,minimum size=1.3cm] (aux) at (7,  2.4) {$\aux$};
\node[circle,draw,thick,minimum size=1.3cm] (v)   at (14, 3.2) {$v$};
\node[circle,draw,thick,minimum size=1.3cm] (w)   at (14, 1.6) {$w$};

% Inverter buyers from u to aux
\node[draw, inner sep=0.22cm, thick] (invuaux1) at (3.2,  3.9) {$\inv_{u\aux}^1$};
\node[draw, inner sep=0.22cm, thick] (invuaux2) at (3.2,  2.9) {$\inv_{u\aux}^2$};
\node[draw, inner sep=0.22cm, thick] (invuaux3) at (3.2,  1.9) {$\inv_{u\aux}^3$};
\node[draw, inner sep=0.22cm, thick] (invuaux4) at (3.2,  0.9) {$\inv_{u\aux}^4$};

% Inverter buyers from aux to v
\node[draw, inner sep=0.22cm, thick] (invauxv1) at (10.7, 3.7) {$\inv_{\aux v}^1$};
\node[draw, inner sep=0.22cm, thick] (invauxv2) at (10.7, 2.7) {$\inv_{\aux v}^2$};

% Inverter buyer from aux to w
\node[draw, inner sep=0.22cm, thick] (invauxw) at (10.7, 1.3) {$\inv_{\aux w}$};

% Input-good edges
\draw[inputedge] (invuaux1.west) -- (u);
\draw[inputedge] (invuaux2.west) -- (u);
\draw[inputedge] (invuaux3.west) -- (u);
\draw[inputedge] (invuaux4.west) -- (u);

\draw[inputedge] (invauxv1.west) -- (aux);
\draw[inputedge] (invauxv2.west) -- (aux);
\draw[inputedge] (invauxw.west)  -- node[below,sloped,pos=0.55] {\scriptsize $t=4/9$} (aux);

% Output-good edges
\draw[outputedge] (invuaux1.east) -- (aux);
\draw[outputedge] (invuaux2.east) -- (aux);
\draw[outputedge] (invuaux3.east) -- (aux);
\draw[outputedge] (invuaux4.east) -- (aux);

\draw[outputedge] (invauxv1.east) -- (v);
\draw[outputedge] (invauxv2.east) -- (v);
\draw[outputedge] (invauxw.east)  -- (w);

\end{tikzpicture}
\end{center}
\caption{\PURIFY gate gadget. All \inverter buyers, except $\inv_{\aux w}$, have $t=2/9$.}
\label{fig:PURIFY-inverters}
\end{figure}

This completes the description of the market. Clearly, the market can be constructed in polynomial time given the \pcircuit instance. The value of $\deltam$ and $\inslope$ will be specified at the end of the proof of correctness. Importantly, since $\inslope$ will be constant, the family of Fisher markets constructed by this reduction will be \emph{simple} (in the sense of \cref{def:simple-condition}). In particular, this also relies on the fact that in the \pcircuit instance, every node is the output of exactly one gate, and the input to exactly one gate. We now proceed with the proof of correctness.

\paragraph{\bf Mapping of solutions.}
Given (approximate) market equilibrium prices $p$, we will obtain an assignment $\valonly$ to the nodes of the \pcircuit instance as follows. Let $H = 9/2$ and $L = 100/\inslope$. For any node $u \in V$, we define the assignment as follows:
\begin{itemize}
\item If $p_u \geq H$, then set $\val{u} := 1$.
\item If $p_u \leq L$, then set $\val{u} := 0$.
\item If $p_u \in (L,H)$, then set $\val{u} := \garbo$.
\end{itemize}

In the next section, we prove some general properties of \inverter buyers. These properties will then be used in the subsequent section to argue the correctness of the reduction. Before we proceed, let us note that all prices are non-zero.

\begin{lemma}\label{lem:hard:prices-pos}
In any $(\eps_m, \delta_m)$-approximate equilibrium, all goods have strictly positive prices.
\end{lemma}

\begin{proof}
By construction, for every good $j$ there exists a buyer $i$ such that the utility function $u_{ij}$ is linear with strictly positive slope. This holds for goods $\aux_g$ by construction of the \PURIFY gates, and for all goods corresponding to \pcircuit nodes because every node is the output of a gate. As a result, if $p_j = 0$, the optimal utility value of buyer $i$ is infinite. Thus, the (finite) bundle $x_i$ of buyer $i$ cannot possibly be $(1-\deltam)$-approximately optimal, for any $\deltam < 1$, a contradiction.
\end{proof}

\subsubsection{Properties of \inverter buyers}

In this section we prove some general properties of \inverter buyers. We assume throughout that we are at an $(\epsm, \deltam)$-approximate\footnote{Although our reduction requires $\epsm < 1/9$, in this section we only use the (trivial) bound $\epsm \leq 1$.} equilibrium $(p,x)$, with strictly positive prices.

\begin{lemma}\label{lem:hard:out-lower-spend}
For any \inverter buyer with parameter $\inthres \in [0,1]$, we have $p_\outg \cdot x_{\inverter,\outg} \geq 1 - p_\ing \cdot \inthres - \deltam (1 + \inslope \cdot p_\outg)$.
\end{lemma}

\begin{proof}
Let $z_\outg$ denote the amount of money spent by \inverter on good \outg, i.e., $z_\outg = p_\outg \cdot x_{\inverter,\outg}$. Note that if $z_\outg > 1 - p_\ing \cdot \inthres$, then the claim trivially holds. Thus, in the rest of the proof we assume that $z_\outg \leq 1 - p_\ing \cdot \inthres$, and thus in particular $1 - p_\ing \cdot \inthres \geq 0$.

Let $u_\ing$ denote the utility obtained by \inverter from good \ing, i.e., $u_\ing = \inslope \cdot \min\{\inthres, x_{\inverter, \ing}\}$. Note that $u_\ing \leq \inslope \inthres$. The total utility obtained by \inverter can now be written as
$$\Uactual = u_\ing + z_\outg \frac{1}{p_\outg}.$$
The maximum utility $\Uopt$ achievable by \inverter can be bounded as follows
$$\Uopt \geq u_\ing + (1 - p_\ing \cdot \inthres) \frac{1}{p_\outg}.$$
Indeed, this corresponds to the allocation where we spend $p_\ing \cdot \inthres$ on good \ing (thus obtaining utility $\inslope \inthres \geq u_\ing$), and the remaining money on the good \outg.

By $(1-\deltam)$-optimality of the bundles, we must have $\Uactual \geq (1-\deltam) \Uopt$. Plugging in the bounds for the utilities and grouping terms yields
$$(z_\outg - (1 - p_\ing \cdot \inthres)) \frac{1}{p_\outg} \geq - \deltam u_\ing - \deltam (1 - p_\ing \cdot \inthres) \frac{1}{p_\outg}$$
and thus
$$z_\outg - (1 - p_\ing \cdot \inthres) \geq - \deltam u_\ing p_\outg - \deltam (1 - p_\ing \cdot \inthres) \geq - \deltam (1 + \inslope \cdot p_\outg)$$
where we used $u_\ing \leq \inslope \inthres \leq \inslope$.
\end{proof}

\begin{lemma}\label{lem:hard:out-upper-spend}
For any \inverter buyer with parameter $\inthres \in [0,1]$, if we have $\inslope/p_\ing \geq 2/p_\outg$, then $p_\ing \cdot x_{\inverter, \ing} \geq \min\{1, p_\ing \inthres\} - 2\deltam$ and $p_\outg \cdot x_{\inverter, \outg} \leq \max\{0, 1 - p_\ing \inthres\} + 2\deltam$.
\end{lemma}

\begin{proof}
Let $z_\ing$ denote the amount of money spent by \inverter on good \ing, i.e., $z_\ing = p_\ing \cdot x_{\inverter, \ing}$. If $z_\ing > \min\{1, p_\ing \inthres\}$, then the claim trivially holds, so assume for the rest of the proof that $z_\ing \leq \min\{1, p_\ing \inthres\}$.

The utility $\Uactual$ obtained by \inverter can be bounded as follows
$$\Uactual \leq z_\ing \frac{\inslope}{p_\ing} + (1 - z_\ing) \frac{1}{p_\outg}.$$
On the other hand, the maximum utility $\Uopt$ achievable by \inverter satisfies
$$\Uopt \geq \min\{1, p_\ing \inthres\} \frac{\inslope}{p_\ing} + (1 - \min\{1, p_\ing \inthres\}) \frac{1}{p_\outg}.$$
In the case where $p_\ing \inthres \leq 1$, this corresponds to buying $\inthres$ units of \ing, and then spending the remaining budget on good \outg. In the case where $p_\ing \inthres \geq 1$, this corresponds to spending the whole budget on \ing.

By $(1 - \deltam)$-optimality of the bundles, we must have $\Uactual \geq (1 - \deltam) \Uopt$. Plugging in the bounds for the utilities and grouping terms yields
$$(z_\ing - \min\{1, p_\ing \inthres\}) \left(\frac{\inslope}{p_\ing} - \frac{1}{p_\outg}\right) \geq - \deltam \min\{1, p_\ing \inthres\} \frac{\inslope}{p_\ing} - \deltam (1 - \min\{1, p_\ing \inthres\}) \frac{1}{p_\outg}.$$
Now, dividing by $\inslope/p_\ing - 1/p_\outg$ on both sides, and using the fact that $\inslope/p_\ing - 1/p_\outg \geq \inslope/(2p_\ing) \geq 1/p_\outg$, we obtain
$$z_\ing - \min\{1, p_\ing \inthres\} \geq - 2 \deltam \min\{1, p_\ing \inthres\} - \deltam (1 - \min\{1, p_\ing \inthres\}) \geq - 2 \deltam.$$
This proves the first part of the claim. The second part then immediately follows by using the fact that $z_\ing + z_\outg \leq 1$, where $z_\outg := p_\outg \cdot x_{\inverter, \outg}$ denotes the expenditure of \inverter on good \outg.
\end{proof}

\begin{lemma}\label{lem:hard:in-lower-spend-indifferent}
For any \inverter buyer with parameter $\inthres \in [0,1]$, if we have $\inslope/p_\ing \leq 2/p_\outg$, then $p_\ing \cdot x_{\inverter, \ing} \geq \min\{1, p_\ing \inthres\} - 6 p_\ing/\inslope - \deltam$.
\end{lemma}

\begin{proof}
Let $z_\ing$ and $z_\outg$ denote the amount of money the \inverter spends on goods \ing and \outg, respectively. We have
$$z_\outg = p_\outg \cdot x_{\inverter, \outg} \leq 2 p_\outg \leq \frac{4 p_\ing}{\inslope}$$
where we used the fact that $x_{\inverter, \outg} \leq 1 + \epsm \leq 2$, since good \outg has to $\epsm$-clear.

The utility $\Uactual$ obtained by \inverter can be bounded as follows
$$\Uactual \leq z_\ing \frac{\inslope}{p_\ing} + z_\outg \frac{1}{p_\outg}.$$
On the other hand, the maximum utility $\Uopt$ achievable by \inverter satisfies
$$\Uopt \geq \min\{1 - z_\outg, p_\ing \inthres\} \frac{\inslope}{p_\ing} + z_\outg \frac{1}{p_\outg}.$$
In the case where $p_\ing \inthres \leq 1 - z_\outg$, this corresponds to buying $\inthres$ units of \ing, and then spending $z_\outg$ on good \outg. In the case where $p_\ing \inthres \geq 1 - z_\outg$, this corresponds to spending $z_\outg$ on good \outg, and all the remaining budget on good \ing.

If $z_\ing > p_\ing \inthres$, then the claim trivially holds, so we assume that $z_\ing \leq p_\ing \inthres$ from now on. By $(1 - \deltam)$-optimality of the bundles, we must have $\Uactual \geq (1 - \deltam) \Uopt$. Plugging in the bounds for the utilities and grouping terms yields
$$(z_\ing - \min\{1 - z_\outg, p_\ing \inthres\}) \frac{\inslope}{p_\ing} \geq - \deltam \min\{1 - z_\outg, p_\ing \inthres\} \frac{\inslope}{p_\ing} - \deltam z_\outg \frac{1}{p_\outg}.$$
Next, we divide by $\inslope/p_\ing$ on both sides,
$$z_\ing - \min\{1 - z_\outg, p_\ing \inthres\} \geq - \deltam \min\{1 - z_\outg, p_\ing \inthres\} - \deltam z_\outg \frac{p_\ing}{p_\outg \inslope} \geq - \deltam - 2 \deltam \frac{p_\ing}{\inslope}$$
where we used the fact that $z_\outg \leq 2 p_\outg$. Finally, since $z_\outg \leq 4 p_\ing/\inslope$, we obtain
$$z_\ing \geq \min\{1, p_\ing \inthres\} - z_\outg - \deltam - 2 \deltam p_\ing/\inslope \geq \min\{1, p_\ing \inthres\} - 6 p_\ing/\inslope - \deltam$$
where we used $\deltam \leq 1$.
\end{proof}

\subsubsection{Correctness of the reduction}

Recall that $H = 9/2$ and $L = 100/\inslope$. The values of $\deltam$ and $\inslope$ will be specified at the end of this section, but the intuition is that $\inslope$ will be picked to be a sufficiently large constant, and then $\deltam > 0$ will be picked to be a sufficiently small constant. If we instead pick $\deltam$ to be a sufficiently small inverse-polynomial value, then our reduction will satisfy all gates of the \pcircuit instance, thus showing unconditional \ppad/-hardness.

For any good $j$, we let $f_j$ denote the total amount of money spent on $j$ by any buyers who do not get any utility by spending that money on the good. In other words, $f_j$ contains any money spent on good $j$ by any buyer who has zero utility for the good, and also any money spent by any buyer beyond the cap on its utility function for the good. We say that a good $j$ is $X$-faulty, if $f_j \geq X$. We say that a gate is $X$-faulty, if any of the goods of the gate (namely, input, output, and possibly auxiliary good) is $X$-faulty.

\begin{lemma}\label{lem:hard:prices-bounded}
Any good that is not $X$-faulty has price at most $9(12+X)/8$.
In particular, if $X \leq 1$, the price of any good that is not $X$-faulty is at most $20$.
\end{lemma}

\begin{proof}
By construction, for any given good, there are at most $12$ buyers with non-zero utility for it. Furthermore, since the good is not $X$-faulty, at most $X$ units of money are spent on the good by other buyers. As a result the total amount of money spent on the good is at most $12 + X$. Given that the good has to $\epsm$-clear, at least $1 - \epsm \geq 8/9$ units must be bought. Thus, the price can be at most $9(12+X)/8$.
\end{proof}

We now define the quantity $F := 2^{12} \deltam \inslope /\deltac$. When we will fix the parameters $\inslope$ and $\deltam$ later in the proof, we will ensure that
\begin{equation}\label{eq:hard:F-bound}
F \leq 1
\end{equation}
so that goods that are not $F$-faulty have price at most $20$ (by \cref{lem:hard:prices-bounded}). Next, we show that the choice of $F$ ensures that only a small fraction of the gates can be $F$-faulty.

\begin{lemma}\label{lem:hard:gates-faulty}
If we set $F := 2^{12} \deltam \inslope /\deltac$, then a fraction at most $\deltac$ of the gates are $F$-faulty.
\end{lemma}

\begin{proof}
Consider any buyer $i$ in the instance and let $p_\outg$ denote the price of its output good. We let $w_i$ denote the amount of money wasted by buyer $i$, i.e., the amount of money it spends on segments with zero utility.

We begin by showing that if $p_\outg \leq 34$, then $w_i \leq 35 \deltam \inslope$. Let $U$ denote the utility obtained by the buyer, and $\Uopt$ the optimal utility achievable by the buyer. Then we have that $\Uopt \geq U + w_i/p_\outg$, since, instead of wasting the money $w_i$ on zero-utility segments, the buyer could have spent it on the output good with bang-per-buck $1/p_\outg$. On the other hand, by $(1-\deltam)$-optimality we also have $U \geq (1-\deltam) \Uopt$. Putting the two together, we obtain that $w_i/p_\outg \leq \deltam \Uopt$. Note that $\Uopt \leq \inslope \cdot t + 1/p_\outg \leq \inslope + 1/p_\outg$, since the buyer can obtain utility at most $\inslope \cdot t$ from the input good, and at most $1/p_\outg$ from the output good (since its bang-per-buck is $1/p_\outg$ and the buyer can spend at most its whole budget, which is $1$, on it). We thus obtain that $w_i \leq \deltam \Uopt \cdot p_\outg \leq \deltam (\inslope p_\outg + 1) \leq 35 \deltam \inslope$, where we used the assumption that $p_\outg \leq 34$.

Let $N_+$ denote the set of buyers with output good satisfying $p_\outg > 34$, and $N_-$ the set of buyers with output good satisfying $p_\outg \leq 34$. Then, we have that $\sum_i w_i \leq |N_+| + |N_-| \cdot 35 \deltam \inslope$, since each buyer in $N_+$ can waste its whole budget (which is $1$), while each buyer in $N_-$ can waste at most $35 \deltam \inslope$ of its budget, by the previous paragraph. Now, by applying \cref{lem:hard:prices-bounded}, we know that for each buyer in $N_+$, its output good must be $16$-faulty, because it has price greater than $34$. Since at most 8 buyers share the same output good, it follows that there are at least $|N_+|/8$ distinct goods that are $16$-faulty. This implies that $\sum_j f_j \geq 16 \cdot |N_+|/8 = 2|N_+|$. Using the fact that $\sum_j f_j = \sum_i w_i$, we thus obtain that $2|N_+| \leq |N_+| + |N_-| \cdot 35 \deltam \inslope$. But this implies that $|N_+| \leq |N_-| \cdot 35 \deltam \inslope$ and thus
$$\sum_j f_j = \sum_i w_i \leq |N_+| + |N_-| \cdot 35 \deltam \inslope \leq |N_-| \cdot 70 \deltam \inslope \leq 8n \cdot 70 \deltam \inslope \leq 2^{10} n \deltam \inslope$$
where we used the fact that the total number of buyers is at most $8n$, where $n$ is the number of nodes in the \pcircuit instance.

Let $m$ denote the number of gates in the \pcircuit instance. Since every good appears in at most two gates, we have that the number of gates that are $F$-faulty is at most
$$2 \frac{\sum_j f_j}{F} \leq \frac{2^{11} n \deltam \inslope}{F} = \frac{2^{11} n \deltam \inslope}{2^{12} \deltam \inslope/\deltac} \leq \deltac \cdot m$$
where we used $m \geq n/2$, which holds due to the fact that every node must be the output of a gate, and every gate has at most two outputs.
\end{proof}

It remains to show that gates that are not $F$-faulty are correctly simulated. For this, it will be convenient to use the following lemma.

\begin{lemma}\label{lem:outputs-clean}
In any gate that is not $F$-faulty, and for any output node $w$ of that gate, we have that the amount of money spent on good $w$ by buyers that are not a part of this gate is:
\begin{itemize}
\item at most $p_w \cdot 8/9 + F$.
\item at least $p_w \cdot 8/9 - 1000/\inslope - 8 \deltam \inslope$, unless $p_w \geq H$.
\end{itemize}
\end{lemma}

\begin{proof}
Recall that node $w$ is the input to exactly one gate. As a result, by construction there are four \inverter buyers with parameter $\inthres = 2/9$ that have good $w$ as input. The upper bound thus follows immediately by construction (given that every node is the input to at most one gate) together with the fact that the good is contained in a gate that is not $F$-faulty. For the lower bound, combining \cref{lem:hard:out-upper-spend} and \cref{lem:hard:in-lower-spend-indifferent} we obtain that the amount of money spent on good $w$ is at least
$$4(\min\{1, p_w \cdot 2/9\} - 6p_w/\inslope - 2 \deltam).$$
Now if we assume that $p_w \leq H = 9/2$, we have $\min\{1, p_w \cdot 2/9\} = p_w \cdot 2/9$. Together with the fact that $p_w \leq 20$ (by \cref{lem:hard:prices-bounded}), the lower bound can be written as
$$p_w \cdot 8/9 - 1000/\inslope - 8 \deltam \inslope$$
where we also used $1 \leq \inslope$.
\end{proof}

We are now ready to proceed with the analysis of the gates.

\begin{lemma}
In any gate $(\NAND,u,v,w)$ that is not $F$-faulty, we have
\begin{itemize}
\item if $p_u \geq H$ and $p_v \geq H$, then $p_w \leq L$.
\item if $p_u \leq L$ or $p_v \leq L$, then $p_w \geq H$.
\end{itemize}
\end{lemma}

\begin{proof}
First, consider the case where $p_u \geq H$ and $p_v \geq H$. By \cref{lem:hard:out-upper-spend} we either have (a) $p_w \leq 2 p_u/\inslope$ or (b) the buyer $\inv_{uw}^i$ spends at most $\max\{0,1-p_u \cdot 2/9\} + 2\deltam \leq 2\deltam$ money on good $w$. In case (a), by \cref{lem:hard:prices-bounded} we have $p_u \leq 20$ and thus $p_w \leq 2 p_u/\inslope \leq 40/\inslope \leq L$, as desired. Thus, it remains to handle the setting where we are in case (b) for all $8$ \inverter buyers appearing in the gate. Note that all together they spend at most $8 \cdot 2\deltam = 16\deltam$ units of money on good $w$. By \cref{lem:outputs-clean}, the total amount of money spent on good $w$ by all buyers in the market is thus at most $p_w \cdot 8/9 + F + 16\deltam$. Since the good must $\epsm$-clear, at least $1-\epsm$ units of the good must be bought, i.e.,
$$p_w \cdot 8/9 + F + 16\deltam \geq (1-\epsm) p_w.$$
Solving for $p_w$ yields
$$p_w \leq \frac{ F + 16\deltam}{1/9-\epsm}.$$
Thus, we indeed have $p_w \leq L = 100/\inslope$, as long as
\begin{equation}\label{eq:hard:NAND-upper}
\frac{ F + 16\deltam}{1/9-\epsm} \leq 100/\inslope
\end{equation}
which we will ensure with our choice of our parameters.

Now, consider the case where $p_u \leq L$. The case where $p_v \leq L$ can be handled in the same way. By \cref{lem:hard:out-lower-spend}, each of the four buyers $\inv_{uw}^i$ spends at least
$$1 - p_u \cdot 2/9 - \deltam(1+\inslope p_w) \geq 1 - 100/\inslope - 21\deltam \inslope$$
money on good $w$, where we used $p_w \leq 20$ by \cref{lem:hard:prices-bounded}. Together with \cref{lem:outputs-clean} we thus get that either $p_w \geq H$, as desired, or the total amount of money spent on good $w$ is at least
$$4(1 - 100/\inslope - 21\deltam \inslope) + p_w \cdot 8/9 - 1000/\inslope - 8 \deltam \inslope \geq 4 + p_w \cdot 8/9 - 2000/\inslope - 100\deltam \inslope.$$
Since good $w$ must $\epsm$-clear, at most $1+\epsm$ units of the good can be bought, i.e.,
$$4 + p_w \cdot 8/9 - 2000/\inslope - 100\deltam \inslope \leq (1+\epsm) p_w$$
which is equivalent to
$$p_w \geq \frac{4 - 2000/\inslope - 100\deltam \inslope}{1/9 + \epsm}.$$
Thus, we indeed have $p_w \geq H$, as long as
\begin{equation}\label{eq:hard:NAND-lower}
\frac{4 - 2000/\inslope - 100\deltam \inslope}{1/9 + \epsm} \geq 9/2
\end{equation}
which we will ensure with our choice of our parameters.
\end{proof}

\begin{lemma}
In any gate $(\PURIFY,u,v,w)$ that is not $F$-faulty, we have
\begin{itemize}
\item if $p_u \geq H$, then $p_v \geq H$ and $p_w \geq H$.
\item if $p_u \leq L$, then $p_v \leq L$ and $p_w \leq L$.
\item if $p_v \in (L,H)$, then $p_w \geq H$ or $p_w \leq L$.
\end{itemize}
\end{lemma}

\begin{proof}
We first argue that
\begin{itemize}
\item if $p_u \geq H$, then $p_\aux \leq L$.
\item if $p_u \leq L$, then $p_\aux \geq H$.
\end{itemize}
The first bullet follows from a very similar analysis to the proof of the first bullet in the previous lemma for the \NAND gate. For the second bullet, as argued in the previous lemma for the \NAND gate, we have that each of the four buyers $\inv_{u\aux}^i$ spends at least
$$1 - p_u \cdot 2/9 - \deltam(1+\inslope p_\aux) \geq 1 - 100/\inslope - 21\deltam \inslope$$
money on good $\aux$. Now, if we assume that $p_\aux \leq H = 9/2$, then by \cref{lem:hard:out-upper-spend} and \cref{lem:hard:in-lower-spend-indifferent}, the two buyers $\inv_{\aux v}^1$ and $\inv_{\aux v}^2$ each spend at least
$$\min\{1, p_\aux \cdot 2/9\} - 6p_\aux/\inslope - 2 \deltam \geq p_\aux \cdot 2/9 - 200/\inslope - 2 \deltam \inslope$$
money on good $\aux$, where we used $1 \leq \inslope$. As a result, the total amount of money spent on good \aux is at least
$$4(1 - 100/\inslope - 21\deltam \inslope) + 2(p_\aux \cdot 2/9 - 200/\inslope - 2 \deltam \inslope) \geq 4 + p_\aux \cdot 4/9 - 1000/\inslope - 100 \deltam \inslope.$$
Since good $\aux$ must $\epsm$-clear, at most $1+\epsm$ units of the good can be bought, i.e.,
$$4 + p_\aux \cdot 4/9 - 1000/\inslope - 100 \deltam \inslope \leq (1+\epsm) p_\aux$$
which is equivalent to
$$p_\aux \geq \frac{4 - 1000/\inslope - 100 \deltam \inslope}{5/9 + \epsm}.$$
Thus, we indeed have $p_\aux \geq H$, as long as
\begin{equation}\label{eq:hard:PURIFY-aux-lower}
\frac{4 - 1000/\inslope - 100 \deltam \inslope}{5/9 + \epsm} \geq 9/2
\end{equation}
which we will ensure with our choice of our parameters.

Now, it remains to prove the following:
\begin{itemize}
\item if $p_\aux \geq H$, then $p_v \leq L$.
\item if $p_\aux \leq H/2$, then $p_v \geq H$.
\item if $p_\aux \leq L$, then $p_w \geq H$.
\item if $p_\aux \geq H/2$, then $p_w \leq L$.
\end{itemize}
Together, these four bullets imply the statement of the lemma. The first bullet can be shown by using the same ideas as the first bullet in the analysis of the \NAND gate.

For the fourth bullet, by \cref{lem:hard:out-upper-spend} we either have $p_w \leq 2p_\aux/\inslope \leq L$, as desired, or the money spent by $\inv_{\aux w}$ on good $w$ is at most
$$\max\{0,1-p_\aux \cdot 4/9\} + 2 \deltam \leq 2\deltam$$
since $p_\aux \geq H/2 = 9/4$. As a result, very similar arguments to the first bullet in the analysis of the \NAND gate apply here again, and we obtain $p_w \leq L$.

For the second bullet, where $p_\aux \leq H/2$, assume that $p_v \leq H$ (otherwise, we are done). By \cref{lem:hard:out-lower-spend} each of the two buyers $\inv_{\aux v}^1$ and $\inv_{\aux v}^2$ spend at least
$$1 - p_\aux \cdot 2/9 - \deltam(1 + \inslope p_v) \geq 1/2 - 21 \deltam \inslope$$
money on good $v$, where we used $p_\aux \leq H/2 = 9/4$. Using \cref{lem:outputs-clean}, the total amount of money spent on good $v$ is at least
$$2(1/2 - 21 \deltam \inslope) + p_v \cdot 8/9 - 1000/\inslope - 8 \deltam \inslope \geq 1 + p_v \cdot 8/9 - 1000/\inslope - 100 \deltam \inslope.$$
Since good $v$ must $\epsm$-clear, at most $1+\epsm$ units of the good can be bought, i.e.,
$$1 + p_v \cdot 8/9 - 1000/\inslope - 100 \deltam \inslope \leq (1+\epsm) p_v$$
which is equivalent to
$$p_v \geq \frac{1 - 1000/\inslope - 100 \deltam \inslope}{1/9 + \epsm}.$$
Thus, we indeed have $p_v \geq H$, as long as
\begin{equation}\label{eq:hard:PURIFY-second}
\frac{1 - 1000/\inslope - 100 \deltam \inslope}{1/9 + \epsm} \geq 9/2
\end{equation}
which we will ensure with our choice of our parameters.

Finally, for the third bullet, where $p_\aux \leq L$, assume that $p_w \leq H$ (otherwise, we are done). By \cref{lem:hard:out-lower-spend} the buyer $\inv_{\aux w}$ spends at least
$$1 - p_\aux \cdot 4/9 - \deltam(1 + \inslope p_w) \geq 1 - 100/\inslope - 21 \deltam \inslope$$
money on good $w$, where we used $p_\aux \leq L$. Then, the same analysis as for the second bullet in the previous paragraph applies again and shows that $p_v \geq H$, as long as condition \eqref{eq:hard:PURIFY-second} is satisfied.
\end{proof}

\paragraph{\bf Setting the parameters.}
Using the fact that $\eps_m$ is a constant with $\epsm < 1/9$, and $\deltac \in (0,1)$ is constant, it is not hard to see that there exist constants $\deltam \in (0,1)$ and $a > 100$ such that the conditions \eqref{eq:hard:F-bound} to \eqref{eq:hard:PURIFY-second} are all satisfied. As a result, in any $(\epsm, \deltam)$-approximate equilibrium, all but a fraction $\deltac$ of the gates are correctly simulated. This completes the proof of \cref{prop:fisher-pcp-hard} and thus \cref{thm:fisher-pcp-hard}.

\paragraph{\bf Additional theorems.}
In order to prove \cref{thm:fisher-inv-poly-hard}, we set $\deltac = 1/2m$, where we recall that $m$ is the number of gates. Then there exists inverse-polynomial $\delta_m$ and constant $a > 100$ such that conditions \eqref{eq:hard:F-bound} to \eqref{eq:hard:PURIFY-second} are all satisfied. By \cref{lem:hard:gates-faulty} it follows that all gates are correctly simulated in that case, and thus we obtain \ppad/-hardness unconditionally, i.e., without the \pcppad conjecture.

Furthermore, in order to prove \cref{thm:fisher-non-zero-hard}, we consider any $(\epsm, \deltam)$-approximate equilibrium where buyers are not spending any money on goods for which their utility function is the zero function. By construction, for any good $j$, there are at most $12$ buyers with non-zero utility function for it. As a result, the total amount of money spent on good $j$ by all buyers in the market is at most $12$. In particular, it follows that good $j$ is not $16$-faulty, and thus by \cref{lem:hard:prices-bounded} it has price at most $34$. Now by using the arguments in the proof of \cref{lem:hard:gates-faulty}, it follows that for any buyer $i$ we have $w_i \leq 35 \deltam \inslope$, where $w_i$ denotes the amount of money wasted by buyer $i$ on segments with zero utility. Finally, if we consider any good $j$ and again use the fact that there are at most $12$ buyers who can spend money on it, we obtain that $f_j \leq 12 \cdot \max_i w_i \leq 12 \cdot 35 \deltam \inslope = 420 \deltam \inslope$. Here we recall that $f_j$ denotes the total amount of money spent on good $j$ that does not yield any utility to the buyer spending the money. As a result, if we set, say, $F := 500 \deltam \inslope$, then no gate can be $F$-faulty. Furthermore, it is possible to pick constants $\deltam \in (0,1)$ and $a > 100$ such that the conditions \eqref{eq:hard:F-bound} to \eqref{eq:hard:PURIFY-second} are all satisfied. This means that all gates are correctly simulated, and we again obtain \ppad/-hardness unconditionally.

\section{Hardness for Fisher Markets Implies the \pcppad Conjecture}
\label{sec:markettopcp}

In this section we will prove that \ppad/-hardness of finding market equilibria with approximately optimal bundles would imply that the \pcppad conjecture holds. We will do so for a class
of markets that we call the \emph{reducible} markets, which are defined as
follows.
A Fisher market $(G,B,(e_i)_{i\in B},(u_i)_{i \in B})$ is reducible if all of
the following hold.

\begin{itemize}
\item The market has {\em constant} degree $d$. We say that a market has degree $d$ if, for each buyer $i \in B$, there
are at most $d$ goods $j$ such that $u_{i,j}$ is a non-zero function, and for
each good $j$ there are at most $d$ buyers $i$ such that $u_{i,j}$ is a
non-zero function.

\item The budgets are upper and lower bounded by some positive constants $\bmax$ and $\bmin$, i.e., $\max_i e_i \leq \bmax$ and $\min_i e_i \geq \bmin$.

\item The utilities are SPLC, each piecewise linear utility function $u_{i,j}$ has a constant number of
pieces, and the non-zero slopes of each piece lie in the range
$[1, \kappa]$ for some constant $\kappa$. 
Note that so long as, for any given buyer, the ratios between all non-zero slopes are bounded by a
constant, the latter condition can be achieved by a simple re-normalization step.

More formally, for each $i$ and $j$,
if $\langle (s_{i,j,1}, \ell_{i,j,1}), (s_{i,j,2}, \ell_{i,j,2}), \dots, (s_{i,j,m},
\ell_{i,j,m}) \rangle$ is the representation of the utility function $u_{i,j}$,
then we require that $m$ is upper bounded by a constant, and that for all $k$, $s_{i,j,k}$ is either zero or lies in $[1,
\kappa]$.
\end{itemize}

In this section we show the following theorem.
\begin{theorem}
\label{thm:markettopcp}
If there exists $\delta > 0$ such that finding a $(0, \delta)$-approximate market equilibrium in reducible Fisher markets is \ppad/-hard, then the \pcppad conjecture holds.
\end{theorem}

An important point is that the \emph{simple} markets (\cref{def:simple-condition}) for which our hardness result holds (\cref{thm:fisher-pcp-hard}) are in particular also \emph{reducible} markets, as defined above. As a result, \cref{thm:markettopcp} shows that the \pcppad conjecture is necessary to obtain \ppad/-hardness for equilibrium with approximately optimal bundles in simple markets.

We will prove~\cref{thm:markettopcp} by showing that, given a reducible
market and a constant $\delta > 0$, we can construct an instance of $(\epsc, \deltac)$-\gcircuitp for some
sufficiently small constants $\epsc$ and $\deltac$, and then in polynomial time recover a $(0,
\delta)$-approximate market equilibrium from any solution of the generalized circuit instance. For the rest of this section, we fix $(G,B,(e_i)_{i\in B},(u_i)_{i \in B})$ to be a particular reducible market to which we will apply our reduction.

Our reduction follows a five step procedure. The aim of the first four steps will be to find an $(f(\epsc,\deltac), g(\epsc,\deltac))$-approximate market equilibrium, where $f$ and $g$ are some functions such that $f(\epsc,\deltac)$ and $g(\epsc,\deltac)$ can be made arbitrarily small by picking $\epsc$ and $\deltac$ small enough. In the fifth step, we will then show that we can pick sufficiently small constants $\epsc$ and $\deltac$ such that we can go from an $(f(\epsc,\deltac), g(\epsc,\deltac))$-equilibrium to a $(0, \delta)$-equilibrium in polynomial time.

\paragraph{\bf Preprocessing.}
Before proceeding with the first step of the reduction, we perform two preprocessing steps. First, we ensure that $\ell_{i,j,k} \le 2$ for all $i,j,k$. This can be enforced by simply truncating any segments that extend beyond two units of the good. Indeed, note that the buyer can be allocated at most $1+\eps$
of a good that $\eps$-clears, and equilibrium computation in Fisher markets is trivial when $\eps = 1$. Thus, what the utility function looks like beyond two units of good is irrelevant.

Second, for each good $j$, let $I_j = \{ (i,k) \; : \; i \in B \text{ and }
s_{i,j,k} > 0 \}$ denote the set of all segments of that good for which buyers have
non-zero utility. We ensure that $\sum_{(i,k) \in I_j} \ell_{i,j,k} > 1 + \delta$. 
In other words, we insist that every good $j$ would be over-demanded by more than $\delta$ units, if it had price zero. This can be enforced as follows. For any good $j$ with 
$\sum_{(i,k) \in I_j} \ell_{i,j,k} \le 1 + \delta$, we remove the good from the instance. Let $G_-$ denote the set of goods that were thus removed from $G$. We perform our five-step reduction on the market (with the goods in $G_-$ removed) and obtain a $(0,\delta)$-equilibrium. Then, we add the goods in $G_-$ back to the market with corresponding prices set to zero.

It remains to describe how we update the allocation in order to obtain a $(0,\delta)$-equilibrium of the market containing all goods. For every good $j \in G_-$, let $r_j := \sum_{(i,k) \in I_j} \ell_{i,j,k}$ and recall that $r_j \leq 1+\delta$, by definition of $G_-$. If $r_j = 0$, then assign one unit of good $j$ to some arbitrary buyer. If $r_j > 0$, then for each buyer $i$, we assign a quantity
$$\frac{1}{r_j} \sum_{k: s_{i,j,k} > 0} \ell_{i,j,k}$$
of good $j$ to the buyer. By definition of $r_j$, good $j$ clears exactly. Finally, note that for every good $j \in G_-$, every buyer receives a fraction at least $1/r_j \geq 1/(1+\delta) \geq 1-\delta$ of the total quantity of good $j$ that she has positive utility for. As a result, the bundle received by each buyer remains $(1-\delta)$-optimal.

\subsection{Step 1: Reduce to \gcircuitp}

The first step is to construct an $(\epsc, \deltac)$-\gcircuitp instance, where the constants
$\epsc$ and $\deltac$ will be fixed only at the very end of the proof of~\cref{thm:markettopcp}. At the end of Step 1, we aim to have a price
vector $p$ and an allocation $x$ in which a $(1 - c \cdot \deltac)$-fraction of
the goods $(c \cdot \epsc)$-clear, and in which a  $(1 - c \cdot
\deltac)$-fraction of the buyers receive $(1 - c \cdot \epsc)$-optimal allocations.

Our \gcircuitp formulation will have the following set of variables.
\begin{itemize}
\item For each good $j$, we have a price variable $p_j \in [\pmin, \pmax]$,
where $\pmin = \bmin / (8 \cdot d \cdot \kappa)$ and $\pmax = 2 \cdot d \cdot
\bmax$. Observe that for a reducible market, both $\pmin$ and $\pmax$ are
constant. 

\item For each buyer $i$, each good $j$, and each segment $k$ of $u_{i,j}$ we
define a variable $q_{i,j,k}$ to denote the amount of money that buyer $i$
spends on this segment. We constrain $q_{i,j,k} \in [0, \ell_{i,j,k} \cdot
\pmax]$.
\end{itemize}
Observe that in a reducible market, the bounds on these variables are constant.

When we build our \gcircuitp instance, we will often use long-form expressions
to specify a sequence of gates and intermediate variables that should be built.
When we do so, we ensure that the bounds on each intermediate variable are set
so that they never affect the computation. For example, if we were to introduce
an intermediate variable $x := y + z$, we set the bounds $x_l = y_l + z_l -
\epsc$ and $x_u = y_u + z_u + \epsc$ to ensure that the result is never truncated
by the bounds on $x$. In this way, we ensure that truncation will only occur for the
variables $p_j$ and $q_{i,j,k}$ defined above. 

Our \gcircuitp formulation will have the following constraints. 
\begin{itemize}
\item For each good $j$, we first introduce an intermediate variable $d_j$ to
hold the excess demand of good $j$.
\begin{equation}
\label{eqn:d}
d_j := \sum_{i,k} q_{i,j,k} - p_j,
\end{equation}
where the sum only includes segments for which $q_{i,j,k} > 0$. 
Then we set
\begin{equation}
\label{eqn:p}
p_j := p_j + d_j,
\end{equation}
where here the $G_{+}$ gate will truncate the computation using the bounds for
$p_j$.

\item For each buyer $i$, we introduce an intermediate variable $b_i$ to
hold a budget signal for that buyer. This signal will output $1$ if the buyer is
significantly underspending their budget, and $-1$ if the buyer significantly
overspends. 
\begin{equation}
\label{eqn:budget}
b_i := 2 \cdot \left(e_i > \sum_{j,k} q_{i,j,k} \right) - 1
\end{equation}
where the sum only includes segments for which $q_{i,j,k} > 0$.
Here, if we ignore the $\epsc$ errors, then the $G_{>}$ gate will output a
value in $[0, 1]$, and we then rescale that value to the range $[-1, 1]$.

\item For each buyer $i$, and each pair of slopes $s_{i,j,k}$ and $s_{i,j',k'}$
we introduce an intermediate comparator variable $c_{i,j,k,j',k'}$ which will
output $1$ if the bang-per-buck of $s_{i,j,k}$ is significantly worse than the
bang-per-buck of $s_{i,j',k'}$, and $-1$ if the latter is significantly worse
than the former. 
\begin{equation*}
c_{i,j,k,j',k'} := 2 \cdot \Bigl( s_{i,j,k} \cdot p_{j'} < s_{i,j',k'} \cdot p_{j}
\Bigr) - 1.
\end{equation*}
This is simply a reformulation of the standard bang-per-buck comparison
$s_{i,j,k} / p_{j} < s_{i,j',k'} / p_{j'}$ that avoids using a division.
Note that since the prices are constrained to be strictly positive, there are
no issues with division by zero here. 

\item For each buyer $i$, good $j$, and piece $k$, we introduce an intermediate
variable $g_{i,j,k}$ that will hold an aggregate signal for $q_{i,j,k}$. Let
$M_i$ be the total number of utility-function segments for which $s_{i,j,k} >
0$, and note that $M_i$ is upper bounded by a constant. We set
\begin{equation*}
g_{i,j,k} := \sum_{j',k'} \max(0, c_{i,j,k,j',k'}) + (M_i + 1) \cdot b_i,
\end{equation*}
where in the sum we only sum over $j'$ and $k'$ such that $s_{i,j',k'} > 0$.
That is, the segment receives positive signals from other segments when it has
a better bang-per-buck than them, and it also receives a budget signal that is
weighted to be able to overwhelm the bang-per-buck signals if necessary.

\item Finally, 
for each buyer $i$, good $j$, and piece $k$, we set
\begin{equation}
\label{eqn:q}
q_{i,j,k} := \max\Biggl(\min\Bigl(q_{i,j,k} + g_{i,j,k}, \;\; p_j \cdot
\ell_{i,j,k}\Bigr), 0\Biggr),
\end{equation}
which makes the expenditure on the segment move according to its signal, but
truncates from above to ensure that the total expenditure does not exceed the total cost
of buying the full segment, and truncates from below to ensure that the expenditure does not fall below zero.

\end{itemize}

\paragraph{\bf Correctness of Step 1.}

Given an $(\epsc, \deltac)$-\gcircuitp solution of the instance, we use the
$p_j$ variables directly as a price vector, and we construct an allocation for
each buyer $i$ using the following procedure.
\begin{enumerate}
\item First we create an initial allocation $x'_{i,j} = \sum_k q_{i,j,k} /
p_j$.
\item Then, if $\sum_{j} p_j \cdot x'_{i,j} < e_i$, meaning that the initial allocation
under-spends, we create the final allocation $x_i$
by spending $e_i - \sum_{j} p_j \cdot x'_{i,j}$ extra money optimally in addition to what
is spent by $x'_i$. 
\item On the other hand if $\sum_{j} p_j \cdot x'_{i,j} > e_i$, meaning that the initial
allocation over-spends, we create
the final allocation $x_i$ by removing $ \sum_{j} p_j \cdot x'_{i,j} - e_i$ money optimally
(that is, removing so as to minimize the loss of utility) from $x'_i$. 
\end{enumerate}
The last two steps above can both be implemented in polynomial time by, for
each buyer, adding or removing money greedily from utility-function segments in
bang-per-buck order.

We say that a buyer $i$ is \emph{satisfied} if all of the constraints that involve
buyer $i$ are satisfied. The following lemma states that, if buyer $i$ is
satisfied, then the solution to the \gcircuitp instance will allocate money to
utility-function segments in approximate bang-per-buck order.

\begin{lemma}
\label{lem:slopes}
There is a suitably small constant value for $\epsc$ such that
if buyer $i$ is satisfied, then for all $j$, $j'$, $k$, and $k'$ if $s_{i,j,k} / p_j
> s_{i,j',k'} / p_{j'} + \epsc$ then we cannot have $q_{i,j,k} < \ell_{i,j,k} \cdot
p_j - \epsc$ and $q_{i,j',k'} > \epsc$.
\end{lemma}
\begin{proof}
Suppose for the sake of contradiction that this does occur. 
We first present the argument for the case where $\epsc = 0$. 
Let $z_{j,k}= \sum_{j'',k''} \max(0, c_{i,j,k,i'',k''})$ be the sum of the
positive signals received by the segment $s_{i,j,k}$, and likewise let $z_{j',k'}= \sum_{j'',k''} \max(0,
c_{i,j',k',i'',k''})$. Observe that $z_{j,k} > z_{j',k'} + 1$ since $s_{i,j,k}$ receives a signal from $s_{i,j',k'}$ but no
signal is sent in return. 
\begin{itemize}

\item Since $q_{i,j,k} < \ell_{i,j,k} \cdot p_j - \epsc$, this means that the solution
for $q_{i,j,k}$ was not constrained by the $G_{\min}$ gate in \eqref{eqn:q}, which
means that we must have $g_{i,j,k} \le 0$, which implies that $z_{j,k} + (M_i +
1) \cdot b_i \le 0$. 

\item On the other hand, since 
$q_{i,j',k'} > \epsc$ we have that $q_{i,j',k'}$ was not constrained by the
$G_{\max}$ gate in \eqref{eqn:q}. Therefore, we must have $g_{i,j,k} \ge 0$, and so 
$z_{j',k'} + (M_i + 1)\cdot  b_i \ge 0$.
\end{itemize}

But now we have
\begin{align*}
0 &\le z_{j',k'} + (M_i + 1) \cdot b_i \\
& < z_{j,k} - 1 + (M_i + 1) \cdot b_i  \\
& \le -1,
\end{align*}
which is a contradiction.

When $\epsc$ is not zero, first note that there is still a difference of at
least one in the number of signals received by $s_{i,j,k}$ and $s_{i,j',k'}$
because we have $s_{i,j,k} / p_j < s_{i,j',k'} / p_{j'} - \epsc$, and this gap
is large enough to be detected by the $G_{<}$ gate used to implement the
comparator. Our arguments about $q_{i,j,k}$ and $q_{i,j',k'}$ not being
constrained by the $G_{\min}$ and $G_{\max}$ gates, respectively, still hold
since there is again a gap that is strictly greater than $\epsc$. 
The only other difference is that each value will be computed now with an error
of at most $c \cdot \epsc$, where $c$ is the number of gates used to compute
that value, and since our instance is reducible, we have that $c$ is constant.
So we just need to ensure that the final values of $g_{i,j,k}$ and
$g_{i,j',k'}$ are computed with accuracy strictly better than $0.5$, which then
still allows our $0 < -1$ contradiction to hold, and we can therefore choose
$\epsc < 0.5 / c$, where $c$ is the largest number of gates used to compute a
value that is used in this proof.
\end{proof}

The following lemma states that each satisfied buyer approximately satisfies their
budget constraint in the allocation $x'$ that we derived from the $\gcircuitp$
solution.

\begin{lemma}
\label{lem:budget}
There is a suitably small constant value for $\epsc$ such that, for each satisfied buyer
$i$, we have $\left|\sum_{j,k} q_{i,j,k} - e_i\right| \le (M_i+1) \cdot
\epsc$.
\end{lemma}
\begin{proof}
We first prove that $\sum_{j,k} q_{i,j,k} \le e_i + (M_i+1) \cdot \epsc$ for every
satisfied buyer $i$. Suppose for the sake of contradiction that this is not the
case for some satisfied buyer $i$. Since $M_i+1$ is strictly larger than the number
of summands in $\sum_{j,k} q_{i,j,k}$, then observe that the $G_{>}$ gate in
the definition of the budget signal $b_i$ in~\cref{eqn:budget}
will output $0 \pm \epsc$, and therefore we will have
$b_i \le -1 + 3 \cdot \epsc$, where the extra $2 \cdot \epsc$ error arises from
the $G_{\times 2}$ and $G_{-}$ gates used to compute~$b_i$. 

For each $i$, $j$, and $k$, the signal $g_{i,j,k}$ is obtained by summing at
most $M_i$ comparators $c_{i,j,k,i',k'}$, and then adding $(M_i+1) \cdot b_i$ to
that sum. Note that each comparator outputs a value of at most $1+3\cdot
\epsc$, which implies that each signal 
\begin{align*}
g_{i,j,k} &\le M_i\cdot (1 + 3\cdot \epsc) + (M_i + 1) \cdot (-1 + 3 \cdot \epsc)
\\
&= -1 + (2M_i + 1) \cdot 3 \cdot \epsc \\
&< -\epsc
\end{align*}
where in the final inequality we have used the fact 
that we can select $\epsc \le 1 / (6 \cdot (2M_i + 1))$ because $M_i$ is
constant.
 
Therefore, every segment of buyer $i$ receives a strictly negative signal from
$g_{i,j,k}$. This means that we should have $q_{i,j,k} \le \epsc$ for all $i$, $j$,
and $k$ in order for the constraints on the $q_{i,j,k}$ variables to be
satisfied. But this then contradicts the assumption that 
$\sum_{j,k} q_{i,j,k} > e_i + (M_i+1) \cdot \epsc$, since $e_i$ is bounded away
from zero.

The proof that $\sum_{j,k} q_{i,j,k} \ge e_i - (M_i+1) \cdot \epsc$ for all
satisfied
buyers $i$ is essentially the same. 
If we assume that $\sum_{j,k} q_{i,j,k} < e_i - (M_i+1) \cdot \epsc$, then
this time the budget signal $b_i$ outputs a
value of at least $1 - 3 \cdot \epsc$ which is again big enough to overwhelm
the comparator signals in the computation of the $g_{i,j,k}$ signals. So the
only way for the $q_{i,j,k}$ constraints to be satisfied would be for each
$q_{i,j,k}$ to be at their upper bound of $\pmax$, but $\pmax > \bmax \ge e_i$,
which then contradicts our assumption that $\sum_{j,k} q_{i,j,k} < e_i - (M_i+1)
\cdot \epsc$.
\end{proof}

We can now combine the previous two lemmas to show that the allocation $x$
gives each satisfied buyer $i$ an approximately optimal bundle.

\begin{lemma}
\label{lem:optimal}
There is a suitably small constant value for $\epsc$ such that,
for each satisfied buyer $i$, the allocation given by $x_{i}$ is $(1 - c \cdot \epsc)$
optimal for some constant $c$.
\end{lemma}
\begin{proof}
Let $x_i^*$ denote the actual optimal allocation for buyer $i$ under the price
vector $p$. We will relate this to the actual allocation $x$ using a sequence
of steps. 
\begin{enumerate}
\item First, let $y_i^*$ be the optimal allocation for buyer $i$ 
under the price vector $p$ where buyer $i$'s budget has been reduced by 
$(M_i+1) \cdot \epsc$. Since the slopes $s_{i,j,k}$ are constant for
all utility function pieces in each $u_{i,j}$, 
and since we can select $\epsc$ to be a suitably small constant to ensure that 
$(M_i+1) \cdot \epsc \ll e_i$, 
there exists a constant $c_y$ such
that $\sum_j u_i(y^*_{i}) \ge (1 - c_y \cdot \epsc) \cdot u_i(x^*_i)$.

\item Next consider the hypothetical allocation $z_i$ in which buyer $i$ spends 
exactly $q_{i,j,k}$ on each line segment $s_{i,j,k}$. This cannot be
implemented in practice, since we must fully buy $s_{i,j,k}$ before we can
spend money on $s_{i,j,k+1}$, but it will serve as a useful intermediary.

Note that by~\cref{lem:budget}, we have that $z_i$ always spends at least
as much as $y^*_i$, and that $z_i$ spends at most 
$2 \cdot (M_i+1) \cdot \epsc$ more than $y^*_i$. We now compare the utility
obtained by these two allocations. 

Let $j_w$ and $k_w$ denote the segment such that $y^*_{i,j_w,k_w} > 0$ and 
$j_w$ and $k_w$ minimize $s_{i,j_w,k_w} / p_{j_w}$. That is, these are the indices
of the worst utility-function segment that is bought by $y^*$.
Then we divide the goods into three sets:
\begin{itemize}
\item The set $A$ contains pairs $(j,k)$ such that 
$s_{i,j,k} / p_j > s_{i,j_w,k_w} / p_{j_w} + \epsc$.
\item The set $B$ contains 
pairs $(j,k)$ such that 
$s_{i,j,k} / p_j < s_{i,j_w,k_w} / p_{j_w} + \epsc$.
\item The set $C$ contains all pairs $(j, k)$ that are not in $A$ or $B$.
\end{itemize}
The idea is that $A$ contains all segments whose bang-per-buck would be
identified by a comparator as being strictly better than the worst segment,
while $B$ contains all segments that a comparator would identify as having
strictly worse bang-per-bucks.

Since $z$ spends at least as much money as $y^*$, \cref{lem:slopes}
implies that $z$ must buy $\ell_{i,j,k} - \epsc$ of any segment $(j, k) \in A$.
Thus the utility gained by $z$ from the segments in $A$ is at least a $(1 -
c \cdot \epsc)$ fraction of the utility gained by $y^*$ from the segments in $A$, for
some constant $c$,
where we are using the fact that each $\ell_{i,j,k}$ is constant. 

While $z$ is not required to buy the same segments as $y^*$ in $B$,
\cref{lem:slopes} and~\cref{lem:budget} imply that it must
spend at least as much money on segments in $B$, and all such segments
differ in bang-by-buck by at most a constant. Therefore the utility gained by 
by $z$ from the segments in $B$ is again at least a $(1 -
c \cdot \epsc)$ fraction of the utility gained by $y^*$ from the segments in $B$, for
some constant $c$.

There may be some allocation of $z$ to segments in $C$, but this only improves
the utility gained by $z$, while we know that $y^*$ does not gain any utility
from segments in $C$. 
By combining all three sets, we get that there exists some constant $c$ such that
$u_i(z) \ge (1 - c \cdot \epsc) \cdot u_i(y^*)$.

\item Next we consider the initial allocation $x'$. We have that $x'$ differs
from $z$ only in that it is forced to buy each line segment $s_{i,j,k}$ before it
buys $s_{i,j,k+1}$. This means that we can construct $x'$ from $z$ by, for each
good $j$, continually shifting money from the worst bang-per-bug segment for
$j$ to the best bang-per-buck segment for $j$ until that segment is fully
bought. This process can only increase the utility that buyer $i$ obtains,
since we are shifting money from worse segments to better ones. So we
immediately get that $u_i(x') \ge u_i(z)$. 

\item Finally we consider the actual allocation $x$. If $x$ was constructed
from $x'$ by adding extra money in the case where $x'$ underspent, then we
clearly have 
$u_i(x) \ge u_i(x')$, since spending extra money can only increase the utility
of buyer $i$. On the other hand, if 
$x$ was constructed
from $x'$ by removing money in the case where $x'$ underspent, note that the
amount of money that was removed was at most 
$(M_i+1) \cdot \epsc$ by~\cref{lem:budget}. Since this money is
always removed from the goods that have lowest bang-per-buck, and since the
ratio between the slopes of any two segments is constant,  we get that
$u_i(x') \ge (1 - c \cdot \epsc) \cdot u_i(x)$ for some constant $c$. 
\end{enumerate}
Combining the three steps above gives us that $u_i(x) \ge (1 - c \cdot \epsc)
\cdot u_i(x^*)$ for some constant $c$, as required.
\end{proof}

We say that a good $j$ is \emph{satisfied} if every buyer that has a non-zero
utility function for $j$ is satisfied. 
The following lemma states that a satisfied good approximately clears in the
initial allocation $x'$. 

\begin{lemma}
\label{lem:clears}
There is a suitably small constant value for $\epsc$ such that,
if good $j$ is satisfied, then $\left| \sum_{i \in B} x'_{i,j} - 1 \right| \le c' \cdot \epsc$ for some constant $c'$.
\end{lemma}
\begin{proof}
Suppose for the sake of contradiction that this does not hold, and let $j$ be a
satisfied good such that the inequality does not hold.
Observe that the demand variable $d_j$ will compute 
$\sum_{i,k} q_{i,j,k} - p_j \pm c \cdot \epsc$ for some constant $c$, because
there are at most constantly many segment indices $i$ and $k$ such that
$s_{i,j,k} > 0$. 
\begin{itemize}
\item If the gate implementing $p_j = p_j + d_j$ in Equation~\eqref{eqn:p} is not truncated by
the bounds on $p_j$, then we must have $|d_j| \le \epsc$, which would then
imply that 
$$\left|\sum_{i,k} q_{i,j,k} - p_j \pm c \cdot \epsc\right| \le \epsc,$$ 
and dividing this inequality by $p_j$ then yields 
$$\left|\sum_{i} x'_{i,j} - 1 \pm c \cdot \epsc\right| \le \epsc/p_j,$$ 
which would contradict our assumption.

\item If the gate implementing $p_j = p_j + d_j$ is truncated by the upper
bound on $p_j$, then we have $d_j \ge -\epsc$ and $p_j = \pmax = 2 \cdot d \cdot \bmax$. 
However, since all buyers who are interested in good $j$ are satisfied,
\cref{lem:budget} implies that each such buyer satisfies their budget
constraint
$\left|\sum_{j,k} q_{i,j,k} - e_i\right| \le (M_i+1) \cdot \epsc$, and there
are at most $d$ such buyers. So the
total amount of money spent on good $j$ can be at most $d \cdot (\bmax +
(M+1) \cdot \epsc)$, where $M$ is the maximum over the $M_i$ terms of each
buyer $i$ who is interested in good $j$. 

So long as 
$\epsc < \bmax/(2M+2)$ we have $\sum_{i,k} q_{i,j,k} < 
d \cdot 1.5 \cdot \bmax$, where here we are using the fact that we can
choose $\epsc$ to be a suitably small constant, and that $\bmax$ and $M$ are
both constants. Therefore, the excess demand variable $d_j$ as computed in
Equation~\eqref{eqn:d} must satisfy 
\begin{align*}
d_j &\le d \cdot 1.5 \cdot \bmax - d \cdot 2 \cdot \bmax + c \cdot \epsc \\
&= -0.5 \cdot d \cdot \bmax + c \cdot \epsc
\end{align*}
for some constant $c$, where the 
$c \cdot \epsc$ term arises from the errors in the constantly-many gates that
are used to compute $d_j$. Since $c$, $d$, and $\bmax$ are constant, we can pick
$\epsc$ to be a small enough constant to ensure that $d_j < -\epsc$, which then
contradicts the fact that $d_j \ge -\epsc$ in this case.

\item If the gate implementing $p_j = p_j + d_j$ is truncated by the lower
bound on $p_j$, then we have $d_j \le \epsc$ and $p_j = \pmin = \bmin / (8 \cdot d
\cdot \kappa)$. 

We argue that each buyer $i$ who has a non-zero utility
function for good $j$ must buy their full allocation of $j$. To see why,
observe that if there is a good $j'$ and a segment $k'$ such that $s_{i,j',k'}
/ p_{j'} \ge s_{i,j,k} / p_j - \epsc$, then we must have
\begin{align*}
p_{j'} &\le (s_{i,j',k'} / s_{i,j,k}) \cdot p_j + (p_j p_{j'}/s_{i,j,k}) \cdot \epsc \\
&\le \kappa \cdot p_j + \pmax^2 \cdot \epsc,
\end{align*}
where in the second inequality we
have used the fact that all slopes lie in $[1,\kappa]$ and that prices are upper bounded by $\pmax = 2 \cdot d \cdot \bmax$. Since by assumption the buyer can demand at most 2 units of $p_{j'}$
it therefore costs at most $2 \cdot (\kappa \cdot p_j + \pmax^2 \cdot \epsc)$ to buy those units. Since
there are at most $d$ goods $j'$ for which this can occur, it therefore costs
at most $d \cdot 2 \cdot (\kappa \cdot p_j + \pmax^2 \cdot \epsc) = \bmin/4 + 2 \cdot d \cdot
(2 \cdot d \cdot \bmax)^2 \cdot \epsc$ for buyer $i$ to buy all
such segments. Since $\bmin$, $\bmax$, and $d$ are all constant, we can choose
a suitably small constant value for $\epsc$ so that $\epsc < \bmin/(8 \cdot d
\cdot (2 \cdot d \cdot \bmax)^2)$, and doing so ensures 
that buying all such segments costs at most $\bmin/2$ money.

Buying all segments of $j$ costs at most $2 \cdot \pmin < \bmin/4$ money, where
we have used the assumption that no buyer demands more than 2 units of any
good. Therefore, by~\cref{lem:slopes}, the solution of the \gcircuitp instance must buy all
segments of good $j$, with at most $\epsc$ error in each segment. 
That is, $\sum_{i,k} q_{i,j,k} \ge \sum_{i,k} (\ell_{i,j,k} \cdot p_j - M_i \cdot \epsc)$,
where $M_i$ is the total number of non-zero segments of buyer $i$. 

Recall that $I_j = \{ (i,k) \; : \; i \in B \text{ and }
s_{i,j,k} > 0 \}$ denotes the set of all segments of good $j$ for which buyers have non-zero utility. By the preprocessing, we have that $\sum_{(i,k) \in I_j} \ell_{i,j,k} > 1 + \delta$.
This implies that $\sum_{i,k} q_{i,j,k} > (1 + \delta - \sum_{(i,k) \in I_j} (M_i
\cdot \epsc))
\cdot p_j$, and therefore so long as $\epsc + \sum_{(i,k) \in I_j} M_i \cdot \epsc
< \delta$
we get that $d_j > \epsc$, and we can choose $\epsc$ to be suitably small to
ensure this because $M_i$ and $|I_j|$ are upper bounded by a constant, and $\delta > 0$ is constant. So we have concluded
$d_j > \epsc$ but also $d_j \le \epsc$, which provides our contradiction.
\end{itemize}
\end{proof}

Finally, we show that every satisfied good approximately clears in $x$.

\begin{lemma}
\label{lem:clears2}
There is a suitably small constant value for $\epsc$ such that, if good $j$ is
satisfied, then 
$\left| \sum_{i \in B} x_{i,j} - 1 \right| \le c \cdot \epsc$ for some
constant $c$.
\end{lemma}
\begin{proof}
By \cref{lem:clears} we have that
$\left| \sum_{i \in B} x'_{i,j} - 1 \right| \le c' \cdot \epsc$ for some
constant $c'$, and by \cref{lem:budget} we have that $x'$ differs from $x$ only
by allowing each buyer to reallocate at most $c'' \cdot \epsc$ money for some
constant $c''$. 
Since the degree of the market is at most $d$, the amount spent on each good
can therefore change by at most $d \cdot c'' \cdot \epsc$, and since all prices
are at least $\pmin$, the amount of each good that is bought can change by at
most $d \cdot c'' \cdot \epsc/\pmin$. Combining all of this yields
$$\left| \sum_{i \in B} x_{i,j} - 1 \right| \le c' \cdot \epsc + d \cdot c''
\cdot \epsc/\pmin,$$
and since $c'$, $c''$, $d$, and $\pmin$ are all constant, the claim has been
shown. 
\end{proof}

The following lemma summarizes what we have shown for Step 1.

\begin{lemma}
\label{lem:step1}
There exists a constant $c$ and a suitably 
small constant value for $\epsc$ such that 
we can make a single call to the $(\epsc, \deltac)$-\gcircuitp problem and, in
polynomial time, recover a price vector $p$ and allocation $x$ that satisfy the
following properties.
\begin{itemize}
\item A $(1 - c \cdot \deltac)$-fraction of the goods $(c \cdot \epsc)$-clear.
\item A $(1 - c \cdot \deltac)$-fraction of the buyers have 
$(1 - c \cdot \epsc)$-optimal allocations and satisfy their budget
constraints with equality. 
\end{itemize}
\end{lemma}
\begin{proof}
Since each buyer has constantly many constraints, there exists a constant $q$ 
such that a $(1 - q \cdot \deltac)$-fraction of the buyers are satisfied. 
Since the market has constant degree, we also get that at least $(1 -
r \cdot \deltac)$ fraction of the goods are satisfied for some constant $r$.
\cref{lem:optimal} states that any satisfied buyer has a $(1 -
s \cdot \epsc)$-optimal bundle for some constant $s$, while~\cref{lem:clears2} states that any satisfied good $(1 - t \cdot
\epsc)$-clears for some constant $t$. Finally, all buyers satisfy their budget
constraints with equality by construction, since $x$ was constructed from $x'$
in a way that ensures this. Taking $c = \max(q,r,s,t)$ then proves the
claim.
\end{proof}

\subsection{Step 2: Fix the broken buyers}

After Step 1 we have a price vector $p$ and an allocation $x$ in which a $(1 -
c \cdot \deltac)$ fraction of the buyers have $(1 - c \cdot \epsc)$-optimal
allocations. We refer to any buyer $i$ who does not have a 
$(1 - c \cdot \epsc)$-optimal
allocation as \emph{broken}. In Step 2, we modify $p$ and $x$ to eliminate the
broken buyers.

This is achieved in a straightforward way. For every broken buyer $i$, we
simply modify $x_i$ so that buyer $i$ now buys an optimal allocation under the
price vector $p$. The correctness of this is shown in the following lemma.

\begin{lemma}
\label{lem:step2}
Let $p$, $x$, and $c$ be a price vector, allocation, and constant that
satisfy the conclusion of~\cref{lem:step1}, respectively.
There is a suitably small constant value for $\epsc$ such that, 
in polynomial time, we can find an allocation $x'$ for which there is a constant
$c'$ where
\begin{itemize}
\item a $(1 - c' \cdot \deltac)$-fraction of the goods $(c' \cdot \epsc)$-clear,
and
\item all buyers have a 
$(1 - c' \cdot \epsc)$-optimal allocation and satisfy their budget constraint with equality.
\end{itemize}
\end{lemma}
\begin{proof}
The allocation $x'$ is defined so that each buyer $i$ who does not have a 
$(1 - c \cdot \epsc)$-optimal allocation under $x$ buys an optimal allocation
under the price vector $p$, and each other buyer remains unchanged. This new
allocation can be computed in polynomial time by simply buying the
utility-function segments of buyer $i$ in bang-per-buck order for each broken
buyer $i$ until their budget is exhausted. Since all other buyers remain unchanged, we immediately have that
every buyer in allocation $x'$ receives a 
$(1 - c \cdot \epsc)$-optimal bundle.

Since the market has constant degree $d$, and since at most a 
$(1 - c \cdot \deltac)$-fraction of the buyers were broken, 
it follows that at
most a 
$(1 - c' \cdot \deltac)$-fraction of the goods change their allocations between
$x$ and $x'$ for some constant $c'$. If a good's allocation was not changed, and it 
$(c \cdot \epsc)$-clears in $x$, then it also 
$(c \cdot \epsc)$-clears in $x'$, while on the other hand, we have no guarantees
about the clearing of the goods that were changed. 

So for $x'$ we have that
there is at most a 
$(1 - c \cdot \deltac)$-fraction of the goods that do not 
$(c \cdot \epsc)$-clear because they did not 
$(c \cdot \epsc)$-clear in $x$ and they have not subsequently been changed, and we have that
there is at most a 
$(1 - c' \cdot \deltac)$-fraction of the goods that do not 
$(c \cdot \epsc)$-clear because their allocation changed from $x$ to $x'$.
All other goods continue to 
$(c \cdot \epsc)$-clear in $x'$. Therefore there is a constant $c'' \ge \max(c,
c')$ such that at least a $(1 - c'' \cdot \deltac)$-fraction of the goods
$(c \cdot \epsc)$-clear, as required.
\end{proof}

\subsection{Step 3: Fix the over-clearing}

In Step 3, we start with a price vector $p$ and allocation $x$ that satisfy the
conclusion of~\cref{lem:step2}.
Step 3 consists of two different procedures called the shift-and-burn procedure
and the pump-and-shift procedure, which will be called alternately until all of the
over-clearing has been eliminated.

We say that a good $j$ \emph{over-clears} if 
$\sum_{i} x_{i,j} > 1 + c_2 \cdot \epsc$, where $c_2$ is the constant arising from~\cref{lem:step2}, and we say that $j$ \emph{under-clears} if
$\sum_{i} x_{i,j} < 1 - c_2 \cdot \epsc$. Our goal is to eliminate the
over-clearing.

\paragraph{\bf Burning.}

In Step 3 we will allow buyers to \emph{burn} $\epsc$ money. If money
is burned, then it is not spent on any goods, and instead simply destroyed.
Formally, this is implemented by weakening the budget constraint on the buyers.
For each buyer $i$ we require that 
$$e_i - \epsc \le \sum_{j} p_j \cdot x_{i,j} \le e_i.$$
We say that a
buyer has \emph{fully burned} if $\sum_{j} p_j \cdot x_{i,j} = e_i - \epsc$.

\paragraph{\bf Invariant.} Let $c_3$ be a constant to be fixed later. Throughout Step 3, our goal will be to ensure that the following invariant is always satisfied.

For the current prices and allocation $(p,x)$ and for any buyer $i$, define $\marbpb_i(p,x)$ to be the \emph{marginal bang-per-buck} for buyer $i$, i.e., the bang-per-buck of the most desirable segment (in terms of bang-per-buck under the current prices $p$) that is not fully bought by buyer $i$. 
More formally, let $(j_1, k_1), (j_2, k_2), \dots$ be the
utility-function segments of buyer $i$ listed in bang-per-buck order with respect to $p$, with $(j_1, k_1)$ being the segment that has the best bang-per-buck. Let $a$ be the smallest index such that segment $(j_a,k_a)$ is not fully bought in $x$, i.e., this is the segment with highest bang-per-buck with that property. 
Then, $\marbpb_i(p,x) := s_{i,j_a,k_a}/p_{j_a}$. Note that Step 1 ensures that the prices are strictly positive and thus the division in the bang-per-buck expression is well defined. Furthermore, we will only ever increase prices and never decrease them, so they remain strictly positive throughout. Furthermore, we also define the \emph{window} for buyer $i$ at the current prices and allocation $(p,x)$ as the interval $[\marbpb_i(p,x)/(1 + c_3 \cdot \epsc), \marbpb_i(p,x)]$.

We say that the current prices and allocation $(p,x)$ satisfy the invariant if for every $i$, buyer $i$ does not spend any money on segments with bang-per-buck strictly worse than $\marbpb_i(p,x)/(1 + c_3 \cdot \epsc)$. Equivalently, the invariant is satisfied if buyer $i$ does not spend any money on segments with bang-per-buck strictly below the window. Note that segments with bang-per-buck lying strictly above the window are necessarily fully bought.

Ultimately, we will use the fact that any allocation satisfying this invariant must give an approximately-optimal allocation to each buyer (as long as every buyer spends most of their budget).

\begin{lemma}\label{lem:invariant-works}
Consider a buyer $i$ who satisfies the invariant and spends at least $e_i - \epsc$ of its total budget $e_i$. Then, the buyer receives a $(1 - c \cdot \epsc)$-optimal allocation, for some constant $c$.
\end{lemma}

\begin{proof}
For any amount of money $M > 0$, let $U^*(M)$ denote the optimal utility that buyer $i$ can achieve by spending $M$ at the current prices. Now let $M$ be the amount of money actually spent by buyer $i$ in the current allocation. By assumption, we have $e_i - \epsc \leq M \leq e_i$. The first observation we make is that $U^*(M) \geq (1 - \epsc/e_i) \cdot U^*(e_i)$. Indeed, by spending a fraction $(e_i - \epsc)/e_i$ of the total budget $e_i$, buyer $i$ can guarantee at least a fraction $(e_i - \epsc)/e_i$ of the optimal utility attainable with the full budget $e_i$, by concavity of the utility function.

Next, we compare the utility obtained by buyer $i$ in the current allocation with the optimal utility $U^*(M)$ achievable with the same expenditure. Since buyer $i$ satisfies the invariant, they do not spend any money on a segment with bang-per-buck strictly worse than $\marbpb_i(p,x)/(1 + c_3 \cdot \epsc)$. Recall that $\marbpb_i(p,x)$ is the bang-per-buck of the best segment that is not fully bought by buyer $i$. As a result, the current allocation yields utility at least $U^*(M)/(1 + c_3 \cdot \epsc)$, which is greater or equal to $(1 - c_3 \cdot \epsc) \cdot U^*(M)$. Putting everything together, we obtain that buyer $i$ has utility at least $(1 - c_3 \cdot \epsc) \cdot (1 - \epsc/e_i) \cdot U^*(e_i) \geq (1 - c \cdot \epsc) \cdot U^*(e_i)$, for some constant $c$, since $c_3$ is a constant and $e_i$ is upper bounded by a constant.
\end{proof}

\paragraph{\bf Preprocessing.}

The following lemma gives a
preprocessing procedure that ensures that the invariant holds at the start of
Step 3.

\begin{lemma}
There is a constant $c_3$ and a polynomial-time algorithm that, given the allocation $x$ arising
from~\cref{lem:step2}, produces an allocation $x'$ that
satisfies the invariant for some constant $c_3$, and is such that a $(1 - c_3
\cdot \deltac)$-fraction of the goods $(c_3 \cdot \epsc)$-clear in $x'$.
\end{lemma}

\begin{proof}
Step 2 modifies some of the buyers to buy optimal allocations, and
these buyers therefore buy segments in bang-per-buck order and therefore
satisfy the invariant. 

For each other buyer $i$, we can refer back to~\cref{lem:slopes} to see that
buyer $i$ spends at most $\epsc$ money on segments that lie strictly below their
window, and they need to spend at most $\epsc$ more money to fully buy any
segment that lies strictly above their window.

Therefore, we apply the following preprocessing step to each buyer $i$
that did not buy an
optimal allocation in Step 2.
\begin{enumerate}
\item For each segment $s$ of buyer $i$ that lies strictly above their window
that they do not currently fully buy, we move money from buyer $i$'s current
worst bang-per-buck segment $s'$ to $s$ until it is fully bought. If $s'$ does
not have enough money to do this, then we move all money from $s'$ to $s$ and
repeat the procedure until $s$ is fully bought. Note that it is possible that
$\marbpb_i(p,x)$ increases during this process, which will occur if we happen
to remove all of the money from the current segments that have bang-per-buck
equal to $\marbpb_i(p,x)$. This is not an issue, and in fact makes the
procedure terminate faster, because we no longer need to move money onto the
segments that join the new window. 

\item Then, for each segment $s$ of buyer $i$ that lies strictly below their
window, we move all money from $s$ to any of the segments whose bang-per-buck
is equal to $\marbpb_i(p,x)$ that are not currently fully bought. Note that if
we fully buy all segments with bang-per-buck $\marbpb_i(p,x)$, then
$\marbpb_i(p,x)$ will decrease, in which case, we repeat the procedure until
all money has been removed from the segments that lie beneath the current
window. 
\end{enumerate}

If, at the end of step 2, the invariant is not satisfied (which can happen only if
$\marbpb_i(p,x)$ changes in both step 1 and step 2), then we repeat the process
starting from step 1. To see that this will terminate in polynomial time,
observe that $\marbpb_i(p,x)$ changes only when all money is removed from a
segment in step 1, or when a segment is fully bought in step 2, and we can
never do either of these operations twice for the same segment. 

Let $x'$ be the allocation when the preprocessing terminates , and note
that $x'$ satisfies the invariant by construction. We must show that there
exists a constant $c_3$ such that a $(1 - c_3 \cdot \deltac)$-fraction
of the goods $(c_3 \cdot \epsc)$-clear in $x'$.
 
Recall from~\cref{lem:step2} that there exists a constant $c$ such that 
a $(1 - c \cdot \deltac)$-fraction
of the goods $(c \cdot \epsc)$-clear in $x$. We argue that the amount of money
moved by the preprocessing procedure is small enough to ensure that this does
not significantly change.

Note that at most $\epsc$ money is added to any segment that is below any 
window that is considered in step 1, and so in total each buyer needs to remove $M_i \cdot \epsc$ money from
their other goods, where $M_i$ is the total number of utility-function segments
of buyer $i$. In the worst case, each buyer that is interested in some good $j$
may remove $M_i \cdot \epsc$ from good $j$, and so $d \cdot M_i \cdot \epsc$
money is removed in total from good $j$. Thus, the amount bought of any good can
change by $d \cdot M_i \cdot \epsc/\pmin$ units during this operation, and
since $d$, $M_i$, and $\pmin$ are all constant, good $j$ continues to $(c'
\cdot \epsc)$-clear for some constant $c'$ after this operation.

On the other hand, at most $\epsc$ money is taken off any segment that is above
any window considered in step 2, and so at most $M_i \cdot \epsc$ money is removed by any one buyer.
In the worst case, for a good $j$, each buyer interested in good $j$ may place
all of their money on to good $j$, so at most $d \cdot M_i \cdot \epsc$ extra
money is spent. So the amount bought of any good can change by at most $d \cdot
M_i \cdot \epsc / \pmin$ during this operation. Therefore good $j$ continues to
$(c' \cdot \epsc)$-clear for some constant $c'$ after this operation.
\end{proof}

\paragraph{\bf The shift-and-burn procedure.}

We now define a flow problem that allows money to be shifted along 
segments with bang-per-buck lying in the window of each buyer.
\begin{itemize}
\item The vertices of the flow problem consist of all buyers in the market, all
goods in the market, and two special vertices $s$ and $t$, which will be the
source and sink vertices, respectively. That is, we set $V = G \cup B \cup \{s,
t\}$. 

\item For each good $j$ with $\sum_{i} x_{i,j} > 1 + c_2 \cdot \epsc$, i.e., each good
that is currently over-clearing, we add an edge from $s$ to $j$ with capacity 
$p_j \cdot (\sum_{i} x_{i,j} - (1 + c_2 \cdot \epsc))$.

\item For each buyer $i$, we add an edge from $i$ to $t$ with capacity $\max\left(0,
\sum_{j} p_j \cdot x_{i,j} - (e_i - \epsc)\right)$, which represents the amount of money that
buyer $i$ can still burn.

\item For each buyer $i$ and good $j$ and each utility function segment $k$, if
$(j,k)$ has bang-per-buck lying in the window of buyer $i$, then we add the following two edges.
\begin{itemize}
\item We add an edge from good $j$ to buyer $i$ with capacity $a_{j,i,k}$, where
$$a_{j,i,k} = p_j \cdot \min\left(\max\left(x_{i,j} - \sum_{k' = 1}^{k-1} \ell_{i,j,k'}, \;
\ell_{i,j,k}\right), 0\right)$$
is the amount of segment $k$ that is used by $x_{i,j}$. 

\item We add an edge from buyer $i$ to good $j$ with capacity 
$ a_{i,j,k} = p_j \cdot \ell_{i,j,k} - a_{j,i,k}$. 
\end{itemize}
This means that the flow problem is defined over a multi-graph, since
there can be multiple segments $(j,k)$ of good $j$ that lie in the window for buyer $i$.
\end{itemize}
We then find a maximal flow $f$ from $s$ to $t$, which can be computed in
polynomial time.

We use the flow $f$ to create a new allocation $x'$ in the following way. 
Let $f_{i,j}$ denote the total amount of flow from buyer $i$ to good $j$, where flow on edges from $i$ to $j$ is counted positively and flow on edges from $j$ to $i$ is counted negatively. Recall that only edges corresponding to segments for good $j$ that are in the window for buyer $i$ are included in the graph. We set $x'_{i,j} = x_{i,j} +
f_{i,j}/p_j$.

\begin{lemma}\label{lem:shift-and-burn-invariant}
The shift-and-burn procedure maintains the invariant. Furthermore, it holds that $\marbpb_i(p,x') \leq \marbpb_i(p,x)$ for all buyers $i$.
\end{lemma}
\begin{proof}
Fix any buyer $i$. Since we have not modified the prices, the ordering of segments with respect to bang-per-buck has not changed. Let the \emph{old window} refer to the window with respect to $(p,x)$, and the \emph{new window} refer to the window with respect to $(p,x')$. Since we have not changed the amount of money spent on segments lying strictly above the old window, they are still fully bought in $x'$, and thus it must be that $\marbpb_i(p,x') \leq \marbpb_i(p,x)$. As a result, the new window lies weakly below the old window. 
Now, consider any segment with bang-per-buck lying strictly below the new window. By the above, it also lies strictly below the old window. By construction of the flow problem, in allocation $x'$ buyer $i$ does not spend any money on that segment. Indeed, the amount of money that can be added to the good by buyer $i$ in the flow is upper bounded by the total amount of money that can be borne by the segments of that good in the old window. Therefore, the invariant is maintained.
\end{proof}

In addition, the shift-and-burn procedure weakly decreases the over-clearing of every good, while not changing the clearing of all other goods.

\begin{lemma}
The new allocation $x'$ satisfies the following two statements.
\begin{itemize}
\item If good $j$ over-cleared in $x$, then $(1+c_2 \cdot \epsc) \le \sum_i x'_{i,j} \le \sum_i x_{i,j}$. 
\item If good $j$ did not over-clear in $x$, then $\sum_i x'_{i,j} = \sum_i x_{i,j}$.
\end{itemize}
\end{lemma}
\begin{proof}
For every good $j$ we have that the amount of flow incoming at $j$ and the
amount of flow outgoing at $j$ are equal, by definition. If good $j$ does not
have an incoming edge from $s$, then the total demand of that good must be
unchanged, since all money routed on to the good by the flow is matched with
money that is removed by other buyers. If a good $j$ does have an incoming edge
from $s$, then the total amount of money spent on that good is decreased by
exactly the amount of money that flows from $s$ to $j$. Hence, the demand of
good $j$ may decrease, but it cannot decrease below $(1+c_2 \cdot \epsc)$ due
to the capacity of the edge between $s$ and $j$.
\end{proof}

If there exists a buyer $i$ such that its window in allocation $x'$ is different than in allocation $x$ (i.e., $\marbpb_i(p,x') < \marbpb_i(p,x)$), then we run the shift-and-burn procedure again (with the graph defined by the new windows). As long as there is such a buyer we keep re-running the shift-and-burn procedure. We stop once we reach the point where the windows for all buyers no longer change and let $x'$ denote the final allocation we obtain.

\begin{lemma}
At most a polynomial number of shift-and-burn iterations are required to find the final allocation.
\end{lemma}
\begin{proof}
Observe that $\marbpb_i(p,x)$ is defined to be the bang-per-buck of some segment
belonging to buyer $i$ under price vector $p$, and also observe that a
shift-and-burn iteration does not change the price vector $p$. Furthermore, by
\cref{lem:shift-and-burn-invariant} we have that $\marbpb_i(p,x)$ weakly
decreases for each buyer $i$ in a shift-and-burn step, and, if we do not
terminate, then there exists some buyer $i$ for which $\marbpb_i(p,x)$ strictly
decreases. Therefore, the total number of iterations is bounded by the total
number of segments in the market, and these segments are part of the input of the problem, so the number of iterations is polynomial.
\end{proof}

If the final allocation $x'$ happens to remove all over-clearing in the market,
then we stop and skip ahead to Step 4. Otherwise we proceed below.

\paragraph{\bf The restricted flow graph.}

After constructing $x'$, we then construct the \emph{restricted flow graph} in
the following way. Consider the flow problem that was solved in the last iteration of the shift-and-burn procedure. Note that the window of each buyer did not change as a result of that last iteration. We will select a subset of the vertices of this graph to define a \emph{restricted} flow graph. Namely:
\begin{itemize}
\item We add all goods that over-clear in $x'$ to the restricted flow graph.
\item If good $j$ is in the restricted flow graph, and $i$ is a buyer such that
there exists an edge from $j$ to buyer $i$ in the flow problem, then we add buyer $i$ to the restricted flow graph if the following condition is satisfied: $f_{j,i} < a_{j,i}$, where $a_{j,i}$ is the sum of the capacities of all edges from good $j$ to buyer $i$. This condition is equivalent to saying that in allocation $x'$ buyer $i$ spends some strictly positive amount of money on a segment of good $j$ that appears in the window of buyer $i$ in $x'$.
\item If buyer $i$ is in the restricted flow graph, and $j$ is a good such that
there exists an edge from $i$ to good $j$ in the flow problem, then we add good $j$ to the restricted flow graph if the following condition is satisfied: $f_{i,j} < a_{i,j}$, where $a_{i,j}$ is the sum of the capacities of all edges from buyer $i$ to good $j$. This condition is equivalent to saying that in allocation $x'$ there exists a segment of buyer $i$ for good $j$ appearing in the window (in $x'$) of buyer $i$ that is not fully bought.
\end{itemize}
The restricted flow graph is defined to be the subset of the flow problem that
we obtain by repeating the above steps until convergence. Intuitively, the
restricted flow graph consists of all goods and buyers that can be reached from
the source by a path of edges that are not saturated by the flow $f$.

\begin{lemma}
\label{lem:rfg}
Every buyer $i$ in the restricted flow graph has fully burned, i.e.,
$\sum_{j} p_j \cdot x'_{i,j} = e_i - \epsc$.
\end{lemma}
\begin{proof}
This can easily be shown by contradiction. If a buyer $i$ has not
fully burned, then the edge from $i$ to $t$ is not saturated in the flow $f$. Furthermore, there is a path of unsaturated edges from some over-clearing good $j$ to $i$. Since $j$ still over-clears in $x'$, the
edge from $s$ to good $j$ must be unsaturated. We could therefore increase the flow
from $s$ to $t$ by increasing the flow along the edges that connect $s$ to $i$
and along the edge from $i$ to $t$, but this would contradict the maximality of
$f$.
\end{proof}

\paragraph{\bf The pump-and-shift procedure.}

In the pump-and-shift procedure we take the allocation $x'$ and the restricted
flow graph that are produced by the shift-and-burn procedure and we do the following.

\begin{enumerate}
\item Let $G_r$ be the set of goods that appear in the restricted flow graph. 
We construct a new price vector $p'$ by \emph{pumping} the prices of the goods
in $G_r$ by a factor of $(1 + \epsc)$. Formally, we
construct $p'$ in the following way.
\begin{equation*}
p'_j = \begin{cases}
p_j \cdot (1+\epsc) & \text{if $j \in G_r$,} \\
p_j & \text{otherwise.}
\end{cases}
\end{equation*}

\item 
We now need to \emph{shift} the allocations of the buyers. Fix a buyer $i$. For each segment $(j,k)$ of buyer $i$, let $x'_{i,j,k}$ denote the amount of segment $(j,k)$ that is bought by buyer $i$ in allocation $x'_i$. For the ensuing arguments, it will be convenient to introduce the notion of a \emph{pseudo-allocation}. A pseudo-allocation corresponds to an allocation where buyers are allocated parts of utility function segments. This is not necessarily an allocation, because an allocation requires each buyer to fully buy segment $(j,k)$ before spending any money on segment $(j,k+1)$. Let $q_{i,j,k} := p_j \cdot x'_{i,j,k}$ denote the amount of money spent on segment $(j,k)$ by buyer $i$ in allocation $x'$ at prices $p$. Define the pseudo-allocation $x''$ by setting $x''_{i,j,k} := q_{i,j,k}/p'_j$, i.e., the pseudo-allocation that is obtained by spending exactly the same amount of money as before on each segment, but now at the new prices $p'$. Note that by construction we have $x''_{i,j,k} \leq x'_{i,j,k}$ for all $i,j,k$.

Next, we will describe a procedure that will modify the pseudo-allocation $x''$ such that at the end (i) $x''$ is an allocation, (ii) $x''_{i,j} \leq x'_{i,j}$ for all $j$, (iii) $\sum_{j} p_j' \cdot x''_{i,j} = \sum_{j} p_j \cdot x'_{i,j}$, and (iv) $(p',x'')$ satisfies the invariant. For any buyer $i$ that is not interested in any good lying in the restricted flow graph, we do not perform any modification on its allocation. Indeed, there is no need to modify the allocation of that agent, since it already satisfies all the desiderata and none of the goods that are bought were pumped.

For any buyer $i$ who is interested in at least one good lying in the restricted flow graph, we perform the following \emph{shift} procedure. The procedure considers the segments of buyer $i$ in order of bang-per-buck at the new prices $p'$, starting from the segment with highest bang-per-buck and moving down. When the procedure reaches a segment $(j,k)$, we do the following. If $x''_{i,j,k} < x'_{i,j,k}$, then we remove money from the worst bang-per-buck segment $(\bar{j},\bar{k})$ of buyer $i$ that has money on it and shift it onto segment $(j,k)$ until $x''_{i,j,k} = x'_{i,j,k}$, or until we have removed all the money from segment $(\bar{j},\bar{k})$. In the latter case, if we still have $x''_{i,j,k} < x'_{i,j,k}$, then we repeat with the worst bang-per-buck segment on which buyer $i$ now still spends money. If at any point, there is no segment with bang-per-buck strictly worse than $(j,k)$, the procedure terminates. Otherwise, once we have shifted enough money onto $(j,k)$ such that $x''_{i,j,k} = x'_{i,j,k}$, we move on to the next segment in decreasing bang-per-buck order.

Condition (i) holds because for any segment $(j,k)$ with $x''_{i,j,k} < x'_{i,j,k}$, the algorithm will not terminate as long as there is a segment $(j,k')$, $k' > k$, with a positive amount of money on it. This is because segment $(j,k')$ necessarily has a strictly worse bang-per-buck than segment $(j,k)$ due to the SPLC utility function. Condition (ii) holds because we have $x''_{i,j,k} \leq x'_{i,j,k}$ at the beginning of the procedure, and the algorithm maintains this throughout. Condition (iii) holds because we have $\sum_{j,k} p_j' \cdot x''_{i,j,k} = \sum_{j,k} p_j \cdot x'_{i,j,k}$ before running the procedure, and the procedure only shifts money between segments of each buyer $i$. The fact that condition (iv) holds will be proved below.
\end{enumerate}

The following lemma states that the pump-and-shift procedure does not introduce
any new over-clearing goods. 

\begin{lemma}
For every good $j$, if $j$ does not over-clear in $x'$, then it does not
over-clear in $x''$. 
\end{lemma}
\begin{proof}
By condition (ii) above, for every good $j$ we have that $\sum_{i}
x''_{i,j} \le \sum_{i} x'_{i,j}$. 
If good $j$ does not over-clear in $x'$ then we have $\sum_{i} x'_{i,j} \le 1
+ c_2 \cdot \epsc$, and therefore we have 
$\sum_{i} x''_{i,j} \le 1 + c_2 \cdot \epsc$, implying that $j$ also does not over-clear in $x''$. 
\end{proof}

The following two lemmas imply that condition (iv) holds, i.e., the invariant continues to hold after the
pump-and-shift operation.

\begin{lemma}
In allocation $x''$, any buyer $i$ outside the restricted flow graph satisfies the invariant.
\end{lemma}

\begin{proof}
First, consider a buyer $i$ who is not in the restricted flow graph. By definition of the restricted flow graph, this means that for any segment $(j,k)$ lying in the window of buyer $i$ with respect to $(p,x')$, we have that $x_{i,j,k}' > 0 \implies p_j' = p_j$. Indeed, if there existed a segment $(j,k)$ lying in the window and such that $x_{i,j,k}' > 0$ and $p_j' > p_j$, then good $j$ is in the restricted flow graph and buyer $i$ would also have been added to the graph by construction.

Let $W$ denote the window with respect to $(p,x')$, i.e., $W = [b/(1 + c_3 \cdot \epsc), b]$, where $b := \marbpb_i(p,x')$. Now, consider the segments with their bang-per-buck after the prices have been updated to $p'$. Let $S$ denote the set of segments $(j,k)$ with $p_j' > p_j$ and with bang-per-buck strictly larger than $b$, where the bang-per-buck is computed with the new prices $p'$. Finally, let $M$ denote the total amount of money spent by buyer $i$ on segments lying in $W$ (where the bang-per-buck of the segments is computed with respect to $p'$) in the pseudo-allocation $y$ obtained after pumping the prices, but before the shifting procedure is executed. Note that since the pumping ratio $1+\epsc$ is smaller than the ratio defining the window $W$, namely $1+c_3\cdot \epsc$, and thanks to the observation in the previous paragraph, it must be that in the pseudo-allocation $y$ buyer $i$ is not spending any money on segments with bang-per-buck (with respect to $p'$) strictly below $W$.

Now consider two cases. First, if the amount of money $M$ is not sufficient to ensure that all the segments in $S$ are fully bought, then the shifting procedure will terminate while handling a segment $(j,k)$ such that (a) all segments with bang-per-buck strictly better than $(j,k)$ are fully bought, and (b) buyer $i$ spends no money on any segment with bang-per-buck strictly worse than $(j,k)$. As a result, the invariant immediately holds. Here, we crucially used the fact that any segment with bang-per-buck (with respect to $p'$) strictly above $b$ was fully bought in allocation $x'$.

The remaining case is when the amount of money $M$ suffices to ensure that all segments in $S$ are fully bought at the new prices. In that case, the shifting procedure will indeed fully buy the segments in $S$. As a result, all segments with bang-per-buck (with respect to $p'$) strictly above $b$ will be fully bought in the final allocation $x''$. In particular, we have $\marbpb_i(p',x'') \leq b$. Recall that, as argued above, buyer $i$ was not spending any money on segments with bang-per-buck (at the new prices $p'$) strictly below $W$. As a result, this is still the case after the shifting procedure, i.e., no money is spent on segments with bang-per-buck strictly lower than $b/(1 + c_3 \cdot \epsc)$. In particular, no money is spent on segments with bang-per-buck strictly lower than $\marbpb_i(p',x'')/(1 + c_3 \cdot \epsc)$, and thus the invariant holds.
\end{proof}

\begin{lemma}
In allocation $x''$, any buyer $i$ inside the restricted flow graph satisfies the invariant.
\end{lemma}

\begin{proof}
Consider a buyer $i$ who is in the restricted flow graph. By construction of the restricted flow graph, we have that any segment $(j,k)$ lying in the window of buyer $i$ with respect to $(p,x')$ must satisfy $x_{i,j,k}' < \ell_{i,j,k} \implies p_j' > p_j$. Indeed, the existence of a segment $(j,k)$ lying in the window and such that $x_{i,j,k}' < \ell_{i,j,k}$, implies that good $j$ is added to the restricted flow graph, and thus its price is pumped.

Let $W$ denote the window with respect to $(p,x')$, but translated by a factor $1+\epsc$. More formally, we define $W := [b/(1 + c_3 \cdot \epsc), b]$, where $b := \marbpb_i(p,x')/(1+\epsc)$. Note that for any segment $(j,k)$ with bang-per-buck (with respect to $p'$) lying strictly above $W$, we have $x_{i,j,k}' = \ell_{i,j,k}$, i.e., the segment was fully bought before the pump and shift procedure. Indeed, any segment $(j,k)$ with bang-per-buck (with respect to $p'$) lying strictly above $W$ must satisfy one of the following conditions (a) the segment had bang-per-buck (with respect to $p$) that was strictly above $\marbpb_i(p,x')$, or (b) the segment had bang-per-buck (with respect to $p$) lying in the window of buyer $i$ with respect to $(p,x')$. Here we used the fact that $W$ was obtained by translating by a factor of $1+\epsc$ the window of buyer $i$ with respect to $(p,x')$. Since this factor is smaller than $1+c_3 \cdot \epsc$, the factor defining the size of the windows, it follows that the window of buyer $i$ with respect to $(p,x')$ and its translation $W$ must overlap. So, (a) and (b) are indeed the only two possibilities. Now, in case (a) it is immediate that the segment was fully bought before the pump and shift procedure, since it was above the window. For case (b), assume for the sake of contradiction that segment $(j,k)$ was not fully bought. Then, as observed above, it must be that the price of good $j$ was pumped. But this means that segment $(j,k)$ now has bang-per-buck (with respect to the new prices $p'$) lying in the translated window $W$, a contradiction to the assumption that $(j,k)$ is strictly above $W$.

Now, consider two cases. First, assume that the shifting procedure runs out of money before having been able to fully buy all segments lying strictly above $W$ (with respect to the new prices $p'$). Then the shifting procedure will terminate while handling a segment $(j,k)$ such that (a) all segments with bang-per-buck strictly better than $(j,k)$ are fully bought, and (b) buyer $i$ spends no money on any segment with bang-per-buck strictly worse than $(j,k)$. As a result, the invariant immediately holds. Here, we crucially used the fact that any segment with bang-per-buck (with respect to $p'$) strictly above $W$ was fully bought in allocation $x'$.

The remaining case is when the shifting procedure fully buys all segments lying strictly above $W$ (with respect to the new prices $p'$). As a result, the invariant holds, since no money will be spent on any segment lying strictly below $W$, because such segments were strictly below the window with respect to $(p,x')$, and thus had no money spent on them in allocation $x'$.
\end{proof}

The following lemma states that the shift procedure can
introduce under-clearing only in goods that are adjacent to the restricted flow
graph.

\begin{lemma}
\label{lem:s3uc}
Let $j$ be any good that under-clears in $x''$ but does not
under-clear in $x'$. Then, there exists a buyer who is interested in good $j$ and in a good that lies in the restricted flow graph.
\end{lemma}
\begin{proof}
Recall that $G_r$ denotes the set of goods that appear in the restricted flow graph.
Let $j$ be a good such that $j \not\in G_r$ and such that no buyer that is
interested in good $j$ is interested in a good in the restricted flow graph. This means that we did not run the shift procedure on any of the buyers who are interested in good $j$. As a result, the amount of good $j$ that is bought in allocation $x''$ is the same as was already bought in allocation $x'$. So good $j$ can only under-clear in $x''$ if it already did so in allocation $x'$, i.e., $j \notin U$.

Thus, if a good $j$ under-clears in $x''$ when it did not under-clear in $x'$, there
must be at least one buyer that is interested in $j$ and is interested in a good
$j'$ in the restricted flow graph.
\end{proof}

\paragraph{\bf Combining into Step 3.}

In Step 3 we alternate the shift-and-burn and pump-and-shift procedures until
there are no over-clearing goods left in the market. The following lemma states
that we can alternate these procedures at most a constant number of times
before we terminate. 

\begin{lemma}
\label{lem:steps}
The shift-and-burn and pump-and-shift procedures can be alternated at most a
constant number of times before all over-clearing in the market is eliminated.
Once the algorithm has terminated, for all goods $j$ we have $p_j \le c'$ for
some constant $c'$. 
\end{lemma}
\begin{proof}
Let $j$ be a good that over-clears at the start of Step 3. The price of $j$ at
the start of Step 3 is at least $\pmin = \bmin/(8 \cdot d \cdot \kappa)$, since
that was the lower bound on the prices in Step 1, and Step 2 did
not change the prices. 
Recall that $\pmax = 2 \cdot d \cdot \bmax$, and recall that both $\pmin$ and
$\pmax$ are constant. As long as the good is still over-clearing, each iteration of the pump-and-shift procedure increases the price of
good $j$ by a factor of $(1 + \epsc)$ where $\epsc$ is constant. Let $c$ be a
constant such that $(1 + \epsc)^c \cdot \pmin > \pmax$. So, if the procedure
has not terminated after $c$ rounds, the price of any good $j$ that is still
over-clearing must exceed $\pmax$.

We argue that if the price of a good $j$ is greater than or equal to $\pmax$,
then good $j$ cannot over-clear, which proves that the algorithm must terminate after at most $c$ rounds. To see why, observe that there are at most $d$
buyers that are interested in good $j$, and each of those buyers can spend at
most $\bmax$ money. If all such buyers spend their entire budget on good $j$,
then they can buy at most $d \cdot \bmax / \pmax = 1/2$ units of it, and $1/2 <
1 + c_2 \cdot \epsc$, so the good cannot over-clear.

Finally, to obtain an upper bound on the prices, it suffices to note that all
prices $p_j$ satisfy $p_j \le \pmax$ at the start of Step 3, since that was the
upper bound enforced in the \gcircuitp instance in Step 1, and prices have not
been modified since then. As argued previously, in Step~3 each price is
multiplied by at most $(1 + \epsc)^c$ for some constant $c$, and since $\pmax$,
$\epsc$, and $c$ are all constant, we therefore obtain a constant upper bound on
the prices. 
\end{proof}

As shown in~\cref{lem:s3uc}, the price for running this procedure is that
some under-clearing can be introduced in goods that are adjacent to one of the
restricted flow graphs that are pumped. The following lemma states that the
number of goods that under-clear after Step 3 is still relatively small.

\begin{lemma}
\label{lem:s3oc}
There is a suitably small constant value for $\epsc$ such that 
at most a $c \cdot (\deltac/\epsc)$-fraction of the goods under-clear at the end of
Step 3, for some constant $c$.
\end{lemma}
\begin{proof}
By~\cref{lem:step2} there exists a constant $c'$ such that, at the start of
Step 3 a $(1 - c' \cdot \deltac)$-fraction of the goods $(c' \cdot \epsc)$-clear. 
Thus, there are at most $c' \cdot \deltac \cdot |G|$ goods that over-clear, and
the amount of money spent on over-clearing goods can be at most
$M = c' \cdot \deltac \cdot |G| \cdot d \cdot \bmax$, which would occur when every
buyer interested in an over-clearing good spent their entire budget on that
good. 

By~\cref{lem:rfg}, after the shift-and-burn procedure, the restricted flow
graph only contains buyers who have fully burned. Note that the shift-and-burn
procedure moves, via a flow, money from over-clearing goods to the money that
is burned by the buyers, and each buyer can burn at most $\epsc$ money. Thus, if $M / \epsc$ many buyers fully burn, then there
cannot be any over-clearing left in the market. This implies that the set of
buyers $B_r$ who appear in any restricted flow graph at any point in Step 3
must satisfy $|B_r| \le M / \epsc$. \cref{lem:s3uc} states that a good can
become under-clearing in Step 3 only if it 
shares a buyer with
some good that
appears in some restricted flow graph that is considered in Step 3. Thus there
can be at most $d^2 \cdot M / \epsc$ such goods. So in total, the number of
under-clearing goods that are introduced in Step 3 is at most
$$c' \cdot d^3 \cdot \bmax \cdot (\deltac/\epsc) \cdot |G|,$$
and since $d$, $c'$, and $\bmax$ are all constant, this proves the claim.
\end{proof}

Finally we can state the following lemma, which summarizes the result of
applying Step 3. 

\begin{lemma}
\label{lem:step3}
If $\epsc$ is a suitably small constant, then 
in polynomial time we can find an allocation $x$ and a price vector $p$ such
that there is a constant $c$ for which the following hold.
\begin{enumerate}
\item Every buyer receives a $(1 - c \cdot \epsc)$-optimal allocation.
\item Every buyer $i$ spends at least $e_i - \epsc$ and at most $e_i$ of their
money.
\item At least a $(1 - c \cdot (\deltac/\epsc))$-fraction of the goods $j$ satisfy $\sum_{i}
x_{i,j} \ge 1 - c \cdot \epsc$, and all goods $j$ satisfy $\sum_{i} x_{i,j} \le
1 + c \cdot \epsc$.
\item For each good $j$ we have $p_j \le c$.
\end{enumerate}
\end{lemma}
\begin{proof}
The first point follows from \cref{lem:invariant-works}, where we note that both
shift-and-burn and pump-and-shift maintain the invariant, and that
each buyer can burn at most $\epsc$ money, meaning that the precondition of
\cref{lem:invariant-works} is satisfied. Point 2 follows from the fact
that, by
design, the total flow lost 
in the edges from $i$ to $t$ across all of the
rounds of Step 3 is at most $\epsc$. Point 3 follows from~\cref{lem:steps} and~\cref{lem:s3oc}. Point 4 follows from~\cref{lem:steps}. 
\end{proof}

\subsection{Step 4: Fix the under-clearing}

In Step 4 we receive an allocation $x$ and a price vector $p$ that satisfy~\cref{lem:step3} for some constant~$c$. 
Let $U = \{ j \; : \; \sum_{i} x_{i,j} < 1 - c \cdot \epsc \}$ be the set of
under-clearing goods in $x$. 
We use the following three-step procedure to fix the under-clearing.

\begin{enumerate}
\item For each good $j$ that under-clears in $x$, we find a buyer $i$ who has
burned some amount of money, and we force buyer $i$ to spend money on good $j$
instead of burning it, and we move money onto good $j$ until either good $j$
no-longer under-clears, or until buyer $i$ runs out of burned money.
We repeat this process until we either run out of under-clearing goods, or
until we run out of buyers that have burned money. 

\item If there are no more under-clearing goods, but there are still buyers
that have burned money, then we allocate that money across all of the goods,
with each good $j$ receiving money in equal proportion. Formally, we calculate
the amount of remaining burned money as $P = \max(\sum_{i} (e_i - \sum_{j}
x^1_{i,j}), 0)$, where $x^1$ is the allocation at the end of the step above. We
then allocate this money so that $P/|G|$ extra money is spent on each good,
with the allocation of how the buyers spend money on the goods being chosen
arbitrarily.

\item If there are no more buyers with burned money, but there are still
under-clearing goods, then we force all buyers to remove money from each good, 
 where money is removed from all goods $j$ in equal proportion, and then spend
that money on the under-clearing goods in order to clear them. 
Formally, if $x^2$ is the allocation after the step above, then for each good
$j$ we define $D(j) = \max(1 - c \cdot \epsc - \sum_{i} x^2_{i,j}, 0)$ to be the
amount by which good $j$ under-clears. We then calculate 
$Q = \sum_{j} (p_j \cdot D(j))$ to be the amount of money that is needed to
fix the under-clearing. 

We then force each buyer $i$ to remove $Q/|B|$ money from their allocation,
where we remove money from goods in bang-per-buck order, starting with the
worst utility-function segment that is bought in the allocation $x^2$. Then we force those
buyers to allocate that money to the under-clearing goods, with the assignment
of buyers to under-clearing goods being chosen arbitrarily. 
\end{enumerate}
Let $x'$ be the allocation that results from this procedure.

The following lemma bounds the amount of money that is needed to clear the
goods in $U$. 
\begin{lemma}
\label{lem:totuc}
There exists a constant $c'$ such that 
$$\sum_{j \in U} \left(p_j \cdot \left(1 - c \cdot \epsc - \sum_i x_{i,j}
\right)\right) \le c' \cdot (\deltac/\epsc) \cdot |G|.$$
\end{lemma}
\begin{proof}
In the worst case we have $\sum_i x_{i,j} = 0$ for any good $j \in U$. The
price of each good is at most $\pmax$ at the start of Step 3 due to the bounds
imposed during Step 1, and by~\cref{lem:steps} each good is pumped at most
constantly many times during Step 3, so we have $p_j \le c_1 \cdot \pmax$ for
some constant $c_1$. Since $\pmax$ is constant, there exists a constant $c_2$ such that, for each
good $j \in U$ we have $p_j \cdot \left(1 - c \cdot \epsc - \sum_i x_{i,j}
\right) \le p_j \cdot 1 \le c_2$. By~\cref{lem:step3} there exists a
constant $c_3$ such that $|U| \le c_3 \cdot (\deltac/\epsc) \cdot |G|$.
Therefore we have that $\sum_{j \in U} \left(p_j \cdot \left(1 - c \cdot \epsc
- \sum_i x_{i,j} \right)\right) \le c_2 \cdot c_3 \cdot (\deltac/\epsc) \cdot
  |G|$, so setting $c' = c_2 \cdot c_3$ suffices to prove the claim.
\end{proof}

The following lemma provides a lower bound on the utility that each buyer
receives at the end of Step 4. 

\begin{lemma}
\label{lem:finalopt}
There are sufficiently small constant values for $\epsc$ and $\deltac$ such that
every buyer receives a $(1 - c \cdot (\deltac/\epsc)) \cdot (1 - c' \cdot \epsc)$-optimal allocation in
$x'$ for some constants $c$ and $c'$.
\end{lemma}
\begin{proof}
Steps 1 and 2 of the procedure that produces $x'$ can clearly only weakly
increase the utility obtained by each buyer, since they involves making buyers
spend money that they had burned, which provided zero utility, on actual goods,
which have non-negative utility. 

Step 3 of the procedure that produces $x'$ on the other hand has buyers shifting money from goods that gave them
non-zero utility, to goods that potentially give them zero utility. 
By~\cref{lem:totuc}, we have that the total money removed by each buyer
$i$ is at most $Q/|B| \le (c_1 \cdot (\deltac/\epsc) \cdot |G|)/|B|$ for some constant
$c_1$. Since the graph has constant degree $d$, we have that $|G|/|B| \le d$,
and so $Q/|B| \le c_2 \cdot (\deltac/\epsc)$ for some constant $c_2$. 

In $x'$
each
buyer $i$ spends their full budget $e_i \ge \bmin$, and we have that $\bmin$ is
constant, positive, and bounded away from zero. So since $\deltac$ can be chosen
to be an
arbitrarily small constant, we have that $u_i(x') \ge (1- c_3 \cdot
(\deltac/\epsc)) \cdot u_i(x)$ for some constant $c_3$. 
Since $x$ delivers a $(1 - c_4 \cdot \epsc)$-optimal allocation for some
constant $c_4$ by~\cref{lem:step3}, this proves the claim.
\end{proof}

The following lemma states that all goods clear at the end of Step 4.

\begin{lemma}
\label{lem:finalclear}
Every good $(c \cdot \epsc + c' \cdot (\deltac/\epsc))$-clears in $x'$ for some
constants $c$ and $c'$. 
\end{lemma}
\begin{proof}
First consider a good $j$ that under-clears in $x$. Clearly by construction the
procedure that generates $x'$ ensures that good $j$ does not under-clear in
$x'$. 

Next consider a good $j$ that did not under-clear in $x$. Steps 1 and 2 of of
the procedure that generates $x'$ do not change the allocation of good $j$.
Step 3 of that procedure may change the allocation, however. In the worst case,
each of the buyers who are interested in good $j$ may remove $Q/|B|$ money from
good~$j$, and from~\cref{lem:totuc} we have that $Q/|B| \le c_1 \cdot
(\deltac/\epsc) \cdot |G|/|B|$ for some constant $c_1$. Since the degree of the
market is $d$ we therefore have $Q/|B| \le c_1 \cdot (\deltac/\epsc) \cdot d$, and
since $d$ is constant there exists a constant $c_2$ such that 
$Q/|B| \le c_2 \cdot (\deltac/\epsc)$.
Since there are at most $d$ buyers interested in good
$j$, where $d$ is a constant, we get that at most $d \cdot Q/|B| \le c_3 \cdot
(\deltac/\epsc)$ money is removed from good $j$, for some constant $c_3$.

The price of good $j$ is at least $\pmin$, since the \gcircuitp instance
imposed that as a lower bound, and Step 3 can only increase prices, and $\pmin$
is a constant. Thus the amount of good $j$ that is no longer bought is at most 
$(d \cdot Q/|B|)/\pmin \le c_4 \cdot 
(\deltac/\epsc)$ for some constant $c_4$, where we have used the fact that
$\pmin$ is constant. 

In summary, at least $(1 - c \cdot \epsc)$ units of good $j$ were bought at the
start of Step 4 by~\cref{lem:step3}. We have shown that in Step 4 at most 
$c_4 \cdot (\deltac/\epsc)$ units of good $j$ may no-longer be bought, which
proves the claim.
\end{proof}

The following lemma states the outcome of Step 4.

\begin{lemma}
\label{lem:step4}
If $\epsc$ and $\deltac$ are suitably small constants, then 
in polynomial time we can find an allocation $x$ and a price vector $p$ such
that there is a constant $c$ for which the following hold.
\begin{enumerate}
\item Every buyer receives a $(1 - c \cdot (\deltac/\epsc)) \cdot (1 - c \cdot \epsc)$-optimal allocation and satisfies their budget constraint with equality.
\item Every good $(c \cdot \epsc + c \cdot (\deltac/\epsc))$-clears.
\end{enumerate}
\end{lemma}
\begin{proof}
This follows from~\cref{lem:finalopt} and~\cref{lem:finalclear}, while
picking $c$ to be the maximum of any of the constants arising from those
lemmas. 
\end{proof}

\subsection{Step 5: Obtain exact clearing}

In Step 5 we will define $f(\epsc, \deltac) = c \cdot \epsc + c \cdot
(\deltac/\epsc)$ and $g(\epsc, \deltac) = (1 - c \cdot (\deltac/\epsc)) \cdot
(1 - c \cdot \epsc)$, where $c$ is the constant arising from
\cref{lem:step4}. With these definitions in place, \cref{lem:step4}
states that at the start of Step 5 we have an allocation $x$ and a price vector
$p$ such that each buyer receives a $g(\epsc, \deltac)$-optimal allocation and
satisfies their budget constraint with equality, and such that each good $f(\epsc,
\deltac)$-clears. The goal of Step 5 is to modify $x$ and $p$ so that each good
exactly clears, while not losing too much in the optimality of the bundles.

Our procedure for Step 5 proceeds using three sub-steps.

\paragraph{\bf Step 5.1: Remove money so that exactly $(1 - f(\epsc,
\deltac))$ units of each good are demanded.}

For each good $j$, let $D_j = \sum_{i} x_{i,j}$ be the total amount of good $j$
that is demanded under the allocation $x$, and observe that we have $1 -
f(\epsc, \delta) \le D_j \le 1 + f(\epsc, \deltac)$ for all $j$ by
\cref{lem:step4}. Create the allocation $x^1$ so that, for each buyer $i$
and each good $j$, we set $x^1_{i,j}  = (1 - f(\epsc, \deltac))/D_j \cdot
x_{i,j}$.

We define 
$$g'(\epsc, \deltac) = \frac{1- f(\epsc, \deltac)}{1 + f(\epsc, \deltac)} \cdot g(\epsc, \deltac).$$ 
The following lemma shows
that the demand of each good under $x^1$ is exactly $1 - f(\epsc, \deltac)$, and
that each buyer receives a $g'(\epsc, \deltac)$-optimal bundle under $x^1$ and
$p$.

\begin{lemma}
\label{lem:step51}
In the allocation $x^1$ under price vector $p$ we have the following.
\begin{itemize}
\item For each good $j$ we have $\sum_{i} x^1_{i,j} = 1 - f(\epsc, \deltac)$.
\item Each buyer receives a $g'(\epsc, \deltac)$-optimal
bundle.
\end{itemize}
\end{lemma}
\begin{proof}
We start with the first claim. For each good $j$ we have
\begin{align*}
\sum_{i} x^1_{i,j} &= \sum_{i} \left(\frac{1 - f(\epsc, \deltac)}{D_j}\right) \cdot x_{i,j} \\
& = \left(\frac{1 - f(\epsc, \deltac)}{D_j} \right) \cdot \sum_{i} x_{i,j} \\
& = 1 - f(\epsc, \deltac).
\end{align*}
For the second claim, observe that since each $D_j \le 1 + f(\epsc,
\deltac)$, each buyer retains at least a fraction $(1 - f(\epsc, \deltac))/(1 + f(\epsc,
\deltac))$ of their allocation for each good. Since the utility
functions are concave, this implies that they retain at least a fraction
$(1 - f(\epsc, \deltac))/(1 + f(\epsc, \deltac))$ 
of the utility they obtained under $x$. Since $x$
provided a $g(\epsc, \deltac)$-optimal allocation, we have that $x^1$ provides a
$g'(\epsc, \deltac)$-optimal allocation.
\end{proof}

\paragraph{\bf Step 5.2: Add money back to the goods evenly.}

In Step 5.1, each buyer may remove money from each of the goods that they buy
under $x$, meaning that their budget might not be fully spent under
$x^1$. In Step 5.2 we address this by instructing each of the buyers to spend
that money, and we allocate money so that the demand of each good is increased
by the same amount. 

Formally, let $S_i = e_i -  \sum_{j} x^1_{i,j}$ be the surplus budget of
buyer $i$ under the allocation $x^1$, and let $P = \sum_j p_j$ be the
sum of the prices of all goods. 
For each buyer $i$ and good $j$ we construct the allocation $x^2$ so that 
$x^2_{i,j} = x^1_{i,j} + S_i / P$.

The following lemma states that there exists some value $D \in 
[1 - f(\epsc, \deltac), 1 + f(\epsc, \deltac)]$ such that all goods have demand
exactly $D$, and that all buyers meet their budget constraints with equality and continue to
receive $g'(\epsc, \deltac)$-optimal bundles.

\begin{lemma}
\label{lem:step52}
For the allocation $x^2$ and the price vector $p$ we have the following.
\begin{itemize}
\item There exists a constant $D$ in the range $[1 - f(\epsc, \deltac), 1 + f(\epsc,
\deltac)]$ such that all goods $j$ satisfy $\sum_{i} x^2_{i,j} = D$.
\item All buyers receive $g'(\epsc, \deltac)$-optimal bundles, and exactly meet
their budget constraints.
\end{itemize}
\end{lemma}
\begin{proof}
We start with the first claim. We fix $D = 1 - f(\epsc, \deltac) + \sum_{i} S_i
/ P$.
By definition we have that, for each good $j$, the allocation $x^2$ places
$\sum_{i} S_i / P$ extra demand on good $j$ relative to the demand of good $j$
under $x^2$. Therefore, by \cref{lem:step51}, we have that the demand of
all goods is exactly $D$.

To complete the proof of the first claim we must show that 
$D \in [1 - f(\epsc, \deltac), 1 + f(\epsc, \deltac)]$. The fact that $D \ge 
1 - f(\epsc, \deltac)$ is straightforward, because Step 5.1 only removed money
from goods, and thus $S_i \ge 0$ for all $i$. The fact that $D \le 1 + f(\epsc,
\deltac)$ comes from the fact that all of the money spent in Step 5.2 was
obtained by removing money from goods in Step 5.1. Therefore if $\sum_{i} S_i /
P > 2 \cdot f(\epsc, \deltac)$, then there must have existed a good $j$ that
did not $f(\epsc, \deltac)$-clear in allocation $x$ under price vector $p$, but
this would contradict \cref{lem:step4}. 

For the second claim, note that the utility of each buyer can only increase as
they spend money that was previously not spent, so each buyer continues to
receive a $g'(\epsc, \deltac)$-optimal bundle. 
To show that each buyer meets their budget constraint with equality, observe that for each
good $j$, buyer $i$ spends $(S_i / P) \cdot p_j$ extra money to buy their
allocation of good $j$ in $x^2$ compared to the money that they spent in $x^1$.
Thus in total they spend 
$$\sum_{j} (S_i \cdot p_j)/P = S_i \cdot \sum_{j} p_j/P = S_i$$
additional money, meaning that they exactly meet their budget constraint under
$x^2$.
\end{proof}

\paragraph{\bf Step 5.3: Scale the prices to make the goods clear.}

To finish Step 5, we now create a new price vector $p^3$ so that $p^3_j = p_j
\cdot D$ for all goods $j$, where $D$ is the value arising from
\cref{lem:step52}. We then create the allocation $x^3$ so as to keep the
total amount of money spent by each buyer on each good the same. Formally, for each buyer
$i$ and good $j$ we set
$$x^3_{i,j} = \frac{x^2_{i,j} \cdot p_j}{p_j \cdot D} = x^{2}_{i,j} / D.$$

We define $$g''(\epsc, \deltac) = \frac{1- f(\epsc, \deltac)}{1 + f(\epsc,
\deltac)} \cdot g'(\epsc, \deltac).$$ The following lemma states that all goods
exactly clear, and that each buyer receives a $g''(\epsc, \deltac)$ optimal
bundle. 

\begin{lemma}
For the allocation $x^3$ and the price vector $p^3$ we have the following.
\begin{itemize}
\item Every good clears.
\item Every buyer receives a $g''(\epsc, \deltac)$-optimal allocation and exactly spends
their budget.
\end{itemize}
\end{lemma}
\begin{proof}
For the first claim, we have that the total demand on each good $j$ is 
\begin{align*}
\sum_{i} x^3_{i,j} &= \sum_{i} x^{2}_{i,j} / D \\
&=D/D \\
&=1,
\end{align*}
where in the second equality we used \cref{lem:step52}.

For the second claim, first note that each buyer exactly spends their budget
because they do so in $x^2$ under $p$ by \cref{lem:step52}, and because
$x^3$ is constructed to ensure that they spend the same amount of money in
$x^3$ under $p^3$ as they do in $x^2$ under $p$.

We now consider the utility obtained by each buyer in $x^3$ under $p^3$. We use $u_i(x, p)$ to denote the utility obtained by
buyer $i$ under the allocation $x$ and price vector $p$, and
we start by showing show that $u_i(x^3, p^3) \ge \frac{1}{1 + f(\epsc,
\deltac)} \cdot u_i(x^2, p)$. First note that if $D <
1$ then we have that $u_i(x^3,p^3) \ge u_i(x^2, p)$, because all of the
utility-function segments bought by buyer $i$ become cheaper, meaning that
buyer $i$'s utility cannot decrease. On the other hand, if $D \ge 1$, then we
use the following analysis. 
\begin{itemize}
\item First consider the hypothetical allocation $x'$ in which buyer $i$ spends
exactly as much money on each utility-function segment as they do in $x^2$
under price vector $p$, but under $p^3$ instead. This is not implementable in
practice, since it disobeys the rule that buyer $i$ must fully buy utility
function segment $x_{i,j,k}$ before buying $x_{i,j,k+1}$, but it will serve as
a useful intermediary. Clearly we have $u_i(x', p^3) = (1/D) \cdot u_i(x^2,
p)$, because each segment bought by buyer $i$ became more expensive by a factor
of $D$.

\item Next consider the allocation $x^3$ itself, and note that we can obtain
$x^3$ from $x'$ be transferring money from segments $x_{i,j,k}$ to $x_{i,j,\ell}$
with $\ell < k$ until the segments for each good are bought in order. Since the
utility function of each good is concave, this process involves moving money
from segments with lower utility to segments with (weakly) higher utility, and
so we immediately get that $u_i(x^3, p^3) \ge u_i(x', p^3)$. 
\end{itemize}
Finally note that $D \le 1 + f(\epsc, \deltac)$ by \cref{lem:step52}, so
we have shown that
$u_i(x^3, p^3) \ge \frac{1}{1 + f(\epsc, \deltac)} \cdot u_i(x^2, p)$.

Now let $y^{2}$ be an optimal allocation for buyer $i$ under price vector
$p$, and let $y^{3}$ be an optimal allocation for buyer $i$ under price
vector $p^3$. 
We will show that $u_i(y^3, p^3) \le \frac{1}{1 - f(\epsc, \deltac)} \cdot u_i(y^2, p)$.
If $D > 1$ then we have $u_i(y^3, p^3) \le u_i(y^2, p)$ because all of
the goods became more expensive, and this cannot have
increased the optimal utility. On the other
hand, when $D \le 1$ we use the following analysis. 
\begin{itemize}
\item Consider the hypothetical allocation $y'$ that spends the same amount of
money as $y^2$ under the price vector $p$, but under the price vector $p^3$. This is not implementable in
practice since it may buy more than the length of each utility-function
segment, but it will serve as a useful intermediary. Note that $u_i(y^2,
p) = D \cdot u_i(y', p^3)$, since the price of each good has decreased by a
factor of $D$.

\item Next note that
the prices of all goods are increased by the same
factor when we move from $p$ to $p^3$, thus the relative bang-per-bucks of the
utility-function segments stay the same. This implies that an optimal
allocation for price vector $p^3$ should buy the same segments as $y^2$, and
then spend any excess money on worse segments. Hence we have $u_i(y', p^3) \ge
u_i(y^3, p^3)$.
\end{itemize}
Putting the above two steps together yields that $u_i(y^2, p) \ge D \cdot
u_i(y^3, p^3)$. By \cref{lem:step52} we have that $D \ge 1 - f(\epsc,
\deltac)$, so we have shown that
$u_i(y^2, p) \ge (1 - f(\epsc, \deltac)) \cdot u_i(y^3, p^3)$.

Putting all of the above together yields the following. 
\begin{align*}
u_i(x^3, p^3) &\ge \frac{1}{1 + f(\epsc, \deltac)} \cdot u_i(x^2, p) \\ 
&\ge \frac{1}{1 + f(\epsc, \deltac)} \cdot g'(\epsc, \deltac) \cdot u_i(y^2, p)
\\
&\ge \frac{1- f(\epsc, \deltac)}{1 + f(\epsc, \deltac)} \cdot g'(\epsc,
\deltac) \cdot u_i(y^3, p^3),
\end{align*}
where we used \cref{lem:step52} in the second inequality.
Thus each buyer receives a $g''(\epsc, \deltac)$-optimal bundle under $x^3$ and
$p^3$. 
\end{proof}

\subsection{\texorpdfstring{Completing the proof of \cref{thm:markettopcp}}{Completing the proof of Theorem~\ref*{thm:markettopcp}}}

All that remains is to choose suitable constant values for $\epsc$ and
$\deltac$ to ensure that we produce a $(0, \delta)$-market equilibrium of the
Fisher market. \cref{lem:step52} states that all goods clear, and that each
buyer receives a $g''(\epsc, \deltac)$-optimal bundle, so we must ensure that  
$g''(\epsc, \deltac) \ge 1 - \delta$. 

Let us recall the functions that are involved in the definition of $g''$. 
\begin{align*}
f(\epsc, \deltac) &= c \cdot \epsc + c \cdot (\deltac/\epsc) \\
g(\epsc, \deltac) &= (1 - c \cdot (\deltac/\epsc)) \cdot (1 - c \cdot \epsc) \\
g'(\epsc, \deltac) &= \frac{1- f(\epsc, \deltac)}{1 + f(\epsc, \deltac)} \cdot
g(\epsc, \deltac) \\
g''(\epsc, \deltac) &= \frac{1- f(\epsc, \deltac)}{1 + f(\epsc, \deltac)} \cdot g'(\epsc, \deltac)
\end{align*}
To analyze this, suppose that we set $\epsc = z$ and $\deltac = z^2$ for some
parameter $z$. Then note that $f$ tends to $0$ as $z \rightarrow 0$, that $g$
tends to $1$ as $z \rightarrow 0$, and $g'$ and $g''$ tend to $1$ as $z
\rightarrow 0$. Moreover, the rate at which these functions approach their
limiting value depends only on the constant $c$. 
Many of our proofs required us to have the ability to set $\epsc$
and $\deltac$ to be smaller than some constant $c'$, which was derived from the
constants used in the market. 
Since $\delta$ is constant, we can choose some suitably small constant value for $z$ to ensure that
$\epsc, \deltac \le c'$ and that 
$g''(\epsc, \deltac) \ge 1 - \delta$ as required.

\newpage

\appendix

\section*{Appendix}

\section{Equivalence of the Different Versions of the \pcppad Conjecture}
\label{sec:app:pcppad-equivalence}

In this section, we prove the following statement.

\pcpequivalence*

\begin{proof}
Clearly, formulation 1 implies formulation 2. In \cref{lem:pcircuit-to-gcircuit} we show that formulation 2 implies formulation 3, and in \cref{lem:gcircuit-to-pcircuit1} we show that formulation 3 implies formulation 1. Thus, this establishes that formulations 1 to 3 are equivalent.

Finally, in \cref{lem:gcircuitp-to-gcircuit} we show that formulation 4 implies formulation 3. Since formulation 3 trivially implies formulation 4, this completes the proof of the theorem.
\end{proof}

\subsection{The \pcircuit Version Implies the \gcircuit Version}

\begin{lemma}\label{lem:pcircuit-to-gcircuit}
    For any constants $\eps \in (0, 1/10)$, $\del \in (0, 1/5)$, $5 \del$-\pcircuit can be reduced in polynomial time to \edgcircuit.
\end{lemma}

\begin{proof}
    This is straightforward by using the proof of Theorem 4.1 from \citep{DFHM24}. That proof shows that for any constant $\eps < 1/10$, one can reduce \pcircuit to $\eps$-\gcircuit in polynomial time, by simulating each of the \NAND and \PURE gates of \pcircuit with a gadget that uses at most five $\eps$-\gcircuit gates. This means that the resulting $\eps$-\gcircuit instance has $|T| \leq 5 \cdot |C|$ gates. Now if we had a solution of the \edgcircuit for some constants $\eps \in (0, 1/10)$ and $\delta \in (0, 1/5)$, then at most $\delta \cdot |T| \leq 5 \cdot \delta \cdot |C|$ gates would not be satisfied, which allows at most a $5 \cdot \delta \cdot |C| / |C| = 5 \cdot \delta $ fraction of the gadgets simulating \pcircuit gates to not work as intended. Therefore, translating this solution back to a \pcircuit assignment, we get a $5\delta$-\pcircuit solution.
\end{proof}

\subsection{The \gcircuit Version Implies the \pcircuit Version}
\label{sec:gcircuit-implies-pcircuit}

In this section, we prove the following lemma, which shows that the \gcircuit formulation of the \pcppad conjecture implies the \pcircuit formulation.

\begin{lemma}\label{lem:gcircuit-to-pcircuit1}
    For given constants $\eps \in (0, 1]$ and $\del \in (0,1]$, let $M := \ceil{ \max (7/\eps, (9 / \del)^{1/3})}$ and $\kappa := 320 M^4$. The  \edgcircuit problem with unbounded fan-out can be reduced in polynomial time to $\del/\kappa$-\pcircuit, where every node is the input to exactly one gate.
\end{lemma}

We will proceed by first showing a polynomial-time reduction from \edgcircuit to $\del'$-\pcircuit for $\del = \del' = 0$, and then adapting it for constant $\del, \del' > 0$. In particular, given an $\eps \in (0, 1]$ and an \eps-\gcircuit instance, if we construct in polynomial time a simulation of each \eps-\gcircuit gate via a gadget that uses only at most a constant number $\kappa$ of \pcircuit gates, then this implies that the reduction also holds for any $\del \in [0,1]$, $\del' = \del/\kappa$.

We start from an \eps-\gcircuit instance $(V_G, T)$ with $\eps \in (0, 1]$ and will construct a \pcircuit instance $(V_P, C)$. For each node $u \in V_G$ we fix an $M \geq 7/\eps$ and create $M$ \pcircuit nodes $u_1, \dots, u_M$, called \emph{unary vector}, and we will think of them as encoding in unary the value of $u$. Given an assignment $\valtwoonly \in \{ 0, \garbo, 1 \}^{|V_P|}$ of \pcircuit, for any node $u \in V_G$ and $b \in \{ 0, \garbo, 1 \}$, we denote $U_b = |\{ i \in [M] ~:~ \valtwo{u_i} = b \}| $. We extract an assignment $\valonly \in [0,1]^{|V_G|}$ of \eps-\gcircuit by computing for each $u \in V_G$, $\val{u} = U_1 / M$.

\begin{figure*}[t!]
    \begin{center}
	\begin{minipage}{0.22\textwidth}
		\begin{center}
			\begin{tabular}{c||c}
				$u$ & $v$ \\ \hline
				0 & 1 \\
				1 & 0 \\
				$\garbo$ & $\{0, 1, \garbo\}$\\
				\multicolumn{2}{c}{}%this line make this table nicely aligned with the other two tables
			\end{tabular}
			\caption*{\NOT gate}
		\end{center}
	\end{minipage}
    \begin{minipage}{0.38\textwidth}
        \begin{center}
            \begin{tabular}{c|c||c}
                $u$ & $v$ & $w$ \\ \hline
                0 & 0 & 0 \\
                1 & $\{0,1,\garbo\}$ & 1 \\
                $\{0,1,\garbo\}$ & 1 & 1 \\
                \multicolumn{2}{c||}{Else} & $\{0,1,\garbo\}$
            \end{tabular}
            \caption*{\OR gate}
        \end{center}
    \end{minipage}
	\begin{minipage}{0.38\textwidth}
		\begin{center}
			\begin{tabular}{c|c||c}
				$u$ & $v$ & $w$ \\ \hline
				1 & 1 & 1 \\
				0 & $\{0, 1, \garbo\}$ & 0 \\
				$\{0, 1, \garbo\}$ & 0 & 0 \\
				\multicolumn{2}{c||}{Else} & $\{0, 1, \garbo\}$
			\end{tabular}
			\caption*{\AND gate}
		\end{center}
	\end{minipage}
    \end{center}
	\caption{The truth tables of additional \pcircuit gates.}
	\label{fig:more_gates}
\end{figure*}

For ease of presentation, our reduction will first construct a \pcircuit instance with gates from $\{ \NOT, \OR, \AND, \PURE \}$ (see \cref{fig:more_gates} for the definition of the first three), rather than from $\{ \NAND, \PURE \}$ of the definition in \cref{sec:pcircuit_def}. An instance with the former set of gates can be easily transformed into an instance that uses the latter set of gates. In particular, we first simulate \AND with \NOT and \NAND gates by replacing \AND with \NAND and connecting its output to a \NOT gate. Then we simulate \OR with \NOT and \NAND gates by connecting each of its input nodes with a \NOT gate, whose output is an input to a \NAND gate. Finally, we simulate the \NOT gates by \NAND and \PURE by using a \PURE gate whose outputs are inputs to a \NAND gate. One can easily verify that the aforementioned gate simulations follow the truth tables of \cref{fig:more_gates}. Also notice that each gate from $\{ \NOT, \OR, \AND, \PURE \}$ can be simulated using at most 5 gates from $\{ \NAND, \PURE \}$. 

We will first ensure that in any satisfying \pcircuit assignment $\valtwoonly$, every output unary vector of our gadget will have two properties: (i) it will contain at most one \garbo value, (ii) it will be \emph{sorted}, that is, all its $0$'s are to the leftmost places of the vector and all $1$'s are to the rightmost places of it. We will also enforce the fixed constant parameters of the \eps-\gcircuit gates to be carried in the \pcircuit instance as sorted vectors too; for $c \in [0, 1]$, the \pcircuit gadgets of $G_c$ and $G_{\times c}$ will involve a hardcoded sorted vector $c_1, \dots, c_M$ which contains only pure bits (i.e., no \garbo values), where for the cardinality of its $1$-bits, $C_1$, we have $c \cdot M - 1 \leq C_1 \leq c \cdot M$ by construction. Since every input of a gadget is an output of another, the above also ensure that, in a satisfying assignment, all the inputs of our gadgets -- as well as their parameters -- are sorted unary vectors.

To ensure the required properties in our outputs, whenever needed in our gadgets, we will construct what we call a \emph{formatting} sub-gadget, which takes an input vector $r = (r_1, \dots, r_M)$ through a series of purification-sorting-selection stages, and outputs a vector $w = (w_1, \dots, w_M)$ with the required properties; this is a technique that was introduced in the \ppad/-hardness proof of \pcircuit in \citep{DFHM24}.

\paragraph{Formatting sub-gadget.} Let the input of the sub-gadget be the unary vector $r_1, \dots, r_M$, for which $R_{\garbo} \leq M$. In what we call \emph{purification stage}, we create a binary tree of \PURE gates with $M$ leaves $u_{i,j}$, $j \in [M]$, rooted in each $r_i, i \in [M]$. The crucial observation here is that, by \PURE's definition, if $u_i$ is a pure bit then all $u_{i,j}$'s are copies of it, but even if $u_i$ is \garbo, at most one of $u_{i,j}$'s is \garbo. Consider the vector that comes from concatenating the copies of $u_i$, i.e., $t = (t_1, \dots, t_{M^2}) := ((u_{1,j})_{j \in [M]}; \dots ; (u_{M,j})_{j \in [M]})$. In the \emph{sorting stage}, we create a sorting network \citep{Knuth98-book-vol3-sorting} where the implementation of $\min$ and $\max$ is done using \AND and \OR gates. As discussed in \citep{DFHM24}, since our \pcircuit gates are robust (e.g., if one of the two inputs of \AND is 0, then its output is 0 regardless of the other input; and if one of the two inputs of \OR is 1, then its output is 1 regardless of the other input), the pure bits are correctly ordered in the output of the sorting network, $s = (s_1, \dots, s_{M^2})$. Recall that all \garbo values are concentrated together to the right of all $0$'s and to the left of all $1$'s, and also $S_{\garbo} \leq M$, since all the pure bits of vector $t$ have been preserved. Finally, in the \emph{selection stage}, we select node $s_{i \cdot M}$ for $i \in [M]$ to be our output node (which we rename to) $w_i$. The crucial observation here is that, since all non-pure bits are concentrated in the ``middle'' of $s$ and are at most $M$, the selection will pick up at most one of them, therefore $W_\garbo \leq 1$. Notice also that $w$ is sorted. The formatting sub-gadget uses at most $2M \cdot M$ gates for the purification, and at most $(M^2)^2 / 2$ gates for the sorting (e.g., assuming a Bubble Sort network), totalling to at most $M^4 / 2 + 2 M^2$ gates. 

The following claim will be useful in the correctness proofs of our gadgets.
\begin{claim}\label{cl:formatting_bound}
    If $R_{\garbo} \leq k$ for some $k \in [M]$, then $R_1 \leq W_1 \leq R_1 + k$.
\end{claim}

\begin{proof}
    We know that the pure bits of the unary vector $r$ are copied perfectly in the purification stage and are preserved during the sorting stage, therefore, for the unary vector $s$, we have $S_1 \geq R_1 \cdot M$ and $S_0 \geq R_0 \cdot M$. Also, by definition, $M = R_1 + R_0 + R_{\garbo}$, and since $R_{\garbo} \leq k$, we have $M - R_0 \leq R_1 + k$. Therefore, we have $S_1 \leq M^2 - S_0 - S_{\garbo} \leq M^2 - R_0 \cdot M - 0 = M (M - R_0) \leq M(R_1 + k)$. By definition of the selection stage, the $1$-bits that will be selected are at least $\floor{ S_1 / M }$ and at most $\ceil{ S_1 / M }$. From the above bounds of $S_1$, for the unary vector $w$ we get $W_1 \geq \floor{ S_1 / M } \geq R_1$ and $W_1 \leq \ceil{ S_1 / M } \leq R_1 + k$, where the rightmost inequalities hold due to the fact that $R_1$ is an integer.
\end{proof}

Now we are ready to construct the \pcircuit gadgets that simulate the \eps-\gcircuit gates. We describe the constructions and prove their correctness individually.

\paragraph{$G_{+}$ gadget.} Consider the $(G_{+}, u, v, w, \textsf{nil})$ gate with inputs $u, v$ and output $w$. The \pcircuit gadget implements the \OR operation between the unary vector of $u$ and the mirrored unary vector of $v$. Formally, it uses $M$ many \OR gates $(\OR, u_i, v_{M+1-i}, r_i)$ for $i \in [M]$ to produce vector $r$. Focusing on the bits of vector $u$ as they are \OR-ed with the bits of the mirrored $v$, we notice that the last $U_1$ bits produce at least $U_1 - 1$ and at most $U_1$ $1$-bits (since at most one \garbo-bit from $v$ might be \OR-ed with these bits). If $V_1 \leq M-U_1$, the first $M-U_1$ bits of $u$ produce at least $V_1 - 1$ and at most $V_1 + 2$ $1$-bits, since at most one from $u$ or one from each $u$ and $v$ \garbo-bits can contribute to the outcome, respectively. If $V_1 > M-U_1$, the first $M-U_1$ bits of $u$ produce at least $M - V_1 - 1$ and at most $M - V_1$ $1$-bits, since at most one \garbo-bit of $u$ can interfere in that region. Therefore, $R_1 \geq \min (U_1 - 1 + V_1 - 1, U_1 - 1 + M - U_1 -1) = \min (U_1 + V_1 - 2, M - 2)$, and $R_1 \leq \min (U_1 + V_1 + 2, U_1 + M - U_1) = \min (U_1 + V_1 + 2, M)$. Then, $r$ passes through the formatting sub-gadget to produce $w$. Now recall that $U_{\garbo}, V_{\garbo} \leq 1 $, therefore $R_{\garbo} \leq 2$. By \cref{cl:formatting_bound}, we get $R_1 \leq W_1 \leq R_1 + 2$, which implies $\min (U_1 + V_1 - 2, M - 2) \leq W_1 \leq \min (U_1 + V_1 + 4, M)$, and so, $\min (U_1 + V_1, M) - 2 \leq W_1 \leq \min (U_1 + V_1, M) + 4$. Therefore, $\val{w} = W_1 / M \in [ \min (\val{u} + \val{v}, 1) - 2/M, \min (\val{u} + \val{v}, 1) + 4/M ]$, which satisfies the $G_{+}$ constraints since $M \geq 4/\eps$. This gadget uses $M$ \OR gates and the formatting sub-gadget itself, totalling to at most $M^4 / 2 + 2 M^2 + M \leq M^4$ gates, for the given value of $M$.

\paragraph{$G_{-}$ gadget.} Consider the $(G_{-}, u, v, w, \textsf{nil})$ gate with inputs $u, v$ and output $w$. We first invert the unary vector $v$, using $M$ \NOT gates $(\NOT, v_i, v'_i)$ for $i \in [M]$ to produce vector $v'$. Note that $V'_{\garbo} \leq 1$, since $V_{\garbo} \leq 1$ and $V_1 \leq V'_1 \leq V_1 + 1$. Then, we implement the \AND operation between $u$ and $v'$, that is, we create $M$ \AND gates $(\AND, u_i, v'_i, r_i)$ for $i \in [M]$ to produce vector $r$. Focusing on the $u$ vector as it gets \AND-ed with $v'$, its first $M - U_1$ bits produce at most $1$ $1$-bit since $u$ has at most a single \garbo-bit in that region and the rest are $0$-bits. If $U_1 \geq V'_0$, the last $U_1$ bits of $u$ produce at least $U_1 - V'_0 - 1$ and at most $U_1 - V'_0$ $1$-bits, since at most one \garbo-bit can of $v'$ can interfere. If $U_1 < V'_0$, then the last $U_1$ bits of $u$ do not produce any $1$-bits. Therefore, $\max (U_1 - V'_0 - 1, 0) \leq R_1 \leq \max (U_1 - V'_0 + 1, 1)$, implying $\max (U_1 - V_1 - 2, 0) \leq R_1 \leq \max (U_1 - V_1 + 1, 1)$. Then, we use the formatting sub-gadget with input $r$ and output $w$. Recall that $U_{\garbo}, V'_{\garbo} \leq 1$, therefore $R_{\garbo} \leq 2$. By \cref{cl:formatting_bound}, we get $R_1 \leq W_1 \leq R_1 + 2$, and so, $\max (U_1 - V_1 - 2, 0) \leq W_1 \leq \max (U_1 - V_1 + 1, 1) + 2$, which gives $\max (U_1 - V_1, 0) - 2 \leq W_1 \leq \max (U_1 - V_1, 0) + 3$. Thus, $\val{w} = W_1 / M \in [ \max (\val{u} - \val{v}, 0) - 2/M, \max (\val{u} - \val{v}, 0) + 3/M ]$, which satisfies the $G_{-}$ constraints since $M \geq 3/\eps$. This gadget uses $M$ \NOT, $M$ \AND gates, and the formatting sub-gadget, totalling to at most $M^4 / 2 + 2 M^2 + 2 M \leq M^4$ gates, for the given value of $M$.

\paragraph{$G_{\times c}$ gadget.} For a fixed value $c \in [0,1]$, consider the $(G_{\times c}, u, \textsf{nil}, w, c)$ gate with input $u$, and output $w$. We have a hardcoded unary vector $c_1, \dots, c_M$ with $C_{\garbo} = 0$, such that $c \cdot M - 1 \leq C_1 \leq c \cdot M$ holds. We now create $M^2$ \AND gates implementing the outer product of vectors $c$ and $u$. In other words, we create gates $(\AND, u_i, c_j, r_{i,j})$ for $i,j \in [M]$ resulting to vector $r$ with $M^2$ bits. Notice that, $U_{\garbo} \leq 1$, therefore $R_{\garbo} \leq M$. By definition of the implemented operation, we have $U_1 \cdot C_1 \leq R_1 \leq \min (U_1 \cdot C_1 + C_1 , M^2)$. We then implement a sorting network (e.g. Bubble Sort, as in the formatting sub-gadget) that outputs vector $s = (s_1, \dots, s_{M^2})$. Notice that the pure bits have been preserved, therefore $S_{\garbo} \leq M$. We have $R_1 \leq S_1 \leq \min ( R_1 + M, M^2 )$ Then, we select (rename) the $i \cdot M$-th bit to be the output $w_{i}$, $i \in [M]$. Having all the \garbo bits concentrated in the middle of $s$ and selecting bits with step $M$ ensures that at most one \garbo bit finds its way into $w$, which is also sorted as required. In the selection stage we get at least $\floor{S_1 / M}$ and at most $\ceil{S_1 / M}$ $1$-bits. Since $C_1 / M \leq c \leq (C_1 + 1)/M$, we have
\begin{align*}
    W_1 &\geq \floor{ S_1 / M} \geq \floor{ R_1 / M} \geq \floor{ U_1 \cdot C_1 / M} \geq \floor{ U_1 \cdot \max( c - 1/M, 0 )} \geq \floor{ c \cdot U_1 - U_1 / M } \\ 
    &\geq c \cdot U_1 - 2 , \quad \text{and} \\
    W_1 &\leq \ceil{ S_1 / M} \leq \ceil{ \min (R_1 + M, M^2) / M} \leq \ceil{ \min( U_1 \cdot C_1 + C_1 + M, M^2) / M} \\ 
    &\leq \ceil{ \min( c \cdot U_1 + c + 1, M )} \leq \min ( c \cdot U_1 + 3, M) \\
    &\leq \min( c \cdot U_1, M) + 3  
\end{align*}
Therefore, $\val{w} = W_1 / M \in [ \min (c \cdot \val{u}, 1) - 2 / M, \min( c \cdot \val{u}, 1) + 3 / M ]$, which satisfies the $G_{\times c}$ constraints since $M \geq 3/\eps$. This gadget uses $M^2$ \AND gates, and $(M^2)^2$ gates for the sorting, totalling to at most $M^4 + M^2 \leq 2 M^4$ gates.

\paragraph{$G_{<}$ gadget.} Consider the $(G_{<}, u, v, w, \textsf{nil})$ gate with inputs $u, v$, and output $w$. We first use the $G_{-}$ gadget to implement the subtraction $v - u$, outputting the result to vector $r = (r_1, \dots, r_M)$, which is sorted and has $R_{\garbo} \leq 1$. As shown for that gadget, we get $\max (V_1 - U_1, 0) - 2 \leq R_1 \leq \max (V_1 - U_1, 0) + 3$. Now let us focus on node $r_{M-4}$ which occupies the $5$-th position from the right in $r$. Notice that if $V_1 - U_1 \geq 7$, then $R_1 \geq 5$, and since $r$ is sorted, node $r_{M-4}$ is a $1$-bit. On the other hand, if $V_1 - U_1 \leq 0$, $R_1 \leq 3$, and since there is at most one \garbo-bit in $r$, node $r_{M-4}$ is a $0$-bit. We create a binary tree of \PURE gates rooted in $r_{M-4}$, so that it has $M$ leaves constituting unary vector $t$. We know that if $V_1 - U_1 \geq 7$ then $t$ is the all-$1$'s vector, if $V_1 - U_1 \leq 0$ then $t$ is the all-$0$'s vector, otherwise, $t$ is an arbitrary unary vector, which however maintains $T_{\garbo} \leq 1$ (even if $r_{M-4}$ was a \garbo-bit, only at most one of its copies would be a \garbo-bit, by definition of \PURE). We then sort $t$ using a sorting network with output $w$. We have that $W_1 = M$ if $U_1 \leq V_1 - 7$, and $W_1 = 0$ if $U_1 \geq U_1$. Therefore, $\val{w} = W_1 / M = 1$ if $\val{u} < \val{v} - 7/M$, and $\val{w} = 0$ if $\val{u} > \val{v}$, which satisfies the $G_{<}$ constraints since $M \geq 7/\eps$. For this construction, we used the $G_{-}$ gadget, a tree of \PURE gates with $M$ leaves, and a sorting network. Therefore, the number of gates used overall is at most $M^4 + 2M + M^2 / 2 \leq 2 M^4$.

\paragraph{$G_c$ gadget.} For a fixed value $c \in [0,1]$, consider the $(G_c, \textsf{nil}, \textsf{nil}, w, c)$ gate with no input, and output $w$. Again, we have a hardcoded unary vector $w_1, \dots, w_M$ for which $W_{\garbo} = 0$ and $c \cdot M - 1 \leq W_1 \leq c \cdot M$. Therefore, $\val{w} = W_1 / M \in [c - 1 / M , c]$, which satisfies the $G_c$ constraints since $M \geq 1 / \eps$.

\paragraph{$G_{\lor}$ gadget.} Consider the $(G_{\lor}, u, v, w, \textsf{nil})$ gate with inputs $u, v$, and output $w$. By simply using $M$ \OR gates $(\OR, u_i, v_i, r_i)$ we ensure that its output vector $r$ has $\max ( U_1 , V_1 ) \leq R_1 \leq \max ( U_1 , V_1 ) + 2$, by the robustness of the \OR gate and the fact that each of the $u, v$ vectors has at most one \garbo-bit. For the same reason, $R_{\garbo} \leq 2$. Then, we use the formatting sub-gadget with input $r$ and output vector $w$. Using \cref{cl:formatting_bound}, we get that $R_1 \leq W_1 \leq R_1 + 2$, which implies $\max ( U_1 , V_1 ) \leq W_1 \leq \max ( U_1 , V_1 ) + 4$. Therefore, $\val{w} = W_1 / M \in [ \max ( \val{u}, \val{v} ), \max ( \val{u}, \val{v} ) + 4/M]$. This means that $\val{w} \geq 1 - 6/M$ if $\val{u} \geq 1 - 2/M $ or $\val{v} \geq 1 - 2/M $, while $\val{w} \leq 6/M$ if $\val{u} \leq 2/M $ and $\val{v} \leq 2/M $, which satisfies the $G_{\lor}$ constraints since $M \geq 6/\eps$. This gadget used at most $M + M^4 / 2 + 2 M^2 \leq M^4$ gates.

\paragraph{$G_{\land}$ gadget.} Consider the $(G_{\land}, u, v, w, \textsf{nil})$ gate with inputs $u, v$, and output $w$. Similarly to what we did for the case of $G_{\lor}$, we will use $M$ \AND gates $(\AND, u_i, v_i, r_i)$, its output vector $r$ has $\min ( U_1 , V_1 ) \leq R_1 \leq \min ( U_1 , V_1 ) + 2$, by the robustness of the \AND gate and the fact that each of the $u, v$ vectors has at most one \garbo-bit. We also have $R_{\garbo} \leq 2$. We then use the formatting sub-gadget with input $r$ and output vector $w$, and by \cref{cl:formatting_bound}, we have $R_1 \leq W_1 \leq R_1 + 2$. This means that $\min ( U_1 , V_1 ) \leq W_1 \leq \max ( U_1 , V_1 ) + 4$, therefore, $\val{w} = W_1 / M \in [ \min ( \val{u}, \val{v} ), \min ( \val{u}, \val{v} ) + 4/M]$. So, $\val{w} \geq 1 - 6/M$ if $\val{u} \geq 1 - 2/M $ and $\val{v} \geq 1 - 2/M $, while $\val{w} \leq 6/M$ if $\val{u} \leq 2/M $ or $\val{v} \leq 2/M $, which satisfies the $G_{\land}$ constraints since $M \geq 6/\eps$. Similarly to the  $G_{\lor}$ gadget, for this construction we used at most $M^4$ gates.

\paragraph{$G_{\lnot}$ gadget.} Consider the $(G_{\lnot}, u, \textsf{nil}, w, \textsf{nil})$ gate with input $u$, and output $w$. We simply implement the \NOT operation and mirror the output so that the vector becomes sorted again. In particular, we use $M$ gates $(\NOT, u_i, w_{M+1-i})$ which output vector $w$. Notice that by definition of \NOT and the fact that $U_{\garbo} \leq 1$, we have $W_{\garbo} \leq 1$. Since in $w$ the $0$'s and $1$'s are at the leftmost and rightmost positions, respectively, and there is at most one \garbo-bit, $w$ is sorted. Now notice that $W_1 \geq U_0 \geq M - U_1 - 1$, where the first inequality comes by the definition of \NOT, and the second inequality by the fact that at most one \garbo-bit interferes. Also, for the same reasons, $W_1 \leq U_0 + 1 \leq M - U_1 + 1$. Therefore, $\val{w} = W_1 / M \in [ 1 - \val{u} - 1/M, 1 - \val{u} + 1/M ]$. So, $\val{w} \leq 2/M$ if $\val{u} \geq 1 - 1/M $, while $\val{w} \geq 1 - 2/M$ if $\val{u} \leq 1/M $, which satisfies the $G_{\lnot}$ constraints since $M \geq 2/\eps$. This construction used only $M$ gates.

\paragraph{$G_{=}$ gadget.} Consider the $(G_{=}, u, \textsf{nil}, w, \textsf{nil})$ gate with input $u$, and output $w$. For this, we use two $G_{\lnot}$ gadgets one after the other. In particular, the first one has input $u$ and output $r$, with $ M - U_1 - 1 \leq R_1 \leq M - U_1 + 1$, and the second has input $r$ and output $w$, with $ M - R_1 - 1 \leq W_1 \leq M - R_1 + 1$, giving $ M - U_1 - 2 \leq W_1 \leq M - U_1 + 2$. This implies $\val{w} = W_1 / M \in [ 1 - \val{u} - 2/M, 1 - \val{u} + 2/M ]$. Therefore, $\val{w} \leq 4/M$ if $\val{u} \geq 1 - 2/M $, while $\val{w} \geq 1 - 4/M$ if $\val{u} \leq 2/M $, which satisfies the $G_{=}$ constraints since $M \geq 4/\eps$. For this construction we used twice as many gates used for the $G_{\lnot}$ gadget, i.e., $2M$.

\bigskip

Putting everything together, we can finish the proof of the lemma. From the above constructions, one can easily check that the fan-out of the initial \edgcircuit is preserved in the resulting $\del'$-\pcircuit instance. In \citep{BPR16}, it was shown that \edgcircuit with unbounded fan-out is polynomial time reducible to $(\eps, \del/4)$-\gcircuit with fan-out $\ceil{ 8/\del}$. In our $\del'$-\pcircuit construction, whenever we have a node $u \in V_G$ with fan-out greater than 1, we can simply create the unary vector of $u$ and using it as root, create a binary tree of \PURE gates with $\ceil{8 / \del} \leq 9/\del$ leaves. This tree, which we call \emph{fan-out gadget}, would have to use at most $2 \cdot M \cdot 9/\del  = 18 M/\del$ gates. This gadget ensures that every node of \pcircuit is an input to at most one gate.

Finally, recall that our gadget constructions required at most $2 M^4$ gates from $\{ \NOT, \OR, \AND, \PURE \}$, where $M \geq 7/\eps$. As shown above, the output of some of those gadgets is the input to a fan-out gadget, creating a larger gadget that simulates the original \gcircuit gate. 
The fan-out gadget is considered now to be part of the entire gadget, so each \gcircuit gate is simulated by at most $2 M^4 + 18 M/\del$ \pcircuit gates. We choose $M := \ceil{ \max (7/\eps, (9 / \del)^{1/3})}$, so that $4 M^4 \geq 2 M^4 + 18 M/\del$, and now all our constructions require at most $4 M^4$ gates from $\{ \NOT, \OR, \AND, \PURE \}$. At the beginning of this section, we showed that each of these gates can be simulated with at most 5 gates from $\{ \NAND, \PURE \}$, therefore, the construction requires at most $20 M^4$ gates of the latter set.

Having $\del/(4 \cdot 20 M^4)$ fraction of gates being violated in the resulting \pcircuit, implies having at most $\del$ fraction of gates being violated in the original \edgcircuit instance (with unbounded fan-out).

\bigskip

In the final step, we show how to go from our constructed instance of \pcircuit, where every node is the input to at most one gate, to an instance where every node is the input to \emph{exactly} one gate. As long as there is a node that is not used as an input, we do the following:
\begin{enumerate}
    \item If the node is the output of a \NAND gate, then we simply remove the node and the \NAND gate from the instance.
    \item If the node is the output of a \PURE gate and the other output node of the gate is also not used as an input by any gate, then we remove the \PURE gate and both output nodes.
    \item If the node is the output of a \PURE gate and the other output node is used as an input by some gate, then we remove the first output node and the \PURE gate from the instance, and add a small ``copy'' gadget that takes as input the node that was the input to the \PURE gate and that outputs into its remaining output node. This gadget consists of a \PURE gate, followed by a \NAND gate, followed by a \PURE gate, and finally a \NAND gate. The gadget introduces five new intermediate nodes.
\end{enumerate}
It is easy to see that this procedure terminates with a reduced instance where every node is the input to exactly one gate. In particular, note that case 3 above can only occur for one of the original \PURE gates of the instance, not for any of the new \PURE gates introduced as part of a ``copy'' gadget.

Furthermore, from any assignment to the nodes of this reduced instance that satisfies a $(1-\delta')$-fraction of the gates, we can construct in polynomial time an assignment to the original instance that satisfies at least a $(1-4\delta')$-fraction of the gates. This procedure just traces back the steps of the procedure above. Whenever we removed a gate and its output node(s) from the instance in case 1 or 2 above, we can easily extend the assignment to those nodes so that the gate is satisfied. When a \PURE gate and one of its output nodes were removed in case 3, the assignment can be extended to that output node while satisfying the \PURE gate, as long as all the gates in the ``copy'' gadget are satisfied. As a result, we have reduced our problem to a \pcircuit instance where every node is the input to exactly one gate.

\subsection{The \gcircuitp Version Implies the \gcircuit Version}
\label{sec:gcircuitp-implies-gcircuit}

In this section, we prove the following lemma, which shows that the \gcircuitp formulation of the \pcppad conjecture implies the \gcircuit formulation.

\begin{lemma}
\label{lem:gcircuitp-to-gcircuit}
If $\eps$ and $\delta$ are both constant and $\eps < 1/8$, then we can reduce
\edgcircuitp to $(\eps', \delta')$-\gcircuit where $\eps'$ and $\delta'$ are
both constant fractions of $\eps$ and $\delta$, respectively.
\end{lemma}

We present a multi-step reduction. 
The first step is to ensure that all $G_{\times c}$ gates use constants~$c$
that have a constant bit-length.

\begin{claim}
We can reduce
\edgcircuitp to $(\eps', \delta)$-\gcircuitp where $\eps'$ is
a constant fraction of $\eps$, and where all
$G_{\times c}$ gates use constants $c$ that have a constant bit-length.
\end{claim}
\begin{proof}
We do this by rounding all constants used in $G_{\times c}$ gates to the
nearest multiple of $\eps$. 
Thus, a multiplication of the form $w = c \cdot u$ now becomes $w = c' \cdot
u$, where $| c - c' | \le \eps/2$. Since all variable bounds lie in $[L, U]$,
we have that $u \in [L, U]$, and therefore the total amount of additive error
that we introduce in any multiplication is bounded by $U \cdot \eps/2$. Moreover, we
have not changed the number of gates in the circuit. Thus, if we
have a 
$(\eps/(1+U/2), \delta)$-\gcircuitp solution to the new instance, then we can
recover a \edgcircuitp solution to the original, where we note that $U$ is a
constant. Finally, we note that in the modified instance all $G_{\times c}$
gates now have a constant bit-length, because $\eps$ is a constant. 
\end{proof}

The next step is to entirely remove any $G_{\times c}$ gates.

\begin{claim}
If $c$ has bit-length bounded by a constant in every $G_{\times c}$ gate, then
we can reduce \edgcircuitp to $(\eps', \delta')$-\gcircuitp where $\eps'$ and
$\delta'$ are both constant fractions of $\eps$ and $\delta$, respectively, and
where all $G_{\times c}$ gates use constants $c \in [0, 1]$.
\end{claim}
\begin{proof}
We aim to replace every $G_{\times c}$ gate in which $c$ does not lie in $[0,
1]$ with a gadget that implements the operation. First observe that we can
build a $G_{\times 2}$ gate that multiplies a variable $v$ by two using a $G_{+}$ gate to
add $v$ to itself. Then given a positive integer $p$ we can write $v \cdot p$
as $v \cdot (\sum_{i \in B} 2^i)$ where $B$ is the set of bits that are a 1 in
the binary representation of $p$. 
We can eliminate the multiplication of two variables in this expression by
rewriting it as 
$\sum_{i \in B} (2^i \cdot v)$, and noting that we can compute $2^i \cdot v$
by applying a chain of $G_{\times 2}$ gates of length $i$ to $v$.
Thus we can implement $G_{\times p}$ using at
most $\log^2 p + \log p \le 2 \cdot \log^2 p$ many $G_{\times 2}$ and $G_{+}$
gates. To implement a $G_{\times p/q}$ gate for two positive integers $p$ and
$q$ we first multiply by $p$ and then use a $G_{\times 1/q}$ gate, noting that
$1/q \in [0, 1]$. Finally, to implement a $G_{\times c}$ gate for a negative
rational $c$, we first multiply by $|c|$ and then use a $G_{-}$ gate to
subtract the result from $0$. 

So in summary, we can eliminate each $G_{\times c}$ gate by introducing at most 
$2 \cdot \log^2 |c| + 2 \le 3 \cdot \log^2 |c|$ new gates, where $|c|$ gives the number of bits needed to
represent $c$. Moreover, if each gate introduces an error of $\eps$, then our
simulation of a $G_{\times c}$ gate will introduce an error of at most $\eps
\cdot \log^2 |c|$. 
Recall that there is a constant $b$ that bounds the longest bit-length of
any value $c$ that is used in a $G_{\times c}$ gate. 
So if we find a
solution of the 
$(\eps/(3 \cdot \log^2 b), \delta/(3 \cdot \log^2
b))$-\gcircuitp problem in the new instance, then this will also be a solution to the original $\edgcircuitp$ instance.
Since $b$ is required to be constant, we therefore have that our new $\eps$ and $\delta$ parameters are still constant.
\end{proof}

For the rest of this section we will assume that all $G_{\times c}$ gates in
our \gcircuitp instance have $c \in [0, 1]$. 

The next step is to remove the logical gates $G_{\lor}$, $G_{\land}$, and
$G_{\lnot}$ from the circuit. We do this because, although $\gcircuit$ contains
logical gates, these gates can easily be expressed in terms of the other gates,
and so removing them at this stage makes the following steps of the reduction
less cumbersome. 

\begin{claim}
If $\eps < 1/8$, then we can reduce \edgcircuitp to $(\eps,
\delta')$-\gcircuitp where $\delta'$ is a constant fraction of $\delta$,
and where the gates $G_{\lor}$, $G_{\land}$, and $G_{\lnot}$ are not used.
\end{claim}
\begin{proof}
We claim that a gate $(G_{\lor}, u, v, w, \textsf{nil})$ can be re-written as
\begin{equation}
w = (u > 0.5) + (v > 0.5) > 0.5.
\end{equation}
To see why, observe that since $\eps < 1/8$, if $u \ge 1- \eps$ then $u > 0.5 +
\eps$, so the
$G_{<}$ gate implementing $u > 0.5$ will output a value that is greater than
$1-\eps$, and the same property will hold for $v$. Thus, if either of the two
inputs is greater than $1 - \eps$, the $G_{+}$ gate will output a value that is
at least $1 - 2 \eps > 0.5 + \eps$, and the outer $G_{<}$ gate will output a
value that is at least $1 - \eps$, as required.

On the other hand, if $u \le \eps < 0.5 - \eps$, then 
$G_{<}$ gate implementing $u > 0.5$ will output a value that is less than
$\eps$, and the same property will hold for $v$. The $G_{+}$ gate will then
output a value that is at most $3 \eps < 0.5 - \eps$, so the outer $G_{<}$ gate
will output a value that is at most $\eps$, as required. 

Next we claim that a gate $(G_{\lnot}, u, \textsf{nil}, w, \textsf{nil})$ can
be implemented as
\begin{equation*}
w = 1 - u > 0.5.
\end{equation*}
The reasoning here is the same as for the $G_{\lor}$ gate: the $G_{-}$ gate
does the inversion while introducing an extra $\eps$ error, and the $G_{<}$
gate removes that error because $\eps < 1/8$. 

Finally a $G_{\land}$ gate can be replaced with a $G_{\lor}$ gate and three $G_{\lnot}$ gates
by observing that De Morgan's laws still apply to the formulations of 
$G_{\land}$, $G_{\lor}$, and $G_{\lnot}$ used in \gcircuit.

The simulation of 
$G_{\lor}$ introduces four new gates, and the simulation of $G_{\lnot}$
introduces two new gates. Since a $G_{\land}$ gates uses one $G_{\lor}$ gate
and three $G_{\lnot}$ gates, it introduces 10 new gates. On the other hand, the
reduction has not changed $\eps$ at all. Therefore, we have reduced
\edgcircuitp to 
$(\eps, \delta/10)$-\gcircuitp, which proves the claim.
\end{proof}

For the rest of this section we will assume that our \gcircuitp instance does
not use any logical gates. 
The next step is to move all variable bounds to $[-1, 1]$. 

\begin{claim}
If all variables $v$ have bounds $v_l$ and $v_r$ that are constant, we can
reduce \edgcircuitp to $(\eps', \delta')$-\gcircuitp where $\eps'$ and
$\delta'$ are both constant fractions of $\eps$ and $\delta$, respectively,
and where all variables $v$ are bounded to lie in $[-1, 1]$.
\end{claim}
\begin{proof}
Recall that the constants $U$ and $L$ satisfy $U \ge \max_v{v_u}$ and $L \le
\min_v{v_l}$, and let $m = \max(|U|, |L|)$ be
the largest absolute value of the bounds. The idea is to scale down the
instance by a factor of $1/m$, which will ensure that the variables now fit in
the range $[-1, 1]$. At the same time, we need to reimplement the original
bounds on each variable using $\min$ and $\max$ operations.

To implement this, we do the following.
\begin{itemize}
\item Every constant $c$ that appears in a $G_{c}$, a $G_{\max}$, or a
$G_{\min}$ gate is replaced with $c/m$. Note that the constants in a $G_{\times
c}$ gate are left unchanged.
\item For every gate $v$ in $T$, let $w$ be the output of the gate. We instead
make the gate output to a new variable $w'$, and then we set $w = \min(\max(w',
w_l/m), w_u/m)$, which reimplements the bounds for $w$.
\end{itemize}

It is straightforward to show that if $\mathbf{x}$ is an $\eps/m$ solution to the new
instance, then $m \cdot \mathbf{x}$ is a $\eps$ solution of the original.
Furthermore, we have introduced two new gates for each variable in the original
instance, so every $(\eps/m, \delta/3)$-\gcircuitp solution to the new instance
yields a \edgcircuitp solution to the original. Since $U$ and $L$ are both
absolute constants, we have that $m$ is an absolute constant, so our $\eps$ and
$\delta$ parameters are still constant, as required.
\end{proof}

The next step is to remove the $\min$ and $\max$ gates. 

\begin{claim}
If all variables are constrained by $[-1, 1]$, then we can reduce \edgcircuitp
to $(\eps', \delta')$-\gcircuitp where $\eps'$ and $\delta'$ are both constant
fractions of $\eps$ and $\delta$, respectively, 
and where there are no $G_{\min}$ or $G_{\max}$ gates.
\end{claim}
\begin{proof}
We can replace each gate $(G_{\min}, u, \textsf{nil}, w, c)$ in the following
way. First we introduce an intermediate variable $w'$, and then we set $w' = u
+ (1-c)$ using a $G_+$ gate and a $G_c$ gate, and we set $w = w' - (1-c)$ using
a $G_+$ gate and a $G_c$ gate. The idea is as follows.
\begin{itemize}
\item If $\val{u} \le c$, then
$\val{w'} \le c + (1 - c) + 2\eps = 1 + 2\eps$, where the $2 \eps$ error arise
from the $G_{c}$ and $G_{+}$ gates that are used to implement the operation.
This means that
$\val{w} = \val{u} \pm 3 \eps$. 
\item On the other hand, if $\val{u} \ge c$ then $\val{w'} \ge c + (1 - c) +
2\eps = 1 + 2\eps$, and so the value of $w'$ will be truncated by the $\min$
operation in $G_{+}$. Therefore we have $\val{w} = 1 - (1- c) \pm 3 \eps = c
\pm 3 \eps$. 
\end{itemize}
We can symmetrically implement each $(G_{\max}, u, \textsf{nil}, w, c)$ gate
using an intermediate variable $w' = u
- 1 - c$ and then setting $w = w + 1 + c$, and this time the $\max$ operation in $G_-$
truncates $w'$ to $-1$ whenever $\val{u} \le c$. 

This reduction replaces each $\min$ or $\max$ operation with three gates. So if
we find a $(\eps/3, \delta/3)$-\gcircuit solution to the new instance, then we
can recover a \edgcircuitp solution to the original instance.
\end{proof}

For the rest of this section, we will assume that our \gcircuitp instance does
not use any $\min$ or $\max$ gates. 
At this stage we have eliminated all of the extra features introduced by
\gcircuitp except for the fact that
variables are bounded by $[-1, 1]$ and \gcircuit requires them to be bounded by
$[0, 1]$. Our final claim addresses this, and completes the proof of~\cref{lem:gcircuitp-to-gcircuit}.

\begin{claim}
If all variables are constrained by $[-1, 1]$, then
we can reduce
\edgcircuitp to $(\eps', \delta')$-\gcircuit where $\eps'$ and $\delta'$ are
both constant fractions of $\eps$ and $\delta$, respectively.
\end{claim}
\begin{proof}
The idea is to split each variable $v$ into a positive variable $v^+$ with
bounds $[0, 1]$ and a negative variable $v^-$ with bounds $[0, 1]$. The idea is
that $v^+$ will hold the value of $v$ whenever it is positive, and $v^-$ will
hold $-v$ whenever $v$ is negative, meaning that $v = v^+ - v^-$.
To implement this, we replace each gate with new gates in the following way.
\begin{itemize}
\item For each $(G_c, \textsf{nil}, \textsf{nil}, w, c)$ gate, if $c \ge 0$
then we use $G_c$ gates to set $w^+ = c$ and $w^- = 0$.
On the other hand, if $c < 0$ then
we use $G_c$ gates to set 
$w^+ = 0$ and $w^- = -c$.
\item For each $(G_{\times c}, u, \textsf{nil}, w, c)$ gate we use $G_{\times
c}$ gates to set $w^+ = c \cdot u^+$ and $w^- = c \cdot u^-$.
\item For each 
$(G_{=}, u, \textsf{nil}, w, \textsf{nil})$ gate we use $G_{=}$ gates to set
$w^+ = u^+$ and $w^- = u^-$. 
\item For each $(G_{+}, u, v, w, \textsf{nil})$ gate we use $G_{+}$ and $G_{-}$
gates to set $w^+ = u^+ - u^- + v^+ - v^-$ and $w^- = - u^+ + u^- - v^+ + v^-$.
\item For each $(G_{-}, u, v, w, \textsf{nil})$ gate we use $G_{+}$ and $G_{-}$
gates to set $w^+ = u^+ - u^- - v^+ + v^-$ and $w^- = - u^+ + u^- + v^+ - v^-$.
\item For each $(G_{<}, u, v, w, \textsf{nil})$ gate, we use $G_{c}$, $G_{-}$,
$G_{\times 0.5}$ and $G_{+}$ gates to compute $u' = 0.5 -
u^-/2 + u^+/2$ and $v' = 0.5 - v^-/2 + v^+/2$ and then we use a $G_{<}$ gate to
set $w^+ = u' < v'$ and a $G_c$ gate to set $w^- = 0$.
\end{itemize}
In addition to this, we insert a new type of gate, the $G_{\text{os}}$ gate,
after each other gate in the circuit. That is, if $w^+$ and $w^-$ are the
output variables of some gate, then we instead make that gate output to fresh
intermediate variables $v^+$ and $v^-$, and we implement a 
$G_{\text{os}}$ between the $v$ and $w$ variables in the following way.
\begin{itemize}
\item We set $w^+ = v^+ - v^-$.
\item We set $w^- = v^- - v^+$.
\end{itemize}

In a solution to the circuit we say that a variable $w$ is
\emph{one-sided} if we have that either $w^+ \le c \cdot \eps$ or $w^- \le
c \cdot \eps$ for some constant $c$. The purpose of the $G_{\text{os}}$ gate is to force a value to be
one-sided. Observe that if we apply a $G_{\text{os}}$ gate between $w$ and $v$
the following hold.
\begin{itemize}
\item If $v$ is one-sided, then $w$ will also be one-sided and represent the
same value (up to an additional error introduced by the $G_{\text{os}}$ gate).
For example, if $v^+$ is large and $v^- \le c \cdot \eps$ then the $G_{\text{os}}$
gate is effectively a no-op, and sets $w^+ = v^+ \pm (c+1) \cdot \eps$ and $w^-
\le \eps$, where in the latter case the value will be truncated at $0$ by the
$G_{-}$ gate. 
\item If $v$ is not one-sided, then $w$ will be one-sided, though it will not
necessarily represent the same value. This follows because the two $G_{-}$
gates used to implement $G_{\text{os}}$ compute opposite values, and one of
those values will be truncated at $0$. 
\end{itemize}

Since we use a $G_{os}$ gate after every other gate, this means that any solution of the \gcircuit
problem must input a one-sided value to each other gate in the circuit. Note
that for each other type of gate, if the gate receives a one-sided value as
input, then the gate will output a one-sided value. 
\begin{itemize}
\item For $G_{c}$, $G_{\times c}$, $G_{=}$, and $G_{<}$ this holds by
definition.
\item For $G_{+}$ and $G_{-}$, note that when both inputs are one-sided, the
computation boils down to a single $G_{+}$ or $G_{-}$ gate with $(2+c) \cdot \eps$ extra
error introduced by the terms that are less than $c \cdot \eps$. Moreover note that the plus and
minus outputs of these gates perform exactly the opposite computations, meaning
that one of them will be truncated at zero, and so the output will be one-sided.
\end{itemize}
Therefore every solution to the \gcircuit instance will be one-sided. It is not
difficult to verify that each gate, when given a one-sided input, outputs the
correct value for the output with a constant amount of extra error. Moreover,
each gate from the \gcircuitp instance is implemented by a constant number of
gates in the \gcircuit instance. Thus there exists a constant $c$ such that a
\edgcircuitp solution of the original instance can be recovered from a
$(\eps/c, \delta/c)$-\gcircuit solution of the new instance.
\end{proof}

\section{Inapproximability for Arrow-Debreu Exchange Markets}
\label{sec:app:exchange}

Our hardness results for Fisher markets from \cref{sec:hardness} also imply analogous hardness results for \emph{Arrow-Debreu exchange markets}. Indeed, it is not hard to check that a folklore reduction from Fisher markets to exchange markets (see, e.g., \citep[Theorem 5.1]{DeligkasFHM24-fisher-constant}) preserves the quantities \eps and \del in approximate equilibria.

The only difference between an exchange market and a Fisher market is that, in the former, each buyer $i \in B$ initially owns some endowment $w_{i,j} \geq 0$ of each good $j \in G$, instead of a budget $e_i$. Given a price vector $p \in \mathbb{R}_{\geq 0}^{|G|}$, the set of optimal bundles for buyer $i$, denoted $\opt_i(p) \subseteq \mathbb{R}_{\geq 0}^{|G|}$ is defined as in Fisher markets (see \cref{sec:fisher_markets}), with the only difference that, in the former, $e_i$ is replaced by $\sum_{j \in G} p_j w_{i,j}$ in the budget constraint \eqref{eq:opt-bundle}. In other words, the buyers acquire a budget by selling each endowment $w_{i,j}$ at price $p_{j}$. The definition of market equilibrium is identical to the one for Fisher markets, and its existence is guaranteed if the following condition holds \citep{Maxfield1997,vazirani2011market-plc-ppad}:
\begin{quote}
    \textbf{Sufficient Condition:} The economy graph of the market is strongly connected. This graph is defined on the set of buyers $B$ by introducing a directed edge from buyer $i$ to buyer $i'$ if there exists a good $j \in G$ such that $w_{i,j} > 0$ and $u_{i',j}$ is a strictly increasing function.
\end{quote}
Finally, we define the set of $(1-\delta)$-optimal bundles for buyer $i$ as in Fisher markets, namely, the set of all $(1-\delta)$-optimal solutions, denoted $\opt^\delta_i(p) \subseteq \mathbb{R}_{\geq 0}^{|G|}$.

Applying the folklore reduction (\citep[Theorem 5.1]{DeligkasFHM24-fisher-constant}) to a simple Fisher market, as defined in \cref{def:simple-condition}, yields an exchange market with endowments $w_{i,j} := 1/|B|$ and a strongly connected economy graph. Consider any $(\eps,\delta)$-equilibrium $(p,x)$ of the exchange market. Note that we can assume without loss of generality that the prices sum up to $|B|$, since prices can be scaled without changing the quality of approximation in terms of $\eps$ or $\delta$. As a result, each buyer in the exchange market has budget equal to $1$ after selling its endowment. Since the buyers in the two markets have exactly the same utility functions and budget constraints, it follows that $(p,x)$ is also an $(\eps,\delta)$-equilibrium of the Fisher market. Since the hardness results from \cref{sec:hardness} hold for \emph{simple} Fisher markets, as defined in \cref{def:simple-condition}, the results we obtain here for exchange markets hold for a similar restriction.

\begin{definition}
A family of Arrow-Debreu exchange markets is \emph{simple} if the utility functions are SPLC \emph{linear capped}, all buyers have identical endowments, and there exist constants such that for any market in the family:
\begin{itemize}
    \item Every buyer $i$ has non-zero utility function $u_{ij}$ for at most a constant number of goods $j$.
    \item For every good $j$, the utility function $u_{ij}$ is non-zero for at most a constant number of buyers $i$.
    \item For every buyer $i$, the non-zero slopes in the utility functions $u_{ij}$ have ratio bounded by a constant.
\end{itemize}
Furthermore, for every buyer $i$, there exists at least one good $j$ such that $i$ is not satiated with good $j$ (in other words, the function $u_{ij}$ is linear (uncapped) with strictly positive slope). In particular, this implies that the sufficient condition for the existence of equilibrium is satisfied.
\end{definition}

We therefore obtain the following as corollaries of \cref{thm:fisher-pcp-hard,thm:fisher-inv-poly-hard,thm:fisher-non-zero-hard}.

\begin{theorem}
Fix any constant $\eps < 1/9$. Assuming the \pcppad conjecture, there exists a constant $\delta > 0$ such that finding an $(\eps,\delta)$-approximate market equilibrium in \emph{simple} Arrow-Debreu exchange markets is \ppad/-hard.
\end{theorem}

\begin{theorem}
Fix any constant $\eps < 1/9$. For inverse-polynomial $\delta > 0$ given as part of the input, finding an $(\eps,\delta)$-approximate market equilibrium in \emph{simple} Arrow-Debreu exchange markets is \ppad/-hard.
\end{theorem}

\begin{theorem}
Fix any constant $\eps < 1/9$. Then, there exists a constant $\delta > 0$ such that finding an $(\eps,\delta)$-approximate market equilibrium in \emph{simple} Arrow-Debreu exchange markets is \ppad/-hard, \emph{if buyers are not allowed to spend any money on goods for which their utility function is the zero function}.
\end{theorem}

\bigskip
\subsubsection*{Acknowledgments}

We thank the reviewers for useful comments and suggestions.
Argyrios Deligkas was supported by EPSRC Grant EP/X039862/1 ``NAfANE: New Approaches for Approximate Nash Equilibria''.
John Fearnley was supported by EPSRC Grant EP/W014750/1
``New Techniques for Resolving Boundary Problems in Total Search''.

\bibliographystyle{plainnat}
\bibliography{references}

\end{document}